\documentclass[10pt]{article}
\usepackage{caption}
\usepackage{wrapfig}
\usepackage{fullpage}
\usepackage{float}
\usepackage[font=small]{caption}
\usepackage{longtable}
\usepackage{booktabs}
\usepackage{array}

\usepackage{latexsym}
\usepackage{amssymb}
\usepackage{amsmath}
\usepackage{enumerate}
\usepackage{mathtools}
\usepackage{verbatim}
\usepackage{braket}
\usepackage{xcolor}
\usepackage{url}	
\usepackage{graphicx} 
\usepackage{subcaption}
\usepackage{tikz-cd}

\usepackage{thmtools}

\usepackage[colorlinks=true, linkcolor=blue]{hyperref}

\newtheorem{theorem}{Theorem}
\newtheorem{theoremS}{Theorem}[section]
\newtheorem{lemma}[theoremS]{Lemma}
\newtheorem{definition}[theoremS]{Definition}
\newtheorem{corollary}[theoremS]{Corollary}
\newtheorem{conjecture}[theoremS]{Conjecture}

\newtheorem{remark}[theoremS]{Remark}

\def\finito{{\hspace*{\fill}  \mbox{$\blacksquare $}}}

\newcommand{\tabfigscale}[2]{%
  \includegraphics[scale=#2]{figures/#1}%
}

\newcommand{\tabfigscaletrim}[3][0mm]{%
  \includegraphics[
    scale=#2,
    trim=#1 0 #1 0,
    clip
  ]{figures/#3}%
}

\newcommand{\hfillih}{\hspace{0.3em}}
\newcommand{\hfilli}{\hspace{0.7em}}
\newsavebox{\galbox}
\newsavebox{\galboxii}
\newcommand{\galscale}{0.4}
\newcommand{\galmaxw}{0.92\textwidth}
\newcommand{\galmaxh}{0.8\textheight}
\newcommand{\drawfig}[2]{%
  \sbox{\galbox}{\includegraphics[scale=\galscale]{#1/#2}}%
  \ifdim\wd\galbox>\galmaxw \sbox{\galbox}{\includegraphics[width=\galmaxw]{#1/#2}}\fi
  \ifdim\ht\galbox>\galmaxh \sbox{\galbox}{\includegraphics[height=\galmaxh]{#1/#2}}\fi
  \usebox{\galbox}}
\newcommand{\subfig}[2]{\subcaptionbox{}{\drawfig{#1}{#2}}}
\newcommand{\sfg}[1]{\subfig{figures}{#1}}
\newcommand{\tabfig}[2]{%
  \sbox{\galbox}{\includegraphics[scale=\galscale]{figures/#1}}%
  \ifdim\wd\galbox>#2 \sbox{\galbox}{\includegraphics[width=#2]{figures/#1}}\fi
  \usebox{\galbox}}

\begin{document}

\title{
Hardness of approximation for minimum-weight decoding of two-dimensional topological quantum codes
}

\author{Louay Bazzi
\thanks{Department of Electrical and Computer Engineering, American University of Beirut,  Beirut, Lebanon.}~\thanks{Email: \texttt{lb13@aub.edu.lb}}~
and   Georges Khater\footnotemark[1]~\thanks{Email: \texttt{gak22@mail.aub.edu}}}

\maketitle
\begin{abstract}
Efficient decoding is essential for the practical realization of
fault-tolerant quantum computers. We study the computational complexity of
minimum-weight decoding for topological quantum codes. For surface codes
under the depolarizing channel, we consider Minimum-Weight decoding, which
seeks a minimum-weight Pauli error consistent with both the $X$- and
$Z$-syndromes. For color codes under independent $X$- and $Z$-error models,
we consider Separate Minimum-Weight decoding.

Assuming $P\neq NP$, we establish polynomial additive inapproximability gaps
for these problems.  
Specifically, for the toric code and the $4.8.8$ color
code on the torus,
there exists a constant $c>0$ such that no
polynomial-time algorithm can always produce a solution whose weight
is within $cN^{1/14}$ of the optimum, where $N$ is the number of qubits, unless $P=NP$. 
 For the planar surface code, we obtain an
$\Omega(N^{1/18})$ gap. Our inapproximability results use H{\aa}stad's
hardness of approximation for MAX-3SAT.

Our reduction develops a general, modular framework for embedding logical constraints into coupled primal--dual join problems on a lattice. A key ingredient is a localization argument that controls unintended interactions between different parts of the construction.
\end{abstract}

\tableofcontents

\section{Introduction}\label{intro}

Quantum error correction protects quantum information against noise and is a central ingredient in the development of fault-tolerant quantum computers~\cite{Shor1995,Shor1996,CalderbankShor1996,Steane1996,Gottesman1997}.  
Among the leading candidates for practical fault-tolerant quantum computation are \emph{topological quantum codes}, in which stabilizer checks act locally while logical information is encoded nonlocally \cite{Bombin2013}.
The most prominent family of topological quantum codes is that of  
\emph{surface codes}, with the toric code~\cite{Kitaev1997} as its canonical example. Its planar counterpart, the planar surface code~\cite{BravyiKitaev1998, Dennis2002}, is particularly well suited to physical implementations due to its  high error threshold and compatibility with two-dimensional architectures.
Alongside surface codes,  \emph{two-dimensional color codes}  constitute a  widely studied family of topological quantum codes  and support a broader set of transversal logical gates~\cite{BombinMartinDelgado2006}.

Efficient decoding is essential for the  realization of fault-tolerant
quantum computers. 
We assume perfect syndrome measurements. An   ideal decoding objective is
\emph{Maximum-Likelihood} decoding, which seeks, given an
observed syndrome, the most likely equivalence class of errors under the
assumed noise model.  A simpler decoding objective is to seek a most likely individual error
consistent with the observed syndrome, rather than the most likely
equivalence class. 
Under the depolarizing channel, this amounts to finding a minimum-weight
Pauli error consistent with both the $X$- and $Z$-syndromes. We refer to this
as \emph{Minimum-Weight} decoding.
In the more tractable independent $X$- and
$Z$-error model, the two types of errors can be decoded separately: one seeks
a minimum-weight $X$-error consistent with the $Z$-syndrome and a
minimum-weight $Z$-error consistent with the $X$-syndrome. We refer to this
as \emph{Separate Minimum-Weight} decoding. This paper studies the computational complexity of these two
minimum-weight decoding problems for topological codes.

For surface codes, Separate Minimum-Weight decoding reduces to minimum-weight matching and can
be performed in $O(N^3\log N)$ time~\cite{Higgott2022,HiggottGidney2023},
where $N$ is the number of qubits. 
Recently, this complexity was improved to
$O(N^{3/2}\log N)$~\cite{Bazzi26,Bazzi26C}. 
Beyond these theoretical worst-case bounds, several  efficient
implementations have been developed in practice; see the survey
in~\cite{deMartiOlius2024}.

The situation becomes more subtle for Separate Minimum-Weight decoding of
color codes and Minimum-Weight decoding of surface codes. Efficient suboptimal decoders with good empirical performance have been
developed for Separate Minimum-Weight decoding of color codes~\cite{DColor}
and Minimum-Weight decoding of surface codes~\cite{DTcor,YuanLu2022}.

This raises a natural complexity-theoretic question: to what extent can
Minimum-Weight decoding for surface codes and Separate Minimum-Weight decoding
for color codes be performed efficiently? 

Separate Minimum-Weight decoding for
general stabilizer codes is known to be NP-hard~\cite{HsiehLeGall2011}.
More recently, Fischer and Miyake~\cite{FischerMiyake}
showed that decoding surface codes is NP-hard under non-identically distributed
depolarizing noise, where the error probabilities depend both on the qubit and
on whether the error is of type $X$, $Y$, or $Z$. Their hardness result relies
on carefully constructed, spatially varying error probabilities and therefore
does not imply intractability of Minimum-Weight decoding under the identically
distributed depolarizing channel. Nevertheless, it points to the difficulty
introduced by correlated Pauli errors, since Separate Minimum-Weight decoding
remains tractable under non-identically distributed noise.

Very recently, two works established NP-hardness results for Separate
Minimum-Weight decoding of two-dimensional color codes and Minimum-Weight
decoding of surface codes. In March 2026, Walters and Turner~\cite{WT26} proved that Separate
Minimum-Weight decoding for two-dimensional color codes is NP-hard via a
reduction from 3SAT. Then, also in March 2026 and independently and concurrently with our work, Gu, Wang,
and Kubica~\cite{GWK26} established NP-hardness of Minimum-Weight decoding
for surface codes via a reduction from 3-Dimensional Matching. They also
recovered the NP-hardness result of Walters and Turner for color codes through
a different reduction. These results rule out efficient exact decoding unless
$P=NP$, and both works raised the question of how well minimum-weight
solutions can be efficiently approximated.

Even more recently, this question has been investigated from the algorithmic
side. In June 2026, Walters~\cite{WaltersApprox26} developed a polynomial-time approximation scheme for Separate Minimum-Weight decoding of color codes, while, independently and concurrently, Gu, Wang, and Kubica~\cite{GWKApprox26} developed one for Minimum-Weight decoding of surface codes and, more broadly, topological codes. 
In particular, these results show that, for every fixed $\varepsilon>0$, a
multiplicative $(1+\varepsilon)$-approximation is possible in polynomial time;
that is, one can produce a solution whose weight is at most
$(1+\varepsilon)$ times the optimum. 
This leaves open the question of whether
polynomial-time approximation is possible with additive error sublinear in
$N$, where $N$ is the number of qubits.

In this work, we partially answer this question in the negative. In
particular, assuming $P\neq NP$, we establish polynomial additive
inapproximability gaps for Minimum-Weight decoding of surface codes and
Separate Minimum-Weight decoding of the $4.8.8$ color code. Specifically, our results apply to the toric code, the planar surface code,
and the $4.8.8$ color code on the torus. 
The inapproximability gap is   
 $\Omega(N^{1/14})$ for the toric code and the $4.8.8$ color code
on the torus, and $\Omega(N^{1/18})$ for the planar surface code.  That is, unless $P=NP$, no polynomial-time algorithm can always
produce a solution whose weight is within the corresponding additive gap
of the optimum. 
Our inapproximability results are obtained using H{\aa}stad's hardness of
approximation of MAX-3SAT~\cite{Has01}.

We next compare our results and techniques with the closely related works discussed above.

The reduction of Walters and Turner encodes logical constraints through
geometrically arranged decoding gadgets and provides an important starting
point for the present work. Our work was motivated in part by their
construction, and, at a high level, our reduction follows the same standard
circuit-gadget paradigm for reductions from 3SAT. In its realization as a
decoding problem, the underlying ideas behind our wire, clause, and
garbage-collection gadgets are adapted from those of Walters and Turner.  
While their framework is developed directly for color codes, 
ours is developed for surface codes. 
We first prove that Minimum-Weight decoding 
of surface codes is NP-hard and then strengthen this result to
hardness of approximation. We subsequently transfer these hardness results to Separate Minimum-Weight
decoding of two-dimensional color codes arising from the $4.8.8$ tiling on
the torus, through a reduction from Minimum-Weight decoding of the toric code.

The soundness argument in our setting involves a difficulty that is absent
from the construction of Walters and Turner. 
 Their argument exploits separated syndromes for which
the number of defects provides a  direct  lower bound on the weight of any
correction. Requiring a correction to attain this lower bound strongly
constrains its behavior locally and thereby considerably simplifies the
soundness argument. 
In our setting, no analogous counting argument is available, and different parts of a correction can potentially interact across gadgets.  To overcome this issue, we establish a Localization Lemma,
which rules out such unintended interactions between distinct gadgets and
thereby recovers the local control needed for the reduction. Thus, although
the gadget vocabulary and some of the geometric ideas are inspired by
Walters and Turner, the mechanism underlying the soundness of our
construction is substantially different.   

The  work of Gu, Wang, and Kubica~\cite{GWK26} and the present work
independently establish NP-hardness of surface-code decoding through
different reductions. 
 Whereas their proof proceeds from 3-Dimensional
Matching, ours gives a direct reduction from 3SAT and, in addition,
establishes hardness of approximation. The two approaches thus provide
complementary perspectives on the computational complexity of optimal
surface-code decoding.

Our results show that the correlations introduced by the depolarizing channel
fundamentally alter the computational landscape of surface-code decoding:
while independent error models admit efficient matching-based decoders,
hardness under correlated Pauli errors persists for sufficiently
accurate additive approximations, despite the existence of polynomial-time approximation schemes.  Likewise, for two-dimensional color codes, Separate Minimum-Weight decoding remains computationally hard for sufficiently accurate additive approximations, despite the existence of polynomial-time approximation schemes.

Beyond the
hardness results themselves,  we believe that the gadget and localization framework developed
here provides a reusable toolkit for studying the computational complexity of
decoding problems on topological codes. Its modular structure separates the
encoding of logical constraints into atomic gadgets from the control of
nonlocal interactions through localization, potentially allowing the framework
to be adapted to other codes and complexity reductions.

\subsection{Preliminaries}\label{prel}

We begin with the graph-theoretic formulation of the Minimum-Weight Decoding problem for the toric and planar surface codes. We then introduce its graph-theoretic abstraction, the \emph{Minimum-Weight Join} problem, which will serve as the central object of study throughout the paper. The definition of the decoding problem for color codes is deferred to Section~\ref{color}.

\subsubsection{Minimum-Weight Decoding for the toric and planar surface codes}


{
\renewcommand{\galscale}{0.11}
\begin{figure}[!htbp]
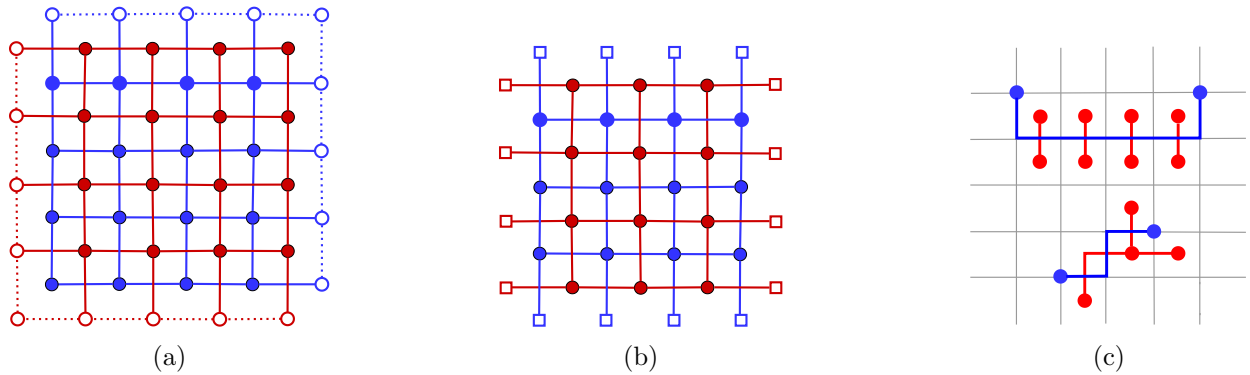
\centering
\sfg{toric.pdf}\hfill
\sfg{planar.pdf}\hfill
{
\renewcommand{\galscale}{0.3}
\sfg{example.pdf}
}
\caption{
(a) \textbf{Primal and dual lattices of the $[[2L^2,2,L]]$ toric code for $L=4$.}
The primal lattice is shown in blue and the dual lattice in red. Vertices and edges identified with those on the opposite side are shown as unfilled circles and dashed lines, respectively.
(b) \textbf{Primal and dual lattices of the $[[2L^2-2L+1,1,L]]$ planar surface code for $L=4$.}
The boundary vertices are shown as squares. 
(c) \textbf{Example of an error.}
The error edges and the corresponding defects are shown on the primal lattice.
}
\label{tpe}
\end{figure}
}

Consider the $[[2L^2,2,L]]$ \emph{toric code}~\cite{Kitaev1997}, whose primal and dual lattices are shown in Figure~\ref{tpe}(a). The \emph{dual} of a primal (dual) edge is the unique dual (primal) edge crossing it. Each qubit is associated with a primal--dual edge pair. An \emph{error} in the toric code is a collection of primal and dual edges. The \emph{weight} of an error is the total number of edges, counting each primal--dual edge pair only once. Equivalently, the weight is the number of distinct qubits associated with the edges of the error.
A primal or dual lattice vertex is called a \emph{defect} with respect to an error if it is incident to an odd number of error edges. The \emph{syndrome} of an error is the set of its primal and dual defects. 
See Figure~\ref{tpe}(c) for an example of an error and its corresponding syndrome. 
In the \emph{Minimum-Weight Decoding} problem for the toric code, we are given a set of primal and dual defects and seek a minimum-weight error whose syndrome is the prescribed set of defects.

The \emph{Minimum-Weight Decoding} problem for the $[[2L^2-2L+1,1,L]]$ \emph{planar surface code}~\cite{BravyiKitaev1998, Dennis2002} is defined analogously. The primal and dual lattices of the planar surface code are shown in Figure~\ref{tpe}(b). The primal and dual vertices shown as squares are called \emph{boundary vertices}. The only difference is that defects are restricted to nonboundary vertices, whereas the parity of the error at each boundary vertex is unconstrained.

\subsubsection{Joins and the Minimum-Weight Join Problem}

Rather than representing errors as sets of edges, it is more convenient for our purposes to represent them as collections of paths. Moreover, it is convenient to restrict attention to acyclic errors. This is not restrictive in the context of approximation algorithms for the decoding problems, since the cycles of an error can be removed in polynomial time while preserving its syndrome and without increasing its weight. Acyclic errors are captured by the notion of a \emph{join}, which we define next.

Before defining joins, we specify the domains in which they are defined. 
Except for Section~\ref{planarS}, we will not be concerned with the boundary vertices of the planar surface code. Indeed, the tools developed in this paper are most naturally formulated on either the infinite square lattice or the torus. Accordingly, throughout the paper we fix a \emph{domain}, consisting of either the infinite square lattice and its dual or the square lattice on the torus, obtained by identifying opposite sides of a $K\times K$ square lattice, and its dual. The former setting is used to model the planar surface code away from the boundary, whereas the latter corresponds to the toric code.

We refer to defects in the primal and dual lattices as \emph{primal defects} and \emph{dual defects}, respectively. 
 Similarly, we call paths in the primal and dual lattices \emph{primal paths} and \emph{dual paths}, respectively. 

A \emph{primal join} with respect to a prescribed set of primal defects of
even cardinality is a collection of edge-disjoint primal paths, each connecting
two defects, whose union forms an acyclic graph and induces odd degree at each
defect. 
\footnote{
Under this definition,  distinct joins may have the same underlying edge set because the paths may intersect at vertices, giving rise to different path decompositions of that edge set. }   The even-cardinality assumption is necessary, since every graph has an even
number of odd-degree vertices. It is also sufficient, since the domain is
connected. 
 A \emph{dual join} is defined analogously with respect to a prescribed set of dual defects of
even cardinality. A \emph{join}, with respect to both the primal and dual defects, consists of a primal join together with a dual join.
The \emph{weight} of a join is the total number of edges in its primal and dual joins, counting each primal--dual edge pair only once.

In the \emph{Minimum-Weight Join} problem, given even-cardinality sets of primal and dual defects in 
the domain, the goal is to find a minimum-weight join for those defects.  
See 
Figure~\ref{tpe}(c) for an example of a 
 minimum-weight join.

\subsection{Summary of results}

We start by  showing that Minimum-Weight Join is NP-hard.

\begin{restatable}{theorem}{NPHTheorem}\label{nphardnessMWJ} 
\textbf{\emph{(NP-hardness of Minimum-Weight Join)}}
The Minimum-Weight Join problem is NP-hard.
\end{restatable} 

The proof is based on a direct reduction from $3$SAT. At a high level, we
represent a Boolean formula by arranging primal and dual defects into local
structures that encode its variables, clauses, and their connections. 
These structures are instances of a small 
 library of gadgets whose minimum-cost
configurations encode the logical constraints of the formula. If the
optimization were restricted to respect the decomposition into gadgets,
satisfiability of the formula would  be directly reflected in the
minimum join weight.  
The main difficulty is that the Minimum-Weight Join problem is inherently
global. A general join may connect defects belonging to different gadgets,
and primal and dual paths associated with different parts of the construction
may interact, potentially bypassing the intended logical constraints.  
We address this difficulty through a localization framework. A separation scale $\Delta$ is
built into the geometry of the gadgets, and by choosing $\Delta$ sufficiently
large, we show that every minimum-weight join 
respects the
intended decomposition into gadgets. 
This establishes the correspondence between the logical behavior of the
gadgets and the global Minimum-Weight Join problem. More broadly, this construction provides a general framework for embedding and composing logical gadgets in coupled primal--dual lattice optimization problems while controlling unintended interactions between them.

We then strengthen the reduction to show  
that Minimum-Weight Join is  hard to approximate.

\begin{restatable}{theorem}{HAMWJ}\label{mainthMWJ} 
\textbf{\emph{(Inapproximability of Minimum-Weight Join)}}
Assume that $P\neq NP$. Then there exists a constant $c>0$ such that no polynomial-time algorithm that, given even-cardinality  sets of 
primal and dual defects contained in a $K\times K$ square or a torus of side length $K$, is guaranteed to return a join for these defects of weight at most
$W^*+cK^{1/7}$,
where $W^*$ denotes the minimum join weight.
\end{restatable}

The NP-hardness construction alone does not provide a sufficiently large gap:
violating the logical behavior of some gadgets may incur only a small
penalty, even when every assignment to the underlying formula leaves
many clauses unsatisfied. To overcome this issue, we construct strong gadgets
for which violating the intended logical behavior incurs a large penalty.
Together with the localization framework, this makes the minimum join weight
quantitatively reflect the minimum number of unsatisfied clauses.
Accordingly, H{\aa}stad's hardness of distinguishing satisfiable $3$SAT
instances from instances in which every assignment leaves a constant fraction
of clauses unsatisfied~\cite{Has01} transfers to the Minimum-Weight Join problem, yielding
Theorem~\ref{mainthMWJ}.

Then we  derive from Theorem~\ref{mainthMWJ}  
inapproximability results for Minimum-Weight  Decoding of the toric and planar surface codes, and for Separate Minimum-Weight Decoding of the $4.8.8$ color code.

For the toric code, the reduction is immediate since minimum-weight decoding is essentially the Minimum-Weight Join problem in the toroidal domain.

 \begin{restatable}{theorem}{HAToric}\label{toricha} 
\textbf{\emph{(Inapproximability of Minimum-Weight Decoding for the toric code)}}
Assume that $P\neq NP$. Consider the $N$-qubit toric code. Then there exists a constant $c>0$ such that no polynomial-time algorithm that, given primal and dual defects, is guaranteed to return an error whose syndrome matches the given defects and whose weight is at most $W^*+cN^{1/14}$,
where $W^*$ denotes the minimum weight of an error with the given syndrome.
\end{restatable} 

For the planar surface code, we work in the square-lattice domain, but additional care is required because decoding allows defects to be matched to boundary vertices of the planar surface-code lattice.

\begin{restatable}{theorem}{HAPlanar}\label{planarha} 
\textbf{\emph{(Inapproximability of Minimum-Weight Decoding for the planar surface code)}}
Assume that $P\neq NP$. Consider the $N$-qubit planar surface code. Then there exists a constant $c>0$ such that no polynomial-time algorithm that, given primal and dual defects, is guaranteed to return an error whose syndrome matches the given defects and whose weight is at most $W^*+cN^{1/18}$,
where $W^*$ denotes the minimum weight of an error with the given syndrome.
\end{restatable}

For color codes, we show that the easier Separate Minimum-Weight
Decoding problem is hard to approximate.  
We establish a weight-scaling reduction from decoding the toric code over the depolarizing channel to decoding $X$ errors in the color code, thereby transferring the toric-code inapproximability result.

\begin{restatable}{theorem}{HAColor}\label{colorha} 
\textbf{\emph{(Inapproximability of Separate Minimum-Weight Decoding for the 4.8.8 2D color code)}}
Assume that $P\neq NP$. Consider the $N$-qubit  
$4.8.8$ color code defined on the torus.
Then there exists a constant $c>0$ such that no polynomial-time algorithm that, given  
$Z$-defects, is guaranteed to return an 
$X$-error whose $Z$-syndrome matches the given defects and whose weight is at most $W^*+cN^{1/14}$,
where $W^*$ denotes the minimum weight of an $X$-error with the given $Z$-syndrome. 
\end{restatable}

\subsection{Paper organization}
We develop the NP-hardness reduction in Section~\ref{nphs} and establish
hardness of approximation in Section~\ref{hardapx}.
In Section~\ref{ccodes}, we derive the consequences of these results for
the codes considered in this paper.
We conclude in Section~\ref{conc} with extensions and open questions.

\section{NP-hardness}\label{nphs}
We show that the Minimum-Weight Join problem is NP-hard by reducing $3$SAT
to the problem of deciding whether a prescribed set of primal and dual defects
admits a join of weight at most a given value.
Recall that in $3$SAT, a \emph{literal} is a Boolean variable or its negation,
a \emph{clause} is a disjunction of literals, and a $3$-CNF formula is a conjunction
of clauses, each containing three literals.
The $3$SAT problem asks whether there exists an assignment of the variables
that satisfies all clauses.

The reduction is built from the atomic gadgets introduced next; an outline
of the complete construction is given in Section~\ref{nphoverview}.

\subsection{Atomic gadget library}\label{glib}

This section describes the atomic gadgets that form the building blocks of our reductions.
 We first specify their logical behavior in terms of satisfying signatures and then describe the common geometric structure shared by all gadgets.

The reduction uses atomic gadgets of different types,  each enclosed in a box, whose symbols are shown in Figure~\ref{types}.

{
\renewcommand{\galscale}{0.2}
\begin{figure}[!htbp]
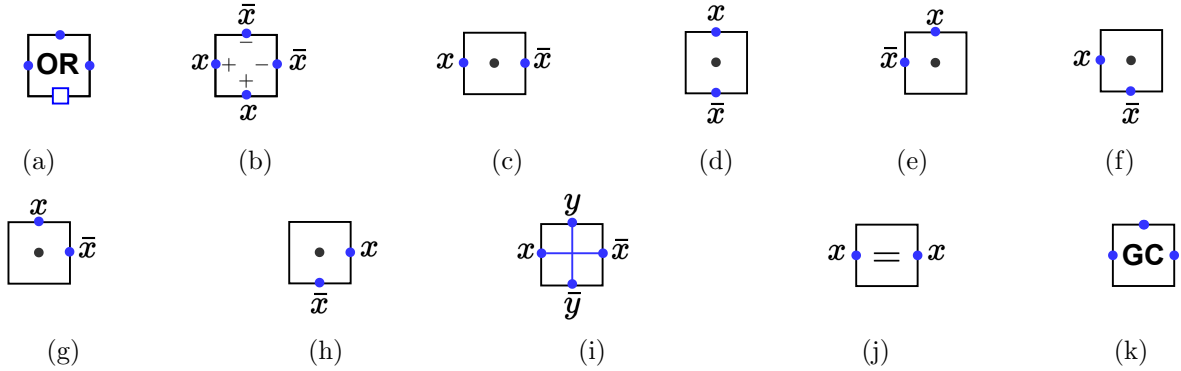
\centering
\sfg{type-clause-labeled.pdf}\hfill
\sfg{type-var-labeled.pdf}\hfill 
\sfg{type-neg1.pdf}\hfill 
\sfg{type-neg2.pdf}\hfill  
\sfg{type-neg3.pdf}\hfill 
\sfg{type-neg4.pdf}\hfill 
\sfg{type-neg5.pdf}\hfill  
\sfg{type-neg6.pdf}\hfill  
\sfg{type-cross-labeled.pdf}\hfill 
\sfg{type-eq-labeled.pdf}\hfill
\sfg{type-gc-labeled.pdf}
\caption{\textbf{Symbols for atomic gadgets.} (a) Clause gadget; (b) Variable gadget; (c)--(h) Wire gadgets; (i) Crossing-wire gadget; (j) Equality gadget; (k) Garbage collection gadget.}
\label{types}
\end{figure}
}

 Each atomic  gadget is constructed from primal and dual defects so as to implement a logical function. The distinguished \emph{primal defects} through which a gadget communicates with other gadgets are called its \emph{outer pins}, shown in blue. All remaining defects are called \emph{internal gadget defects}.
Every gadget contains an even number of internal \emph{dual defects}.   We define the \emph{parity} of a gadget to be the parity of its number of primal defects.

	For each subset $S$ of outer pins whose parity is equal to that of the gadget, an \emph{$S$-join} of the gadget is a join with respect to the set of primal defects excluding the pins in $S$, together with the set of all dual defects of the gadget. Since the number of internal dual defects is even, an $S$-join always exists for every such $S$. 
We adopt the convention that the outer pins in $S$ are assigned the value \emph{True}, whereas those outside $S$ are assigned the value \emph{False}.

 A gadget with $p$ outer pins admits exactly $2^{p-1}$ valid assignments to its outer pins, called its \emph{signatures}
\footnote{The terminology of signatures is inspired by the signature
formalism for planar matchgates~\cite{Val02}, although the notion used here
is different.}.  
 Specifically, a signature assigns an even (respectively, odd) number of outer pins the value \emph{True} whenever the gadget has even (respectively, odd) parity.
 
 The \emph{cost} of a signature $S$ is defined to be the weight of a minimum-weight $S$-join. A signature is called \emph{satisfying} if its cost is minimum among all signatures. Thus, by definition, 
  all satisfying  signatures  have  the same cost, which  is strictly smaller than that of any nonsatisfying signature.
  The gadgets are designed so that their satisfying signatures realize precisely their intended logical behavior.

We use six types of atomic gadgets, whose specifications are given below.

\begin{itemize}
\item[$\bullet$] \textbf{Clause gadget}, 
whose symbol is shown in  Figure~\ref{types}(a). This gadget has odd parity and four outer pins: three \emph{literal pins}, shown as circles, and one \emph{garbage-collection pin}, shown as a small square. Its satisfying signatures have at least one literal pin assigned \emph{True}.   
The garbage-collection pin  ensures  that every assignment of the literal pins extends to a valid signature. Since the gadget has odd parity, the truth value assigned to the garbage-collection pin is the complement of the XOR of the truth values assigned to the literal pins.

\item[$\bullet$] \textbf{Variable gadget}, whose symbol is shown in Figure~\ref{types}(b). This gadget has even parity and  four outer pins.  
It represents a Boolean variable $x$ by duplicating the variable and its negation. In a  satisfying signature  the two pins labeled $x$ must receive the same   value, the two pins labeled $\bar{x}$ must receive the same value, and the values assigned to $x$ and $\bar{x}$ must be opposite.

\item[$\bullet$] \textbf{Wire gadgets}, 
whose symbols are 
shown in Figures~\ref{types}(c)--(h). A  wire gadget  has  odd parity and two outer pins. It  comes in six variants, depending on the locations of the pins.  
Its intended logic is determined by its odd parity: the Boolean variables assigned to its two outer pins are negations of one another. All signatures of the gadget have the same cost and are therefore satisfying signatures.

We refer to these gadgets as  \emph{ wire gadgets} because concatenating multiple copies yields a wire that propagates a Boolean value while negating it (see Section~\ref{WireSec}). We use the negating version for technical convenience. Whenever a literal appears in a clause, the corresponding wire is connected to the complementary literal in the variable gadget, thereby compensating for the negation.

\item[$\bullet$] \textbf{Crossing-wire gadget}, whose symbol is shown in Figure~\ref{types}(i). This gadget has even parity and  four outer pins. It implements the crossing of two wires.  In a satisfying signature, the right pin carries the negation $\bar{x}$ of the Boolean variable $x$ carried by the left pin, and the lower pin carries the negation $\bar{y}$ of the Boolean variable $y$ carried by the upper pin.

While a crossing-wire  gadget that preserves the values of $x$ and $y$ can also be constructed, 
we use the negating version for technical convenience, as it allows a wire to be intercepted by a crossing-wire  gadget without negating the value at the destination pin of the wire.

\item[$\bullet$] \textbf{Equality gadget}, whose symbol is shown in Figure~\ref{types}(j). This gadget has even parity and two outer pins. Its intended logic is determined by its even parity:  the Boolean variables assigned to the two pins are equal. All signatures of the gadget have the same cost and are therefore satisfying signatures.

\item[$\bullet$] \textbf{Garbage collection gadget}, whose symbol is shown in Figure~\ref{types}(k). This gadget has three outer pins and comes in two variants: even and odd, depending on its parity. Its intended logic is determined by its parity.  All signatures of the gadget have the same cost and are therefore satisfying signatures.

\end{itemize}

All gadget types except the equality gadget are used in the NP-hardness reduction, while all are used in the hardness-of-approximation reduction.   
The reduction is obtained by gluing together outer pins of different atomic gadgets until every outer pin is identified with another outer pin. Gluing two outer pins forces the corresponding Boolean variables of the two gadgets to be negations of one another. 
We next describe the common internal structure of the atomic gadgets.
 
All atomic gadgets have the same dimensions, namely $4\Delta\times4\Delta$, and their outer pins are located at distance $2\Delta$ from the corners of the box, where 
$\Delta$ is a global parameter called the  \emph{separation parameter}. 
 Consequently, atomic gadgets can be connected arbitrarily. This compatibility is achieved by the common structure shown in Figure~\ref{general}, which we describe next.

\begin{figure}[!htbp]\centering
\includegraphics[width=0.35\textwidth]{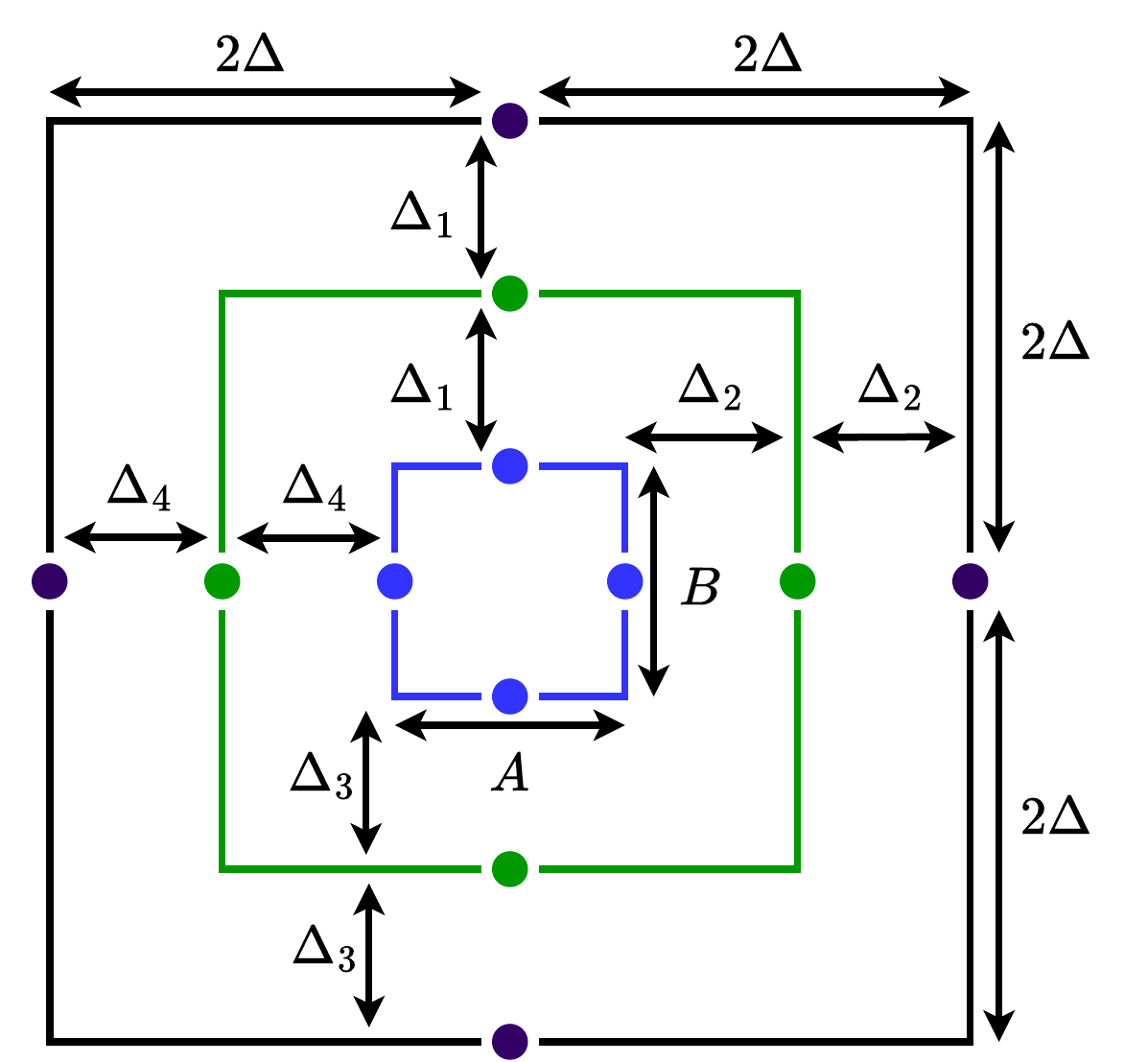}
\caption{\bf{General atomic gadget structure}.}
\label{general}
\end{figure}

Each gadget consists of three nested regions: the \emph{outer box}, shown in black; the \emph{relay box}, shown in green; and the \emph{nucleus box}, shown in blue.  
The outer box contains the gadget's \emph{outer pins}, shown in black, which correspond to the blue circles in the gadget symbols in Figure~\ref{types}; each gadget has $2$, $3$, or $4$ outer pins.  
 Each outer pin has a corresponding \emph{relay pin} on the relay box shown in green, which in turn has a corresponding \emph{nucleus pin}
shown in blue on the nucleus box whenever the nucleus is nonempty (the even garbage-collection gadget and the equality gadget have empty nuclei).
A relay pin is equidistant from its corresponding outer pin and its corresponding nucleus pin or, if the nucleus is empty, from the center of the gadget. We call the distance from a relay pin to its corresponding outer pin the \emph{depth} of the relay pin. The depth is at most $\Delta$. In Figure~\ref{general}, the depths of the relay pins are denoted by
$\Delta_1,\Delta_2,\Delta_3,$ and $\Delta_4$.

The logical behavior of the gadget is realized entirely within the nucleus box. The nucleus pins are distinguished primal defects through which the nucleus communicates with the rest of the gadget. All other primal and dual defects inside the nucleus are called \emph{internal nucleus defects}.

The parity of the nucleus defects is the same as that of the gadget, since
the remaining primal defects occur in outer-pin--relay-pin pairs.
We analogously define a \emph{nucleus signature} by replacing the outer
pins with the nucleus pins in the definition of a signature above and
restricting attention to the defects contained in the nucleus.
Thus, a nucleus signature specifies the set of nucleus pins assigned
the value \emph{True}, subject to the gadget parity constraint.
Analogously, we define the \emph{cost of a nucleus signature}.

The nucleus has dimensions $A\times B$, where $A,B=o(\Delta)$. 
The global separation parameter $\Delta$ is chosen sufficiently large to ensure that every optimal solution to the Minimum-Weight Join problem is \emph{local}; that is, it contains no unintended interactions between distinct gadgets.
 In particular, we set $\Delta=\Theta(m^2)$ for the NP-hardness reduction and $\Delta=\Theta(m^5)$ for the hardness-of-approximation reduction, where $m$ denotes the number of clauses in the input $3$-CNF formula. Moreover, $A=\Theta(1)$ for all atomic gadgets, while $B=\Theta(1)$ for all but the variable gadget in the hardness-of-approximation reduction, for which $B=\Theta(m)$.

Define the \emph{relay cost}  of a gadget  as the sum of depths of its relay pins.  
 Since each relay pin is equidistant from its corresponding outer pin and nucleus pin (or from the center of the gadget if the nucleus is empty), 
 the relay layer maps every 
nucleus signature   
 to an identical signature on the outer pins, increasing its cost by exactly the relay cost of the gadget, independent of the signature. Thus, the relay layer propagates the Boolean values from the nucleus pins to the outer pins without altering the logical behavior of the gadget or the cost differences between signatures.

A final remark is that, since the reduction involves more than one gadget, when the domain is a torus its side length $K$ necessarily satisfies $K\ge 8\Delta$. Henceforth, this condition will be understood without further mention. In particular, it justifies treating each gadget as isolated when computing the costs of its signatures.

\subsection{Outline  of the reduction}\label{nphoverview}

The constructions of atomic gadgets that meet the above specifications are given
in Sections~\ref{AdaptedGadgets}, \ref{VarGadget}, and~\ref{CrossGadget}.
We next give an overview of the reduction, which is based on these gadgets.

Consider a $3$-CNF formula $\phi$. For each clause, we use a clause gadget.
For each variable of $\phi$, we connect enough atomic variable gadgets
by gluing their pins to produce copies of the variable and its negation
for all occurrences of its literals in the clauses.
We connect these copies to the literal pins of the clause gadgets using
wires, with crossing-wire gadgets wherever wires cross.
Garbage-collection gadgets absorb the garbage-collection pins of the
clause gadgets, as well as unused pins of the variable gadgets.
The resulting composite gadget has no free outer pins and has even numbers
of primal and dual defects. 
Its constituent
atomic gadgets admit compatible satisfying signatures if and only if
$\phi$ is satisfiable.
Composite gadgets  are formalized in
Section~\ref{CompositeGadgets}, while the details of the construction
from $\phi$ are given in Section~\ref{redS}.

Since, by definition, the satisfying signatures of each atomic gadget
are precisely its minimum-cost signatures, this reduces the satisfiability
of $\phi$ to determining whether the composite gadget admits a collection of compatible  signatures for its  atomic gadgets, 
in which every atomic gadget has minimum signature cost.
The remaining issue is to relate this optimization over gadget signatures
to the Minimum-Weight Join problem on the defects of the composite gadget.
A general join need not respect the decomposition into atomic gadgets:
paths may connect defects belonging to different gadgets, and primal and
dual paths from different gadgets may interact.
We therefore introduce local joins in Section~\ref{LocalJoins}, for which
such interactions are excluded and the optimization decomposes over the
individual gadgets.
We then establish the Localization Lemma, which shows that, by choosing
the separation parameter $\Delta$ sufficiently large, every minimum-weight
join is local.
Consequently, $\phi$ is satisfiable if and only if the minimum join weight
for the defects of the composite gadget is equal to the sum, over all
atomic gadgets, of their minimum signature costs.
In Section~\ref{redS}, we use this correspondence to complete the reduction
from $3$SAT and establish NP-hardness of the Minimum-Weight Join problem.

\subsection{Gadgets adapted from the  Walters-Turner Reduction}\label{AdaptedGadgets}

The constructions of the wire, clause, and garbage-collection gadgets are adapted from the Walters--Turner reduction~\cite{WT26} by reworking its underlying ideas for our setup on the square lattice. In particular, the idea of using a garbage-collection pin for the clause gadget, together with an associated garbage-collection mechanism, is due to Walters and Turner~\cite{WT26}.

\subsubsection{Wire gadgets and  wires}\label{WireSec}

{
\renewcommand{\galscale}{0.2}
\begin{figure}[!htbp]
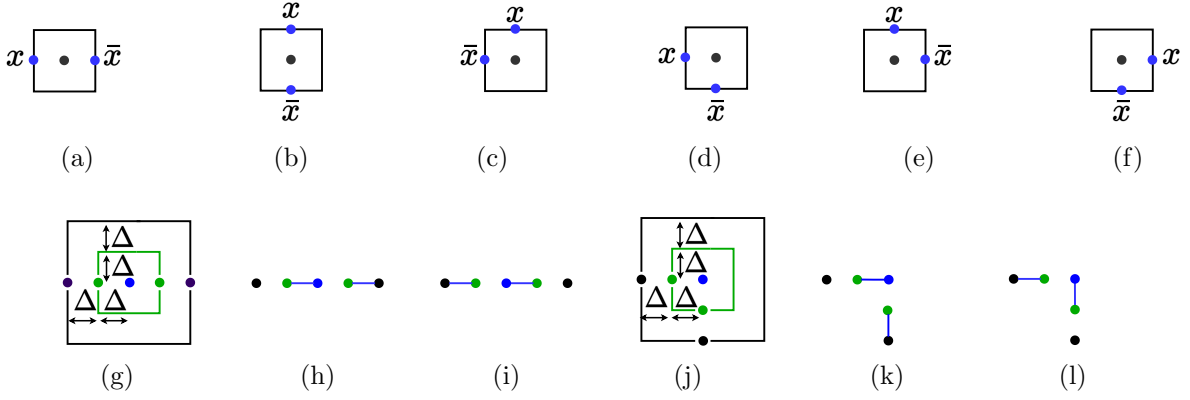
\centering
\sfg{type-neg1.pdf}\hfill 
\sfg{type-neg2.pdf}\hfill  
\sfg{type-neg3.pdf}\hfill 
\sfg{type-neg4.pdf}\hfill 
\sfg{type-neg5.pdf}\hfill  
\sfg{type-neg6.pdf}\hfill
\\[1em]
\sfg{boxed-nw1.pdf}\hfilli\hfilli\hfilli
\sfg{nw1-1.pdf}\hfilli\hfilli\hfilli 
\sfg{nw1-2.pdf}\hfilli\hfilli%
\sfg{boxed-nw2.pdf}\hfilli\hfilli\hfilli 
\sfg{nw2-1.pdf}\hfilli\hfilli\hfilli 
\sfg{nw2-2.pdf}
\caption{\textbf{Wire gadgets}. Parts~(a)--(f) show the symbols of the six variants of the wire gadget.
The structure of the first variant is shown in Part~(g). Its nucleus consists of a single primal defect. The two signatures of this gadget, together with their corresponding minimum-weight joins, are shown in parts~(h) and~(i). Both have cost $2\Delta$. 
The structure of the variant shown in Part~(d) is illustrated in Part~(j), while its two signatures are shown in Parts~(k) and~(l). The remaining variants are constructed analogously. }
\label{negwf}
\end{figure}
}
The internal structure of the wire gadgets is shown in Figure~\ref{negwf}.
 Every signature of each of the six variants has cost $2\Delta$.

{
\renewcommand{\galscale}{0.2}
\begin{figure}[!htbp]
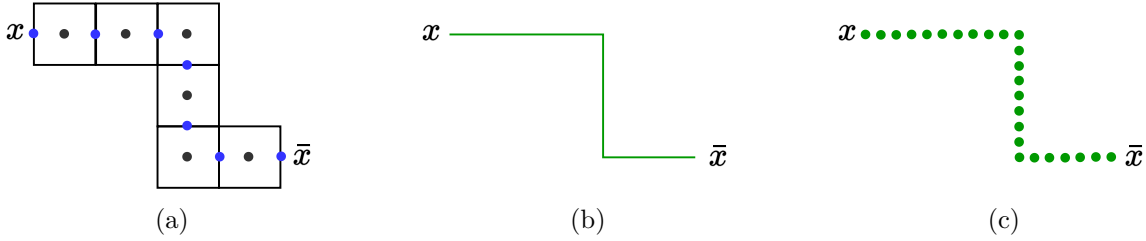
\centering
\sfg{nwex.pdf}\hfill 
\sfg{nwex-format2.pdf}\hfill  
\sfg{nwex-format1.pdf}\hfill 
\caption{\bf{Wire example}.}
\label{negwex}
\end{figure}
}

By connecting  wire gadgets, we can construct  wires of arbitrary shape whose segments are horizontal or vertical, as shown in Figure~\ref{negwex}(a). Such wires propagate the Boolean value from one endpoint to the other while negating it, and their lengths are multiples of $4\Delta$.

Instead of drawing a wire explicitly as in Figure~\ref{negwex}(a), we shall represent it throughout the paper either by a green line, as in Figure~\ref{negwex}(b), or by its defects, as in Figure~\ref{negwex}(c).

An immediate consequence of the atomic wire construction is that any two parallel wire segments are automatically separated by a distance of at least $4\Delta$. Thus, wire spacing need not be enforced separately to prevent wire interactions.

\subsubsection{Clause gadget}\label{ClauseGadget}

The kernel of the clause gadget, which is the logical core of its nucleus,
is shown in Figure~\ref{clauseGadetKernel}.

\begin{figure}[!htbp]\centering
\sfg{Clause-gadget.pdf}\hfilli\hfilli\hfilli\hfilli\hfilli
\sfg{Clause-gadget-1.pdf}
\sfg{Clause-gadget-2.pdf}\hfill
\sfg{clause-gadget-3.pdf}\hfill\sfg{cl-g-4.pdf}\hfill
\sfg{cl-g-5.pdf}\\[1em]
\sfg{cl-g-6.pdf}\hfill
\sfg{cl-g-9.pdf}\hfill
\sfg{cl-g-7.pdf}
\sfg{cl-g-8.pdf}
\hfill
\sfg{cl-g-10.pdf}
\sfg{cl-g-11.pdf}
\sfg{cl-g-12.pdf}
\caption{\textbf{Clause gadget kernel}.
The kernel of the clause gadget is shown in (a).
It   has odd parity and four  pins: three \emph{literal pins}, shown as blue circles, and one \emph{garbage-collection pin}, shown as a black circle.
 Parts (b)–(m) show all eight signatures of the gadget together with their corresponding minimum-weight joins and their costs, denoted by $c$. Pins assigned the value \emph{True} are indicated by green-filled circles. 
Some signatures have multiple minimum-weight joins, all of which are shown in the figure. For every signature in which at least one literal pin is assigned \emph{True}, the cost is $c$ = 5. In contrast, the unique signature in which all three literal pins are assigned \emph{False} has cost $c = 6$.
}
\label{clauseGadetKernel}
\end{figure}

{
\renewcommand{\galscale}{0.3}
\begin{figure}[!htbp]\centering
\sfg{Boxed-clause-gadget.pdf}
\hfilli\hfilli\hfilli
\sfg{type-clause.pdf}
\caption{(a) \textbf{Boxed clause gadget}; (b) Clause gadget symbol.}
\label{fig:boxed-clause}
\end{figure}
}

The boxed clause gadget is shown in Figure~\ref{fig:boxed-clause}(a). Because the kernel pins are not symmetric, the kernel is wrapped in the nucleus and connected to its pins by internal wires, whose defects are shown in black.
By an \emph{internal wire}, we mean a sequence of an odd number of primal defects in which every two consecutive defects are at distance one. Like external wires, internal wires propagate the truth value from one end to the other, negated. 
The internal wires faithfully map the signatures of the kernel pins to signatures of the nucleus pins while preserving the cost differences between them. 

Note that for any signature of the outer pins, there exists a minimum-weight join in which every relay pin is matched to its corresponding outer or nucleus pin. In particular, we need not consider minimum-weight joins in which two relay pins are matched to each other. Indeed, any path connecting two relay pins can be replaced, without increasing the weight, by a path through their corresponding nucleus pins. This observation applies to all gadgets with nonempty nuclei considered in this paper, the only exceptions being the even garbage-collection gadget and the equality gadget.

 Thus,
the relay pins faithfully map the   signatures of the nucleus pins  
 to  signatures of the outer pins while preserving the cost differences between them.

In summary, all satisfying signatures of the clause gadget  have the same cost, which is one less than the cost of its unique nonsatisfying signature.

\subsubsection{Garbage collection gadgets}
\label{GCGadget}

The even and odd atomic garbage-collection gadgets are shown in Figure~\ref{gca-fig}.
The atomic gadgets have degree~$3$. Garbage-collection gadgets of degree~$k$, for $k \geq 4$, are constructed by combining atomic gadgets as shown in Figure~\ref{gcdegkf}.

{
\renewcommand{\galscale}{0.23}
\begin{figure}
[!htbp]
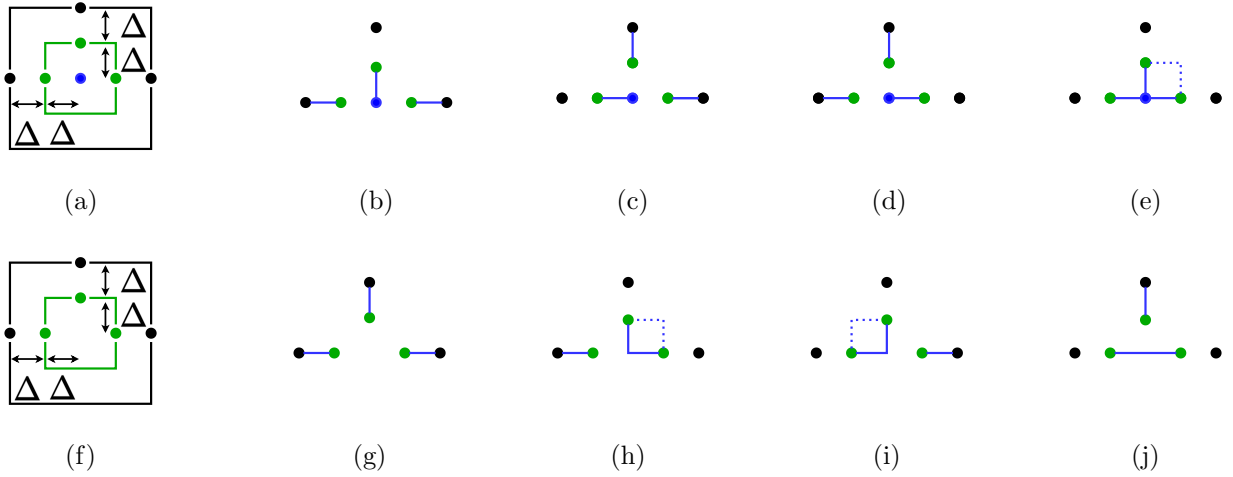
\centering
\sfg{GC-Gadget-odd.pdf}
\hfill
\sfg{GC-odd-1.pdf}\hfill
\sfg{gc-odd-2.pdf}\hfill
\sfg{gc-odd-3.pdf}\hfill
\sfg{gc-odd-4.pdf}\\[1em]
\sfg{GC-Gadget-even.pdf}\hfill
\sfg{gc-even-1.pdf}\hfill
\sfg{gc-even-2.pdf}\hfill
\sfg{gc-even-4.pdf}\hfill
\sfg{gc-even-6.pdf}\hfill
\caption{\textbf{Garbage-collection gadget}. The odd garbage-collection gadget is shown in (a). Its nucleus consists of a single primal defect. Parts~(b)–(e) show its four signatures together with their corresponding minimum-weight joins. The even garbage-collection gadget is shown in (f). Its nucleus is empty. Parts~(g)–(j) show its four signatures together with their corresponding minimum-weight joins. Every signature of both gadgets has cost $3\Delta$. 
 }
\label{gca-fig}
\end{figure}
}

{
\renewcommand{\galscale}{0.2}
\begin{figure}
[!htbp]
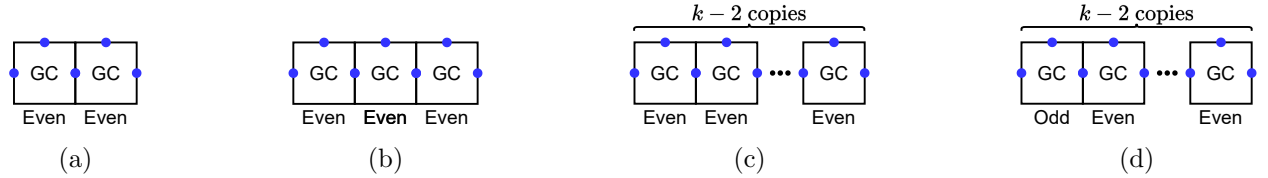
\centering
\sfg{Deg-k-GC-odd.pdf}\hfill
\sfg{deg-k-GC-even.pdf}\hfill
\sfg{Deg-k-GC.pdf}\hfill
\sfg{Deg-k-GC1.pdf}
\caption{
\textbf{Higher-degree garbage-collection gadgets}. Connecting two even atomic gadgets, as shown in (a), yields the odd degree-$4$ garbage-collection gadget. Adding another even atomic gadget produces the even degree-$5$ garbage-collection gadget. As additional even atomic gadgets are added, the parity alternates between odd and even. The gadget shown in Part~(c) has degree $k$ and is odd if $k$ is even, and even if $k$ is odd. To flip the parity of the gadget, we flip the parity of one of its atomic gadgets, as shown in Part~(d).
}
\label{gcdegkf}
\end{figure}
}

\subsection{Variable gadget}\label{VarGadget}

\begin{figure}
[!htbp]
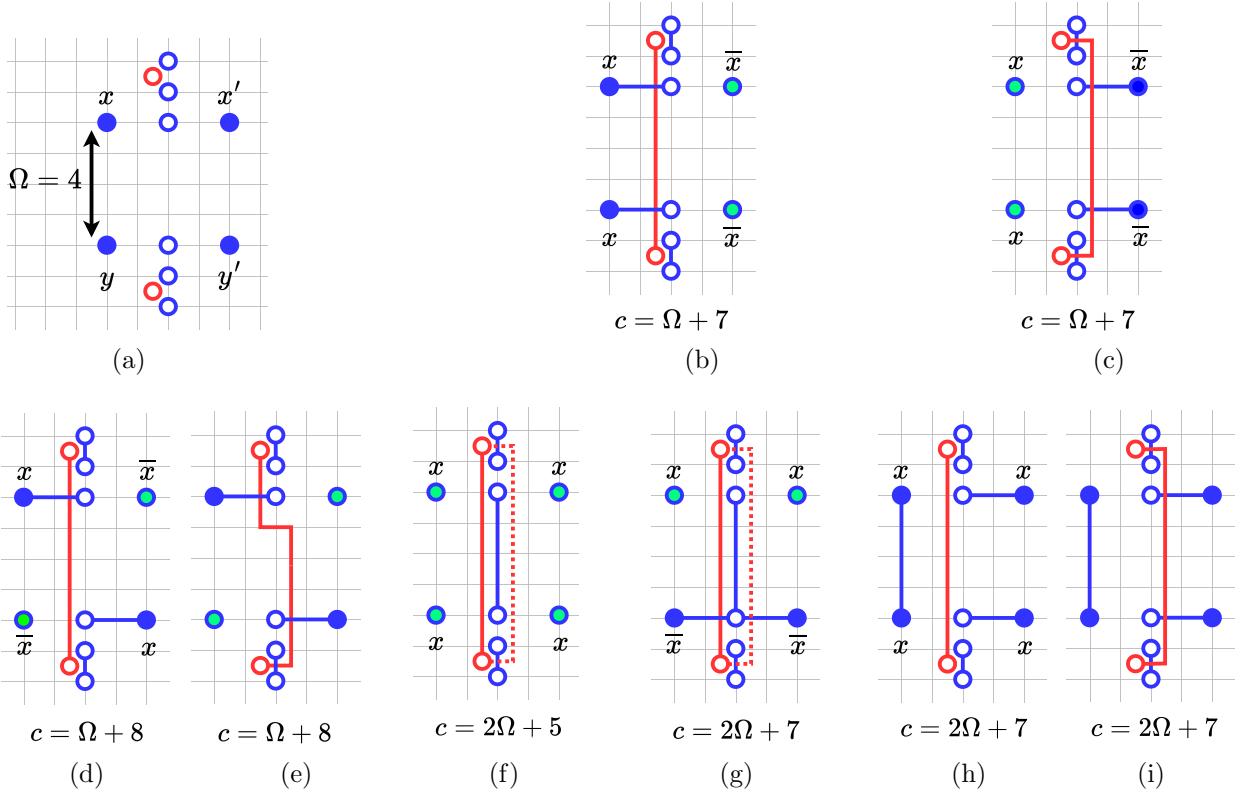
\centering
\sfg{var-gadget.pdf}\hfill\hfill\sfg{var-1.pdf}\hfill\sfg{var-2.pdf}
\\[1em]
\hspace{-2.5em}
\sfg{var-3.pdf}
\hspace{-2.3em}
\sfg{var-4.pdf}\hspace{-0em}\hfill\sfg{var-5.pdf}\hspace{-1em}\hfill\sfg{var-6.pdf}\hspace{-1em}\hfill\sfg{var-7.pdf}
\hspace{-0.5em}\sfg{var-8.pdf}
\caption{(a) \textbf{Variable gadget kernel};
(b), (c) \textbf{Satisfying signatures}, each with its corresponding minimum-weight join and cost $c$; (d)--(i) \textbf{Unsatisfying signatures},  up to symmetry.}  
\label{varkernel}
\end{figure}

The kernel of the atomic variable gadget is shown in Figure~\ref{varkernel}(a). It is parameterized by a positive integer $\Omega$ divisible by 4. The kernel has even parity and  four pins associated with the Boolean variables $x$, $y$, $x'$, and $y'$. Since the gadget has even parity, it has eight signatures corresponding to all assignments in which an even number of variables are assigned \emph{True}. Parts~(b)–(i) show all signatures, up to symmetry, together with their corresponding minimum-weight joins and costs.

The two signatures that capture the gadget's intended logic, shown in Parts~(b) and~(c), satisfy the constraints $x=y$, $x'=\bar{x}$, and $y'=\bar{y}$, and have cost $c = \Omega+7$. Every other signature violates at least one of these constraints and has higher cost.   There are two types of violating signatures. 
If only the constraint $x=y$ is violated while $x'=\bar{x}$ and $y'=\bar{y}$ are satisfied, the cost increases by $1$, as shown in Parts~(d) and~(e). If either $x'=\bar{x}$ or $y'=\bar{y}$ is violated, the cost increases by at least $\Omega-2$, as shown in Parts~(f)–(i).

In the NP-hardness reduction, we set $\Omega=4$, since the distinction among violating signatures is unnecessary; it suffices that every violating signature has cost at least one greater than that of the satisfying signatures. In the hardness-of-approximation reduction, we set $\Omega$ to a
larger value, ensuring that every signature violating
$x'=\bar x$ or $y'=\bar y$ incurs a sufficiently large cost penalty.

The boxed atomic variable gadget is shown in Figure~\ref{boxed-var}(a). As with the clause gadget, the kernel pins are connected to the nucleus pins by  internal wires.

{
\renewcommand{\galscale}{0.3}
\begin{figure}[!htbp]
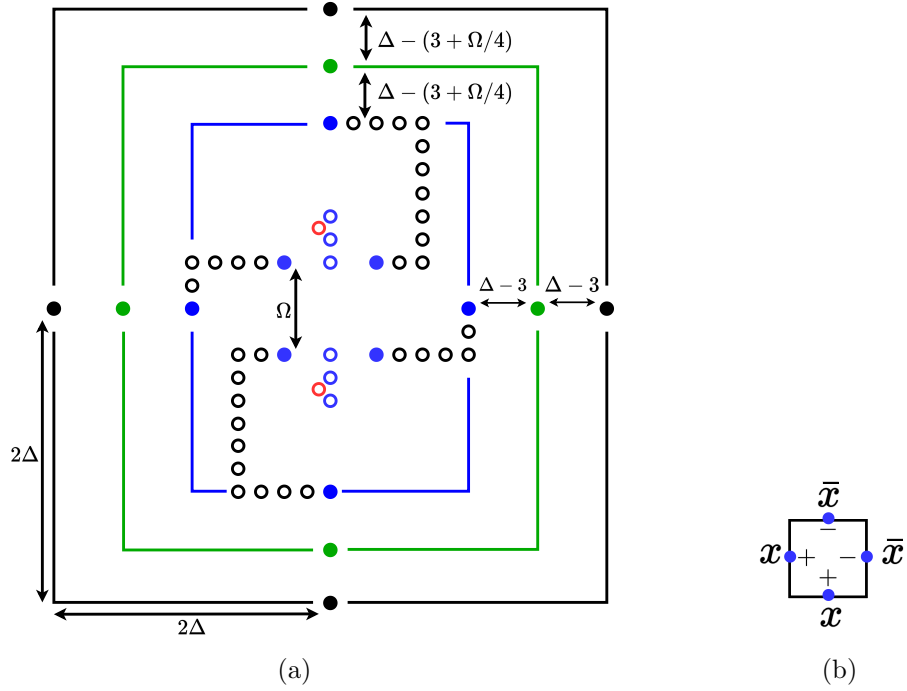
\centering
\sfg{boxed-var.pdf}\hfilli\hfilli\hfilli
{
\renewcommand{\galscale}{0.25}
\sfg{type-var-labeled.pdf}
}
\caption{(a) \textbf{Boxed variable gadget}; (b) Variable gadget symbol.}
\label{boxed-var}
\end{figure}
}

The atomic variable gadget has degree~2. Variable gadgets of degree~$k$, for $k \geq 3$, are constructed from atomic variable gadgets as shown in Figure~\ref{degkvarg}. For degree~1, the atomic wire gadget serves as a variable gadget.

\begin{figure}[!htbp]
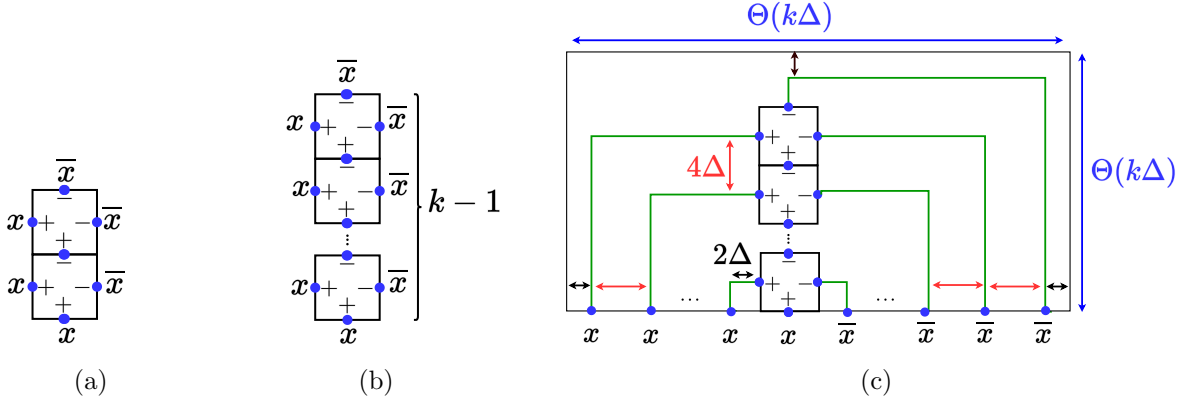
\centering
{
\renewcommand{\galscale}{0.21}
\sfg{deg-2-var-gadget.pdf}\hfill\sfg{deg-k-var-gadget.pdf}\hfill
}
{
\renewcommand{\galscale}{0.19}
\sfg{deg-k-var-wiring.pdf}
}
\caption{\textbf{Higher-degree variable gadgets.} (a) Degree-$3$ variable gadget; (b) degree-$k$ variable gadget, for $k\geq 3$; (c) wired degree-$k$ variable gadget   
with outer pins spaced by $4\Delta$.}
\label{degkvarg}
\end{figure}

\subsection{Crossing-wire  gadget}\label{CrossGadget}

{
\renewcommand{\galscale}{0.30}
\begin{figure}[!htbp]
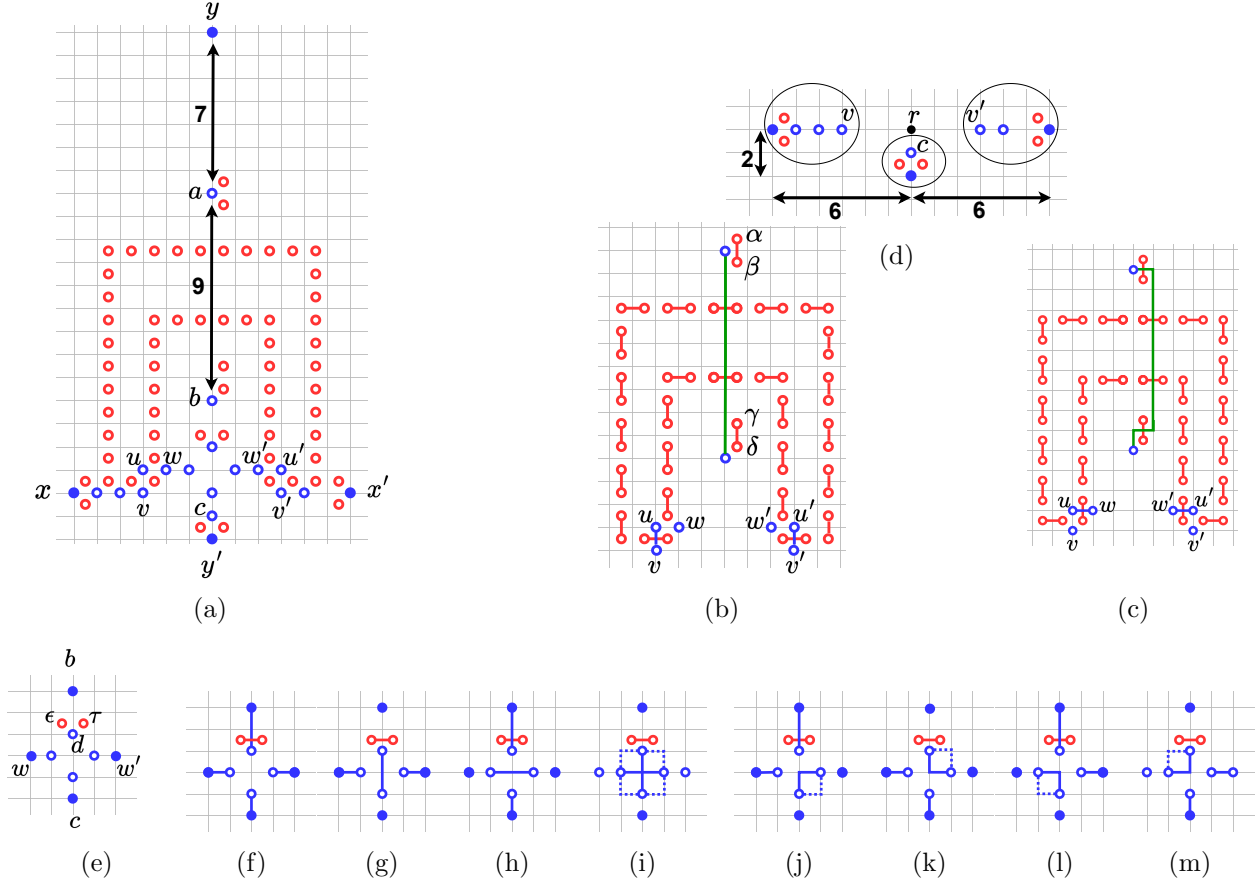
\centering
\sfg{Crossing-Gadget.pdf}\hfilli\hfill
\sfg{cross-justification-1.pdf}\hfill
\sfg{cross-justification.pdf}
\\[-21em]
\hspace{20 em}
\sfg{Crossing-fix.pdf}
\vspace{12 em}
\\[1em]
{
\renewcommand{\galscale}{0.28}
\sfg{cross-just.pdf}\hfill
\sfg{cross-just-1.pdf}\hfill
\sfg{cross-just-6.pdf}\hfill
\sfg{cross-just-3.pdf}\hfill
\sfg{cross-just-8.pdf}\hfill\hfilli
\hfillih
\sfg{cross-just-2.pdf}\hfill
\sfg{cross-just-5.pdf}\hfill
\sfg{cross-just-4.pdf}\hfill
\sfg{cross-just-7.pdf}
}
\caption{\textbf{Crossing-wire gadget nucleus}.}
\label{crossnuc}
\end{figure}
}

The nucleus of the crossing-wire gadget is shown in Figure~\ref{crossnuc}(a). The gadget has even parity and four pins associated with the Boolean variables $x$, $y$, $x'$, and $y'$, and hence admits eight signatures, corresponding to all assignments in which an even number of variables are assigned \emph{True}. The signatures realizing the gadget's intended logic satisfy the constraints $x'=\bar{x}$ and $y'=\bar{y}$. These four signatures all have the same cost, as shown in Figure~\ref{crossgood}, together with their corresponding minimum-weight joins. The remaining four signatures have cost exactly one greater, as shown in Figure~\ref{crossbad}. The corresponding boxed gadget is shown in Figure~\ref{boxedcorss}. 

We next explain the construction of the gadget and the mechanism by which it realizes these costs.
 The gadget consists of three ingredients: a dual cycle that couples the left and right halves, constant-cost parity subgadgets, and auxiliary dual defects to achieve the desired lengths.

A key component of the gadget is the red dual cycle, whose two configurations are shown in Parts~(b) and~(c) of Figure~\ref{crossnuc}. To minimize the join cost, the cycle couples the state of whether $u$ is joined to $v$ or $w$ with the state of whether $u'$ is joined to $v'$ or $w'$, respectively. We show next that this coupling enforces the constraint $x'=\bar{x}$. 
 The constraint $y'=\bar{y}$ then follows from the even parity of the gadget. 
 
 The encircled  structures containing $v$ and  $v'$, shown in Part~(d), are even- and odd-parity constant-cost gadgets, respectively. They double the distances from $r$ to $v$ and $v'$, while ensuring that the Boolean values represented by $v$ and $v'$ are $\bar{x}$ and $x'$, respectively. Hence, the coupling induced by the red cycle enforces $x'=\bar{x}$.
 
In every minimum-weight join, the defect $a$ is joined to either $y$ or $b$.
To make the cost of joining $a$ to $b$ independent of the configuration of the red cycle, the dual defects $\alpha$, $\beta$, $\gamma$, and $\delta$ shown in (b) are introduced. The distances between $a$ and $y$, and between $a$ and $b$, are chosen so that joining $a$ to $b$ has the same cost as joining $a$ to $y$.

The subgadget shown in Part~(e) realizes an even-parity gadget on the pins $w$, $b$, $w'$, and $c$. All of its signatures have the same cost, namely $5$, as shown in Parts~(f)--(m). The red defects $\epsilon$ and $\tau$ compensate for the fact that $b$ is farther from the center than $c$.

The encircled structure containing $c$, shown in Part~(d), is an even-parity constant-cost 
gadget that doubles the distance from $r$ to $c$ while ensuring that the Boolean value represented by $c$ is $\bar{y'}$.

{
\renewcommand{\galscale}{0.29}
\begin{figure}[!htbp]
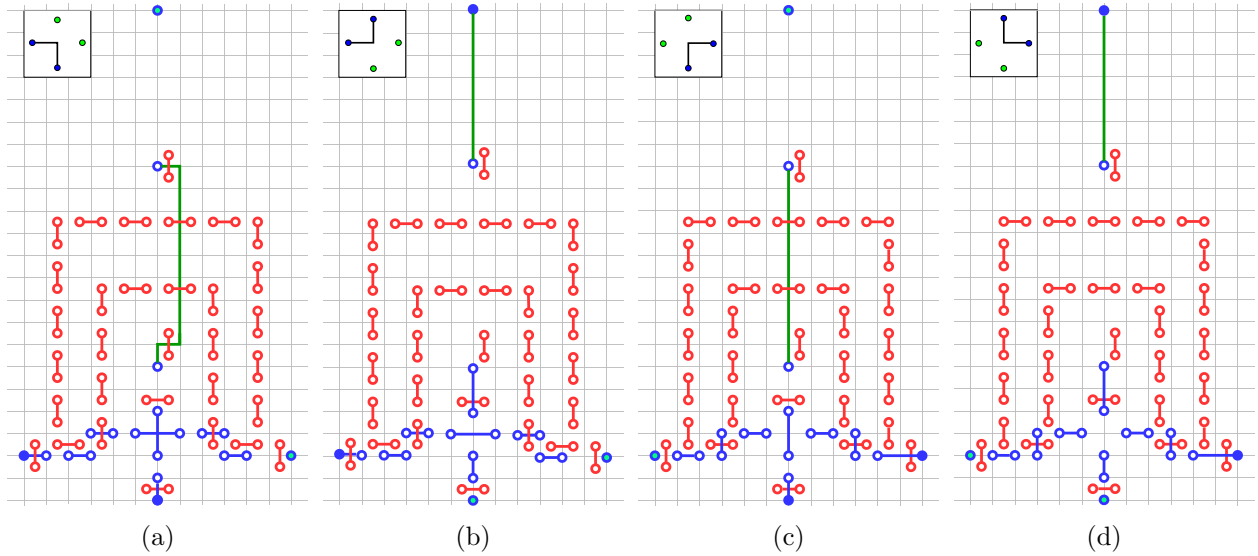
\centering
\sfg{cross-1.pdf}\hfill
\sfg{cross-2.pdf}\hfill
\sfg{cross-3.pdf}\hfill
\sfg{cross-4.pdf}
\caption{\textbf{Satisfying signatures of the  crossing-wire   gadget}.    The four signatures satisfying	
$x'=\bar{x}$ and $y'=\bar{y}$ are shown with their corresponding minimum-weight joins; each signature has cost $45$.}
\label{crossgood}
\end{figure}
}

{
\renewcommand{\galscale}{0.29}
\begin{figure}[!htbp]
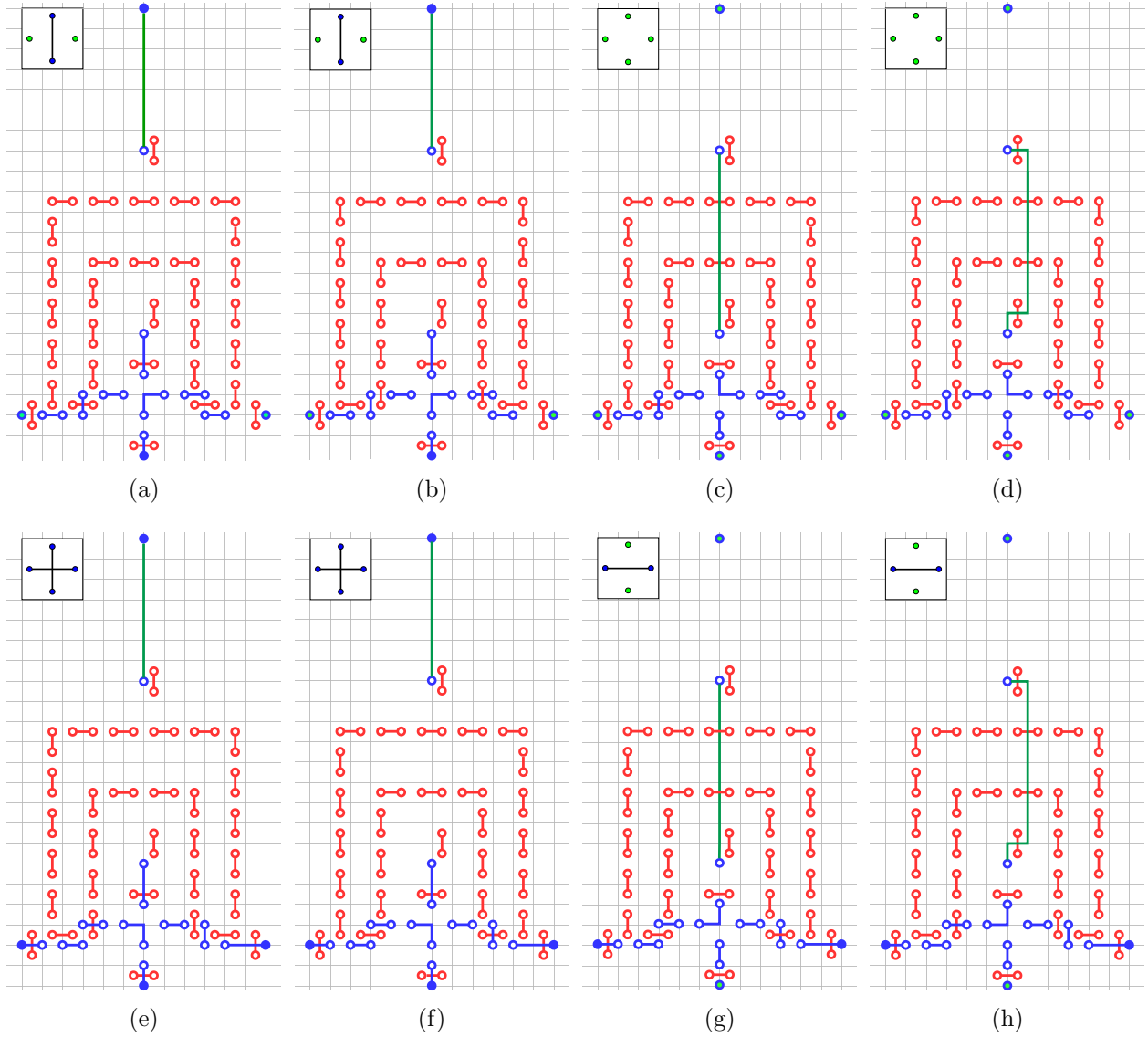
\centering
\centering
\sfg{cross-5.pdf}\hfill
\sfg{cross-6.pdf}\hfill
\sfg{cross-7.pdf}\hfill
\sfg{cross-8.pdf}\\[1em]
\sfg{cross-9.pdf}\hfill
\sfg{cross-10.pdf}\hfill
\sfg{cross-12.pdf}\hfill
\sfg{cross-11.pdf}
\caption{\textbf{Nonsatisfying signatures 
of the crossing-wire   gadget.} The four signatures that violate
$x'=\bar{x}$ or $y'=\bar{y}$ are shown with their corresponding minimum-weight joins; each signature has cost 
 $46$.}
\label{crossbad}
\end{figure}
}

\begin{figure}[!htbp]
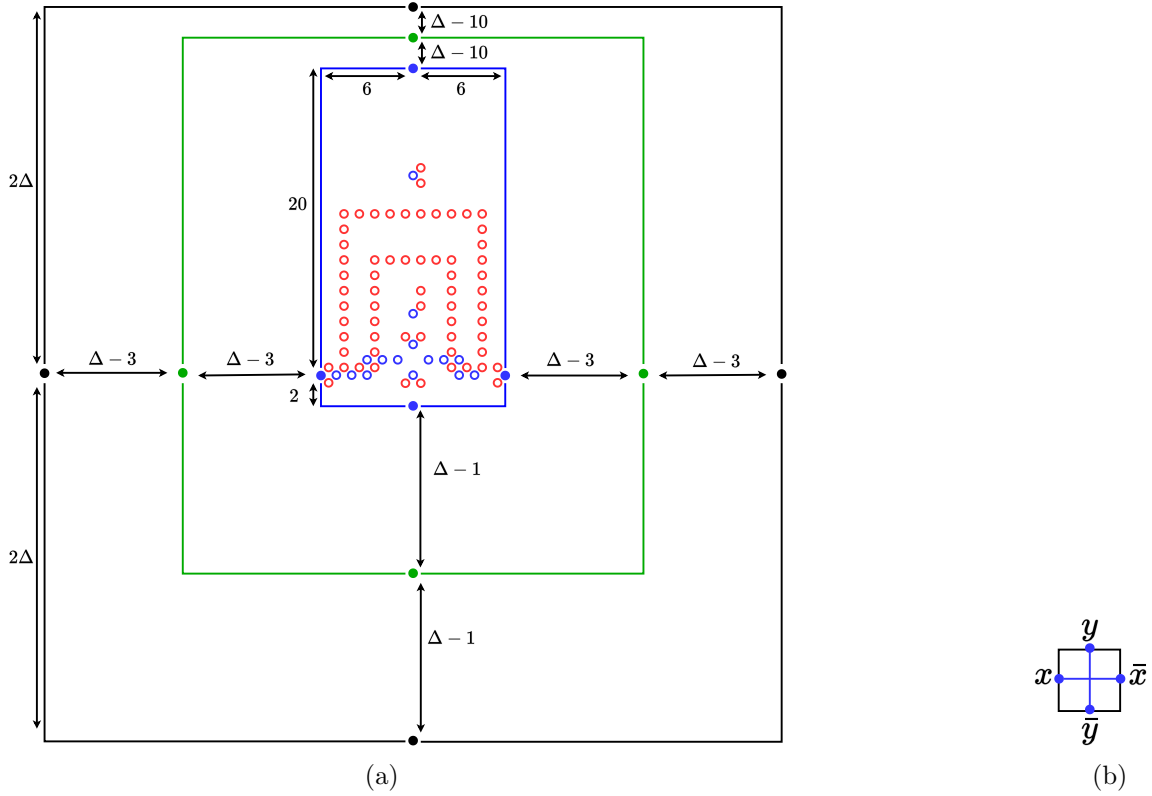
\centering
{
\renewcommand{\galscale}{0.2}
\sfg{Boxed-crossing-gadget.pdf}
}
\hfill
{
\renewcommand{\galscale}{0.2}
\sfg{type-cross-labeled.pdf}
}
\caption{(a) \textbf{Boxed crossing-wire gadget}; (b) Crossing-wire gadget symbol. }
\label{boxedcorss}
\end{figure}

\subsection{Composite gadgets}\label{CompositeGadgets} 
A \emph{composite gadget} is obtained by gluing together outer pins of atomic gadgets. The \emph{outer pins} of the composite gadget are precisely the outer pins of its constituent atomic gadgets that remain unglued. We call a composite gadget \emph{closed} if it has no outer pins.

When the domain is a torus, it is implicitly assumed that the torus is sufficiently large to accommodate the composite gadget. In particular, if a composite gadget is closed, the  side length of the torus must be at least  the minimum side length of a square containing the outer boxes of all its atomic gadgets. Note that, in the closed case, the boundary atomic gadgets have no outer pins. Hence, even if the composite gadget fits tightly in the torus, identifying opposite sides does not affect any of its atomic gadgets.

A \emph{configuration} of a composite gadget is a choice of one signature for each constituent atomic gadget such that, for every pair of glued outer pins, the Boolean values assigned by the corresponding atomic signatures are opposite.

Throughout the paper, every composite gadget is assumed to admit at least one configuration. Equivalently, if each atomic gadget is replaced by a parity constraint requiring the parity of its incident Boolean variables to equal the parity of the gadget, then the resulting system of equations over $\mathbb{F}_2$ is assumed to have a solution.

The \emph{cost} of a configuration    
is the sum of the costs of the signatures assigned to its constituent atomic gadgets.

 A configuration is called \emph{satisfying} if every constituent atomic gadget is assigned a satisfying signature. Since, by definition, a satisfying signature of an atomic gadget has minimum cost, all satisfying configurations have the same cost.

A composite gadget is called \emph{satisfiable} if it admits a satisfying configuration. 
The \emph{relative cost} of a configuration of a satisfiable composite gadget is the difference between its cost and the cost of a satisfying configuration.

\subsection{Local joins}\label{LocalJoins}

To eliminate interactions between distinct atomic gadgets, we introduce the notion of a \emph{local join}. 
In a local join, no path connects defects belonging to distinct gadgets. Moreover, no dual path in one gadget couples with a primal path in another gadget.  
This is the key property underlying our reductions: it decomposes the global optimization problem over joins into a global optimization over gadget signatures together with independent local optimization problems over joins within the individual gadgets.

\begin{definition}[Local  joins]
Consider a closed composite gadget $G$ and a join $J$ for the defects in $G$. 

The \emph{restriction} of $J$ to an atomic gadget $g$ in $G$ is the set of primal and dual paths of $J$ whose endpoints both lie in $g$.

We call $J$ \emph{local} with respect to $G$  if:
\begin{itemize}
\item[(a)]
it is the union of its restrictions to the atomic gadgets of $G$, and
\item[(b)]
for every dual path $d$ in the restriction of $J$ to an atomic gadget $g$, no primal path in the restriction of $J$ to any other atomic gadget contains an edge dual to an edge of $d$.
\end{itemize}

If $J$ is local, the \emph{$J$-signature} of an atomic gadget $g$ is the unique subset $S$ of its outer pins such that the restriction of $J$ to $g$ is an $S$-join.
The \emph{configuration} of $J$ is the collection of the $J$-signatures of all atomic gadgets of $G$.
\end{definition}
\begin{definition}[Locally optimal joins]\label{locoptj}
A local join $J$  
with respect to a closed composite gadget $G$
is \emph{locally optimal}  if, for every atomic gadget $g$ of $G$, the restriction of $J$ to $g$ is a minimum-weight $S$-join, where $S$ is the $J$-signature of $g$.
\end{definition}

Thus, the weight of a local join is at least the cost of its configuration, with equality if and only if the join is locally optimal.

Conversely, each configuration is associated with a locally optimal join whose weight is equal to the cost of the configuration.

The following Localization Lemma shows that every non-local join incurs an additive penalty that grows with~$\Delta$. Consequently, by choosing $\Delta$ sufficiently large, we can ensure that every minimum-weight join is local, and hence that its weight is equal to the minimum cost of a configuration. 
 This property underlies the soundness of both the NP-hardness and the hardness-of-approximation reductions.

The \emph{offset} of a gadget $g$ is the maximum difference between $\Delta$ and the depth of one of its relay pins.
Except for the garbage-collection gadgets and the equality gadget, all gadgets have strictly positive offset.

\begin{restatable}{lemma}{LOCLemma}\label{loclemma}
\textbf{\emph{(Localization Lemma)}}
Consider a closed composite gadget $G$ and a join $J$ for the defects in $G$. 
If $J$ is not local, then the weight of $J$ is at least
\[
R+2\Delta/5-3\delta-1,
\]
where $R$ is the total relay cost of all gadgets, and $\delta$ is the maximum offset among the gadgets, assuming $\delta<2\Delta/5$.
\end{restatable}
The proof of the Localization Lemma is in Appendix~\ref{appA}.

\subsection{NP-hardness 
of   Minimum-Weight Join}\label{redS}

In this section, we prove that the Minimum-Weight Join problem is NP-hard by constructing, from a given $3$-CNF formula $\phi$ with $m$ clauses, a closed composite gadget $G$ such that, for a sufficiently large choice of $\Delta$, $\phi$ is satisfiable if and only if the minimum weight of a join for the defects in $G$ equals the baseline cost, namely, the sum over all atomic gadgets $g$ of the minimum cost of a signature of $g$. We establish this equivalence in Lemma~\ref{soundnessnp}. The proof relies on the Localization Lemma to establish the soundness direction, which requires $\Delta=\Theta(m^2)$. Using this characterization, we conclude in Theorem~\ref{nphardnessMWJ} that the Minimum-Weight Join problem is NP-hard.
 
We construct a closed composite gadget $G$ from $\phi$ as shown in Figure~\ref{layoutf}. For each variable, we include a degree-$k$ variable gadget, as shown in Figure~\ref{degkvarg}(c), where $k$ is the maximum number of occurrences of this variable or its negation in the clauses. For each clause, we include an atomic clause gadget. 

We connect the pins of the variable gadgets to the literal pins of the clause gadgets by wires. Since the wires implement negation, a literal in a clause is connected to the pin of the variable gadget corresponding to its negation. 

Next, we connect the garbage-collection pins of the clause gadgets to a garbage-collection layer, as constructed in Figure~\ref{gcdegkf}. Since, for each variable gadget, the number of copies of 
the corresponding variable 
 equals the number of copies of its negation, some literal pins of the variable gadgets may remain unused. We connect each such pin to the garbage-collection layer. Finally, we choose the parity of the garbage-collection layer so that the total number of primal defects in the composite gadget is even. 

Each wire crossing is implemented using a  crossing-wire  gadget.
To minimize the dimensions of the construction, we route the wires as follows. 
We connect the first pin of the first variable gadget via a wire consisting of three segments (ignoring the additional horizontal segment needed when the  target literal pin is on the left or right side of the clause  gadget), as shown by the solid green line in Figure~\ref{layoutf}. 
For all other pins of the variable gadgets, we use five-segment wires, as
shown in the figure.  We place all crossing-wire gadgets along the two horizontal segments of each of the five-segment wires. Thus, each such wire increases the depth of the construction by  exactly $8\Delta$. The resulting composite gadget $G$ has dimensions $O(m\Delta) \times O(m\Delta)$, as shown in the figure.

\begin{figure}[!htbp]\centering
\includegraphics[width=\textwidth]{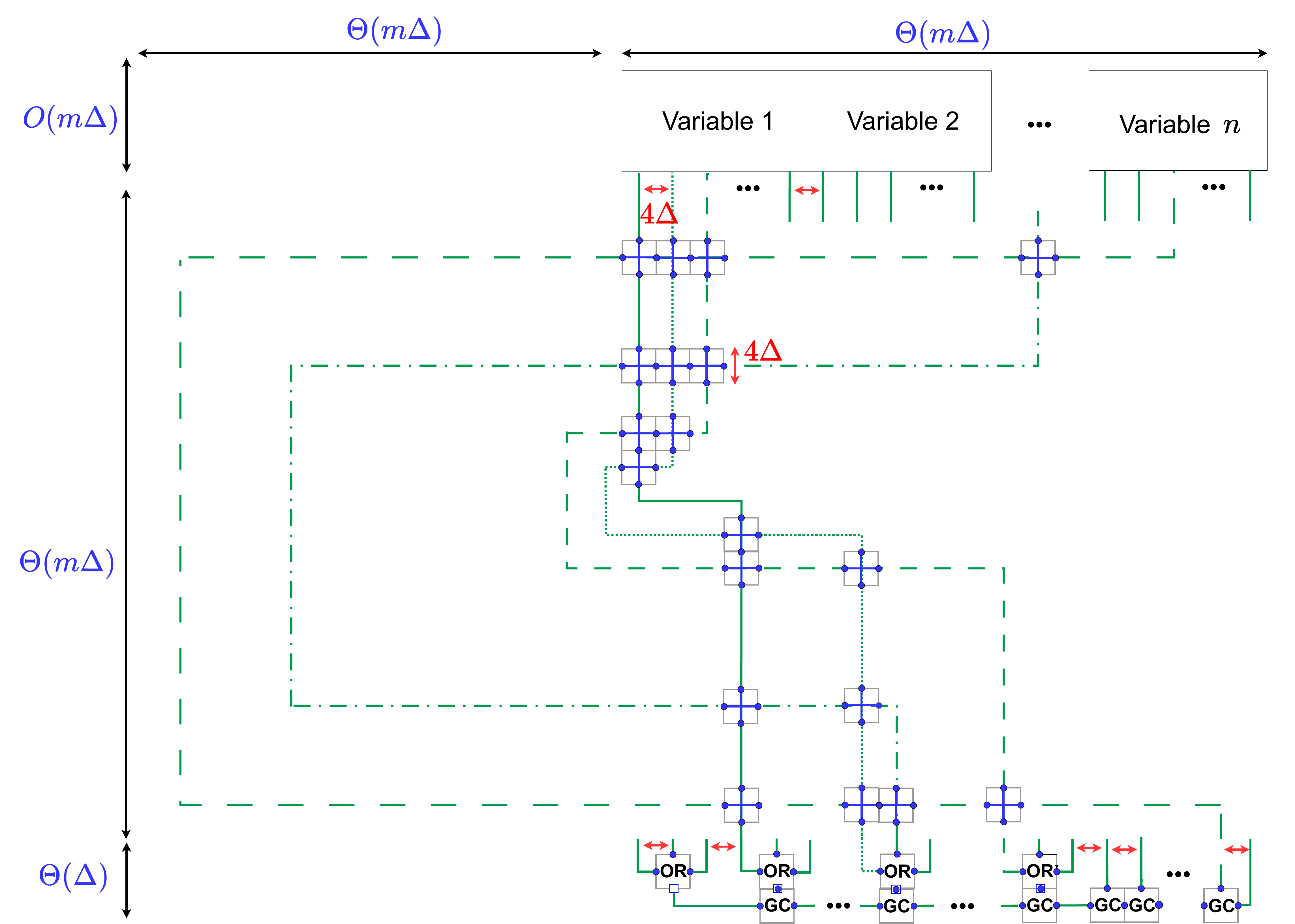}
\caption{\textbf{Closed composite gadget constructed from an instance of  3SAT}.}\label{layoutf}
\end{figure}

Since the parity of the garbage-collection layer is chosen so that the total number of primal defects in the composite gadget is even, there is a one-to-one correspondence between assignments of the variables of $\phi$ and configurations of $G$ in which every atomic gadget, except possibly the clause gadgets, is assigned a satisfying signature. 
Recall that, by definition, $G$ is satisfiable if it admits a configuration in which every atomic gadget is assigned a satisfying signature. It follows that $G$ is satisfiable if and only if $\phi$ is satisfiable.

The next lemma establishes the correctness of the reduction from 3SAT to the Minimum-Weight Join problem. It uses the Localization Lemma to prove the soundness direction for $\Delta=\Theta(m^2)$.

\begin{lemma}[Correctness of the Reduction from 3SAT to Minimum-Weight Join]\label{soundnessnp}
Let  $B$ be the  baseline cost, i.e., the sum, over all atomic gadgets $g$, of the minimum cost of a signature of $g$.
Let $W^*$ denote the minimum weight of a join for the defects in $G$.

For a sufficiently large choice of $\Delta=\Theta(m^2)$, we have
$W^* \ge B$,
with equality if and only if $\phi$ is satisfiable.
\end{lemma}
{\em Proof.}   
Let $C^*$ denote the minimum cost of a configuration of $G$. 
Since, for every atomic gadget, the satisfying signatures are precisely the minimum-cost signatures, we have
$
C^* \ge B$, 
with equality if and only if $\phi$ is satisfiable.

We establish the lemma by showing that, for a sufficiently large choice of $\Delta$, we have $C^*=W^*$. First, note that $W^*\leq C^*$, since the locally optimal join associated with a minimum-cost configuration of $G$ has weight equal to $C^*$. Thus, it suffices to prove that $W^*\geq C^*$. To establish this inequality, it is enough to show that every minimum-weight join $J^*$ for $G$ is local. Indeed, if $J^*$ is local, then its weight $W^*$ is at least the cost of its  configuration, which is in turn at least 
$C^*$, since the latter is the minimum cost of a configuration. Therefore, the locality of every minimum-weight join implies that 
$W^*\geq C^*$. To establish locality, we first derive an upper bound on $W^*$ and then invoke the Localization Lemma.

If $\varphi$ is satisfiable, then $C^*=B$. Otherwise, 
$C^*<B+m$ since the relative cost of a nonsatisfying signature of a clause gadget is~$1$, and  there exists an assignment satisfying at least one clause. We could have used 
the stronger lower bound of $m/2$ on 
 the number of satisfied clauses, but this simple bound is enough for the proof. It follows that, in both cases,  
$W^*< B+m$ since 
$W^*\leq C^*$.
 
Therefore, to guarantee that every minimum-weight join for $G$ is local, it suffices by 
the Localization Lemma
to choose
\[
R + \frac{2\Delta}{5} - 3\delta - 1 \geq B+m,
\]
that is,
\[
\Delta \geq \frac{5}{2}(b+m+3\delta+1),
\]
where $b=B-R$ is the sum, over all atomic gadgets $g$, of the minimum cost of a signature of the nucleus of $g$. 

We exclude the wire and garbage-collection gadgets from this sum, since their contribution is zero.
In total, there are $\Theta(m)$ atomic variable and  clause gadgets, and $\Theta(m^2)$ atomic crossing-wire   gadgets. For each of these gadgets, the minimum cost of a signature of its nucleus is $\Theta(1)$. Thus, $b=\Theta(m^2)$. 
On the other hand, $\delta=\Theta(1)$. Hence, setting $\Delta=\Theta(m^2)$ is sufficient to ensure that every minimum-weight join for the  defects in  $G$ is local. 
\finito

Recall that the \emph{Minimum-Weight Join} problem asks for a minimum-weight join for prescribed primal and dual defects in the domain, which is either the square lattice or the square lattice on a torus.

\NPHTheorem*
{\em Proof.}  Given a $3$-CNF formula $\phi$ with $m$ clauses, construct the closed composite gadget $G$ from $\phi$ as described above, and compute its baseline cost $B$. 
Note that the outer boxes of all atomic gadgets in $G$ are contained in a
$K\times K$ square, where $K=\Theta(m^3)$.
If the domain is a torus, we choose its side length to be $K$. 
 By Lemma~\ref{soundnessnp}, satisfiability of $\phi$ reduces to computing the minimum join weight $W^*$ for the defects of $G$ and comparing it with $B$. In particular, $\phi$ is satisfiable if and only if $W^*=B$. \finito

\section{Hardness of approximation}\label{hardapx}

We establish hardness of approximation for the Minimum-Weight Join problem using 
 H{\aa}stad's Inapproximability Theorem~\cite{Has01}, which asserts that, for every constant $0<\varepsilon<1/8$, it is NP-hard to distinguish between a satisfiable $3$-CNF formula $\phi$ and one for which every assignment leaves at least $(1/8-\varepsilon)m$ clauses unsatisfied. We first give an overview of the reduction and explain the modifications
to the NP-completeness reduction needed to obtain hardness of approximation.

Recall that, in the NP-hardness reduction, if $\phi$ is satisfiable, then the minimum cost $C^*$ of a configuration of the constructed closed composite gadget is equal to its baseline cost $B$, defined as the sum, over all atomic gadgets, of the minimum cost of a signature. 
If  $\phi$ is unsatisfiable, then $C^*$ is larger than $B$ by at least one. Indeed, equality would require every variable, crossing-wire, and clause gadget to be assigned a satisfying signature, which is impossible when the formula is unsatisfiable.

 However, the gap between $C^*$ and $B$ need not grow with the minimum number $U^*$ of clauses left unsatisfied by an assignment. Indeed, a nonsatisfying signature of a variable or crossing-wire gadget can break the correspondence between the Boolean assignment represented by the variable gadgets and the values reaching the clause gadgets, potentially allowing many additional clause gadgets to have satisfying signatures. 
 This can occur while incurring only a small increase in $C^*$ over $B$, since a nonsatisfying signature of a variable or crossing-wire gadget has relative cost one.

To overcome this issue, we replace the variable and crossing-wire gadgets by \emph{strong} versions in which every nonsatisfying signature has relative cost at least $m$.  The wire and garbage-collection gadgets require no modification, since all their signatures are satisfying. 
Consequently, violating the strong versions of the  variable and crossing-wire gadgets 
is at least as costly as leaving all $m$ clause gadgets unsatisfied, 
and hence the gap between $C^*$ and $B$ is exactly $U^*$, i.e.,  
$C^*-B=U^*$.  

This relation allows us to derive hardness of approximation for the Minimum-Weight Join problem from H{\aa}stad's theorem. 
Roughly, if an algorithm could approximate the minimum join weight  
within an additive error of $cm$ for some  constant
$c>0$, then,   
by choosing $\Delta$ sufficiently large, the Localization Lemma would allow us to obtain a $cm$-additive approximation to $C^*$. Since $C^*-B=U^*$, such an approximation   would distinguish between the case $U^*=0$ and the case $U^*\geq(1/8-\varepsilon)m$, contradicting H{\aa}stad's theorem,  
provided that $c<1/8-\varepsilon$.

We formalize the notion of strong composite gadgets in Section~\ref{scgs}. We then construct the strong variable and crossing-wire gadgets in Sections~\ref{strongVar} and~\ref{strongCorss}, respectively; the latter construction relies on the strong equality gadget, which we construct  in Section~\ref{strongEq}. Finally, in Section~\ref{mwjha}, we use these strong gadgets to establish hardness of approximation for the Minimum-Weight Join problem.

\subsection{Strong composite gadgets}\label{scgs}

To define the strength of a composite gadget, which measures the minimum cost of violating its intended logical behavior, we first extend the notions of signatures and relative costs from atomic gadgets to composite gadgets.

The \emph{signature} of a configuration of a composite gadget 
is the assignment induced on the outer pins of the composite gadget. A \emph{signature} of the composite gadget is any signature of one of its configurations. 
The \emph{cost} of a signature is the minimum cost of a configuration inducing that signature.

If the composite gadget is satisfiable, a signature is called \emph{satisfying} if it is induced by a satisfying configuration. 
All satisfying signatures of a satisfiable composite gadget have the same cost, since every satisfying configuration has cost equal to the sum of the minimum signature costs of its constituent atomic gadgets.

The \emph{relative cost} of a signature 
of a satisfiable composite gadget 
 is defined to be the difference between its cost and the cost of a satisfying signature.

The \emph{strength} of a satisfiable composite gadget is the minimum relative cost among all nonsatisfying signatures.

Thus, for an atomic gadget, its strength  is the minimum relative cost of a signature that violates the intended behavior of the gadget. 
The atomic clause, variable, and crossing-wire gadgets each has strength $1$, 
whereas the garbage-collection and wire gadgets have infinite strength.

We will construct (satisfiable) composite variable  and crossing-wire   gadgets of strength~$m$ by combining atomic gadgets. In these constructions, we set the parameter $\Omega$ of the underlying atomic variable gadgets to $\Omega=m_0+2$, where $m_0$ is the smallest integer greater than or equal to $m$ such that $m_0+2$ is divisible by $4$. 
The atomic variable gadget enforces two types of constraints: the \emph{equality constraint} $x=y$ and the \emph{negation constraint} requiring $x'=\bar{x}$ and $y'=\bar{y}$, where the gadget pins are labeled as shown in the first row of Table~\ref{sigtable}. Note that the constraint $y'=\bar{y}$ follows from $x'=\bar{x}$ because the gadget has even parity. 
With this choice of $\Omega$, the signatures of the atomic variable gadget and their relative costs are given in Table~\ref{sigtable}. 
In Table~\ref{sigtable} and throughout what follows, green
(respectively, blue) defects are associated with Boolean variables
assigned the value \emph{True} (respectively, \emph{False}).

{\renewcommand{\arraystretch}{1.4}
\begin{longtable}{|l|p{0.22\textwidth}|p{0.22\textwidth}|c|}
\caption{\textbf{Atomic variable gadget: signatures and their relative costs.}}
\label{sigtable}\\

\hline
Signatures &
Negation constraint &
Equality constraint &
Relative cost \\ \hline
\endfirsthead

\hline
Configurations &
Negation constraint &
Equality constraint &
Relative cost \\ \hline
\endhead

\tabfig{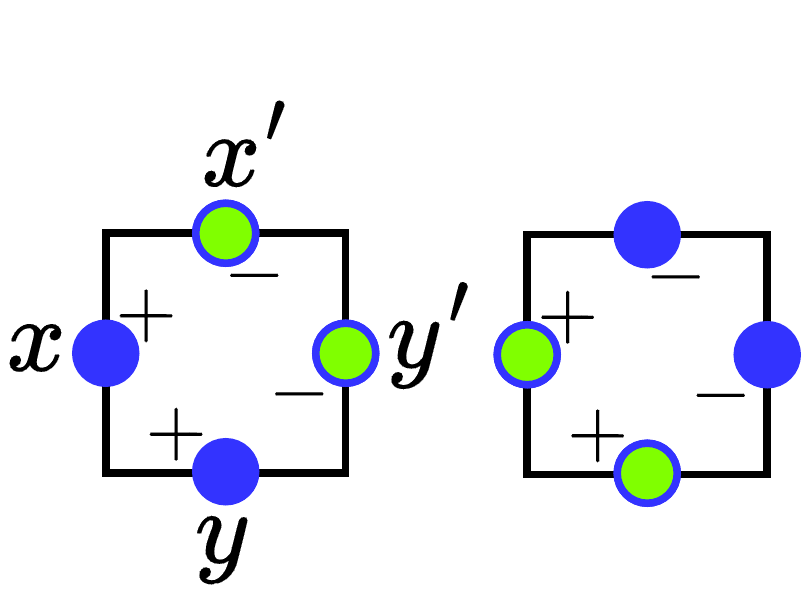}{0.14\textwidth}
&
Satisfied
&
Satisfied
&
0
\\ \hline

\tabfig{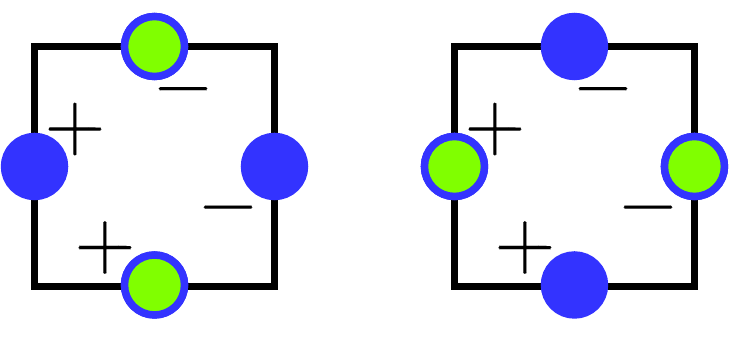}{0.14\textwidth}
&
Satisfied
&
Violated 
&
1
\\ \hline

\tabfig{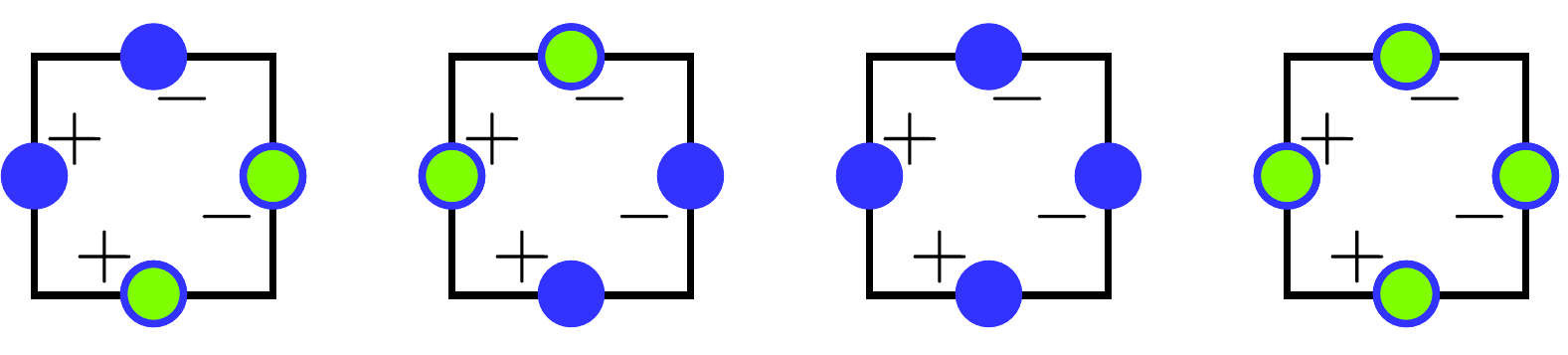}{0.32\textwidth}
&
Violated 
&

&
at least $m$
\\ \hline

\end{longtable}
}

\subsection{Strong variable gadget}\label{strongVar}

{
\renewcommand{\galscale}{0.2	}
\begin{figure}[!htbp]
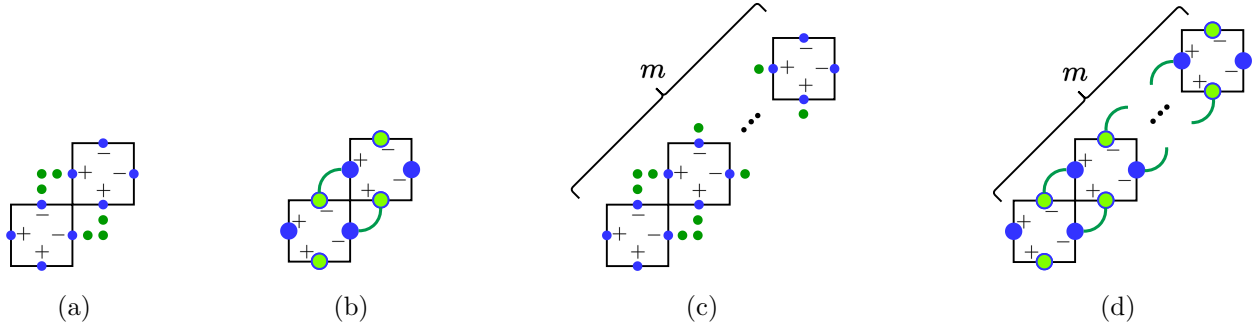
\centering
\sfg{Strong-var-gadget.pdf}\hfill
\sfg{Strong-var-gadget-just.pdf}\hfill
\sfg{strong-var+m.pdf}\hfill
\sfg{strong-var+m-just.pdf}

\caption{\textbf{Strong degree-$2$ variable gadgets.}
(a) Strength-$2$; (c) strength-$m$; (b) and (d) show the propagation of a signature violating the equality constraint.}
\label{strvg}
\end{figure}
}

To illustrate the idea, we begin with the degree-$2$ and strength-$2$ case. 
 A strength-$2$ variable gadget is obtained by connecting two atomic variable gadgets via wires  as shown in Figure~\ref{strvg}(a).

Consider any configuration of the gadget. If the signature of at least one atomic gadget violates the negation constraint, then its relative cost is at least $m$, and hence the total relative cost is at least $m$. Otherwise, all atomic gadgets satisfy the negation constraint. In this case, if one atomic gadget violates the equality constraint, then the connecting wires propagate the same violating signature to the other atomic gadget, as shown in Figure~\ref{strvg}(b). Thus, both atomic gadgets incur a relative cost of~$1$, giving a total relative cost of~$2$. Therefore, the strength of the gadget is~$2$.

The same construction extends to 
strength~$m$ for degree $2$. By connecting $m$ atomic variable gadgets as shown in Figure~\ref{strvg}(c), any violation of the equality constraint propagates to every atomic gadget, as illustrated in Figure~\ref{strvg}(d). Consequently, every nonsatisfying signature has relative cost at least~$m$, and the gadget has strength~$m$.

{
\renewcommand{\galscale}{0.16}
\begin{figure}[!htbp]
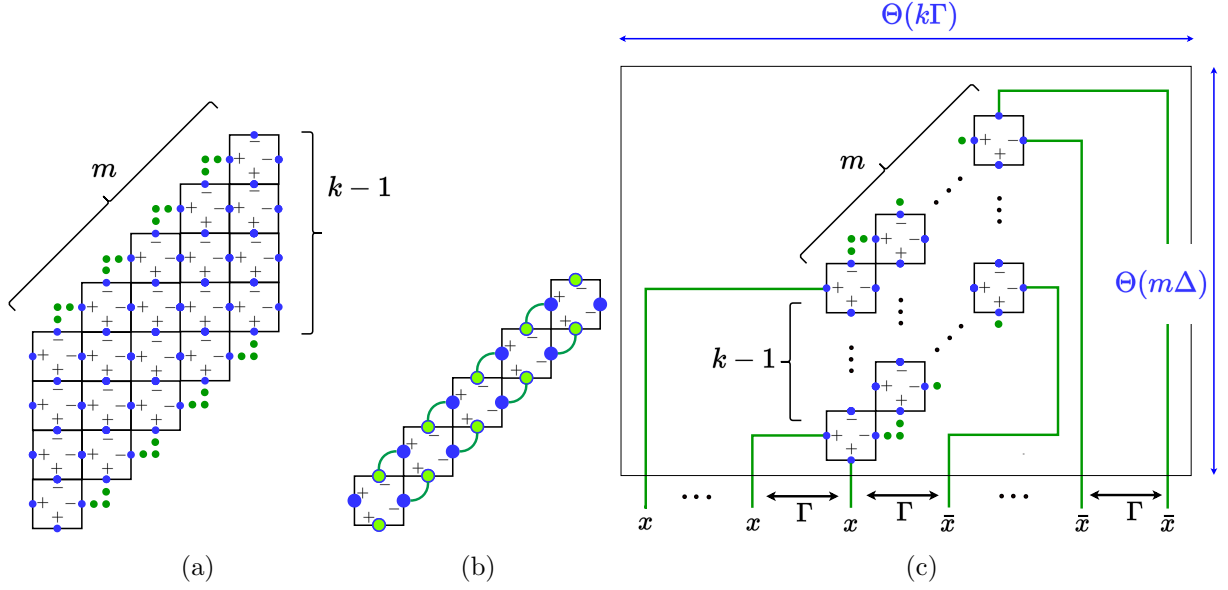
\centering
\sfg{degree-k-strong-var+m.pdf}
\hspace{-2.5em}
\sfg{degree-k-strong-var+m-just.pdf}
\hspace{-0.5em}
\sfg{degree-k-strong-var+m-boxed.pdf}
\caption{\textbf{Degree-$k$ strong variable gadget.}
(a) Degree-$k$ variable gadget of strength~$m$;
(b) diagonal propagation of a signature violating the equality constraint;
(c) wired  gadget with outer pins spaced by a parameter~$\Gamma$.}
\label{strvgk}
\end{figure}
}

To construct a degree-$k$ variable gadget of strength~$m$, for $k\ge3$, we arrange the atomic gadgets in $k-1$ layers, as shown in Figure~\ref{strvgk}(a).

Consider any configuration of the gadget. If at least one atomic gadget violates the negation constraint, then the total relative cost is at least~$m$. Otherwise, every atomic gadget satisfies the negation constraint. If one atomic gadget violates the equality constraint, then this violation propagates along the diagonal containing that gadget, as shown in Figure~\ref{strvgk}(b). The propagation follows from the boundary wires and from the atomic gadgets on the upper and lower diagonals, which themselves act as wires because they satisfy the negation constraint. Consequently, every atomic gadget on that diagonal violates the equality constraint, giving a total relative cost of at least~$m$. Therefore, the degree-$k$ gadget has strength~$m$.

\subsection{Strong equality gadget}\label{strongEq}

The construction of the strong crossing-wire   gadget uses a degree-$k$ equality gadget of strength~$m$. In this section, we show how to construct such a gadget for any even $k\geq 2$.

We begin with the atomic case, shown in Figure~\ref{streq2}. It has degree~$2$ and infinite strength.

{
\renewcommand{\galscale}{0.23}
\begin{figure}[!htbp]
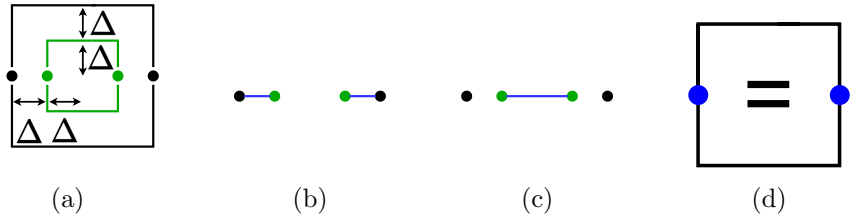
\centering
\sfg{boxed-eq2.pdf}\hfilli\hfilli\hfilli\hfilli
\sfg{eq2-1.pdf}\hfilli\hfilli\hfilli\hfilli
\sfg{eq2-2.pdf}\hfilli\hfilli\hfilli\hfilli
\sfg{type-eq-sized.pdf}
\caption{
\textbf{Degree-$2$ equality gadget.}
The atomic equality gadget is shown in~(a). Its nucleus is empty.
Parts~(b) and~(c) show its two signatures together with their corresponding minimum-weight joins.
Both signatures have cost $2\Delta$.
Part~(d) shows the gadget's symbol.
}\label{streq2}
\end{figure}
}

To construct a degree-$k$ equality gadget of strength~$m$ for even $k\geq 4$, we use a degree-$k/2$ variable gadget of strength~$m$ and connect each of its $k/2$ outer pins corresponding to the negated Boolean variable to an atomic equality gadget, as shown in Figure~\ref{streqm}.

The resulting gadget has two satisfying signatures: one in which all $k$ outer pins are assigned \emph{True}, and one in which they are all assigned \emph{False}. Since the atomic equality gadgets have infinite strength, the resulting gadget has the same strength as the underlying variable gadget.

\begin{figure}[!htbp]
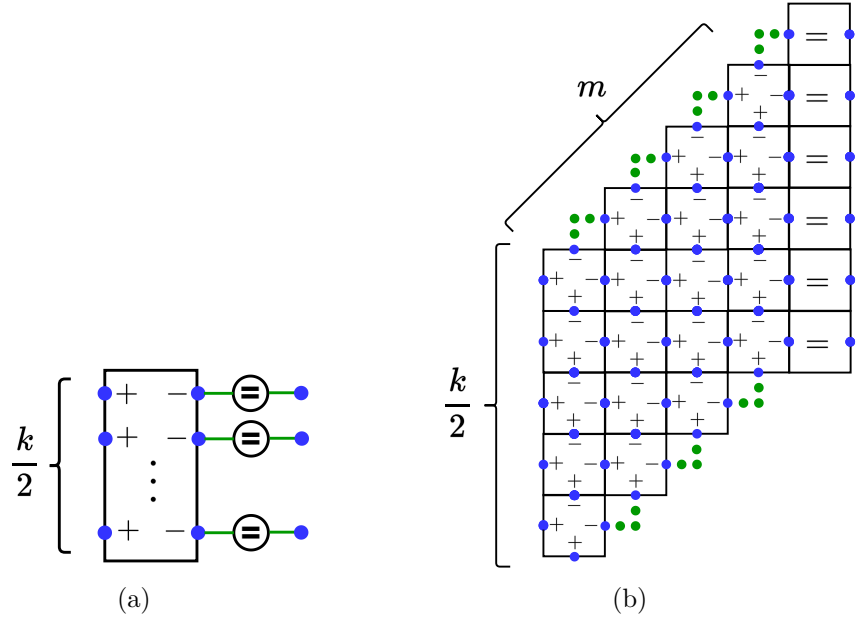
\centering
{
\renewcommand{\galscale}{0.3}
\sfg{strong-eq-gadget-diag.pdf}
}\hfilli\hfilli\hfilli
{
\renewcommand{\galscale}{0.2}
\sfg{degree-k-strong-eq.pdf}
}
\caption{\textbf{Degree-$k$ strong equality gadget, for even $k\geq 4$.} 
(a) Schematic diagram; (b) construction.  
}\label{streqm}
\end{figure}

\subsection{Strong  crossing-wire  gadget}
\label{strongCorss}

In this section, we construct a strong crossing-wire gadget of strength~$m$ from $m'$ atomic crossing-wire gadgets and two degree-$(m'+1)$ equality gadgets of strength~$m$, where
$m'=m$ if $m$ is odd and $m'=m+1$ if $m$ is even.

We begin with the case $m=3$. Consider the composite gadget obtained by  connecting three atomic crossing-wire gadgets and two degree-$4$ equality gadgets of strength~$3$, as shown in Figure~\ref{scf1}(a).

{
\renewcommand{\galscale}{0.2}
\begin{figure}[!htbp]
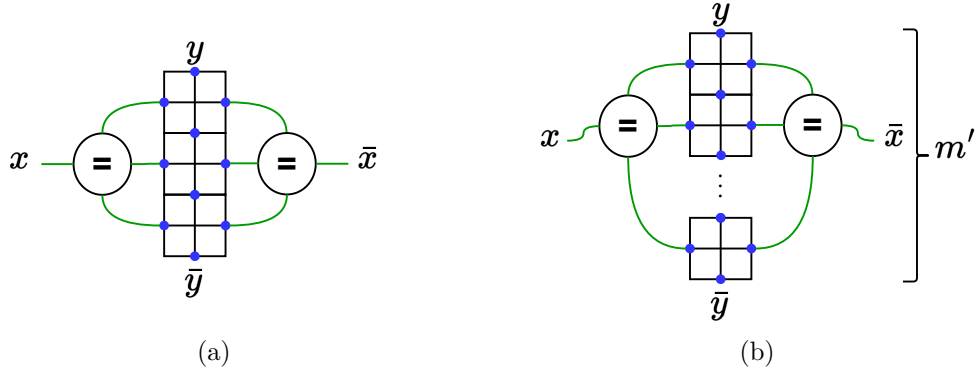
\centering
\sfg{Strong-crossing-gadget.pdf}\hfilli\hfilli\hfilli
\sfg{Strong-crossing-gadget-m.pdf}
\caption{\textbf{Schematic diagrams of a strong crossing-wire gadget.} (a) strength-$3$ gadget; (b) strength-$m$ gadget.}
\label{scf1}
\end{figure}
}

Table~\ref{crstr3table} shows all configurations of the composite gadget that satisfy the equality gadgets together with their signatures.
The first four configurations satisfy the composite gadget, whereas the last four do not. In each of the latter, all three atomic crossing-wire gadgets are unsatisfied, each contributing a relative cost of~$1$, for a total relative cost of~$3$.

Since the composite gadget is an even-parity gadget with $4$ pins,
the resulting signatures are precisely its $8$ signatures. While 
each  satisfying signature of the composite gadget  is induced by a unique configuration, this is not the case  for the nonsatisfying signatures, as they are also  induced by configurations violating one of the equality gadgets. However, every such configuration has relative cost at least~$3$, because the equality gadgets have strength~$3$. Hence, every nonsatisfying signature of the  composite gadget has relative cost~$3$, and therefore the resulting gadget operates as a  crossing-wire gadget and has strength~$3$.

\begin{table}[!htbp]
\centering
\hspace{-2em}
\begin{minipage}{0.49\textwidth}
\centering
{\setlength{\tabcolsep}{8pt}\renewcommand{\arraystretch}{0.5}
\begin{tabular}{|
>{\arraybackslash}m{0.73\textwidth}|
}
\hline
Satisfying  configurations and \\ their signatures\\~\\~ \\ \hline
 \noalign{\vskip 1pt}
\tabfig{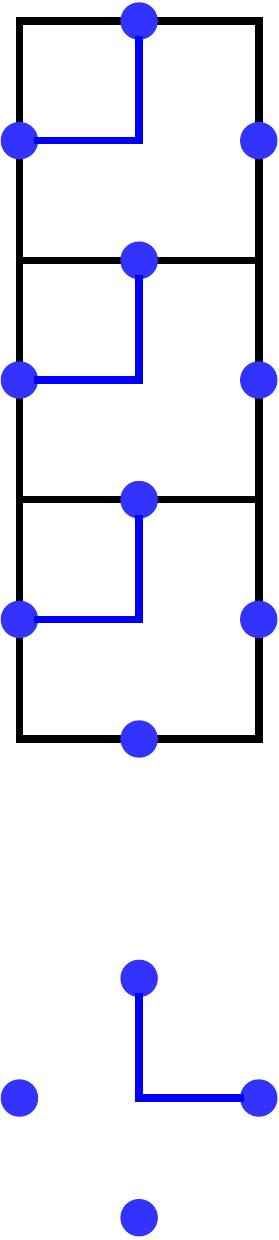}{0.13\textwidth}\hspace{8pt}
\tabfig{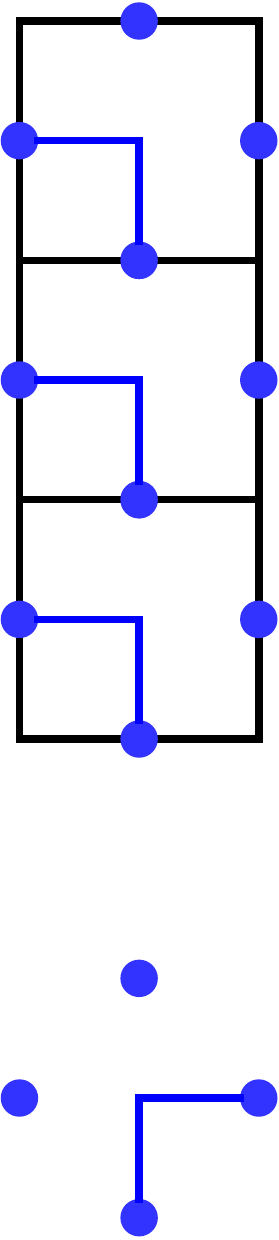}{0.13\textwidth}\hspace{8pt}
\tabfig{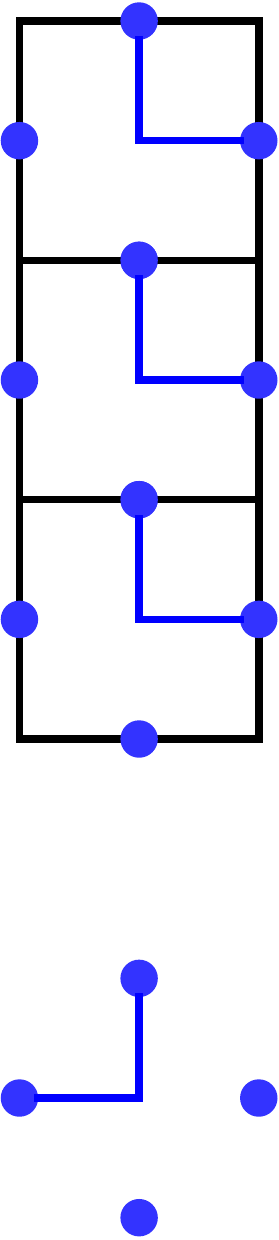}{0.13\textwidth}\hspace{8pt}
\tabfig{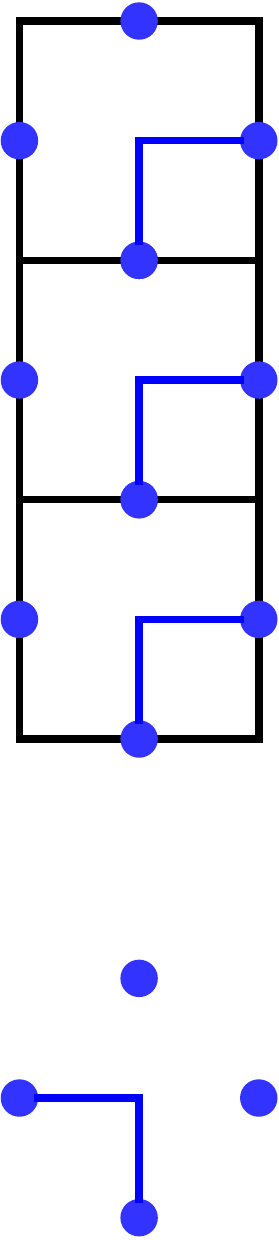}{0.13\textwidth}
\\ \hline
\end{tabular}}
\end{minipage}
\hfill
\begin{minipage}{0.49\textwidth}
\centering
\hspace{-2em}
{\setlength{\tabcolsep}{8pt}\renewcommand{\arraystretch}{0.5}
\begin{tabular}{|
>{\arraybackslash}m{0.73\textwidth}|
>{\centering\arraybackslash}m{0.13\textwidth}|}
\hline
Nonsatisfying configurations that satisfy the equality gadgets, and their signatures
& Relative cost \\ \hline
 \noalign{\vskip 1pt}
\tabfig{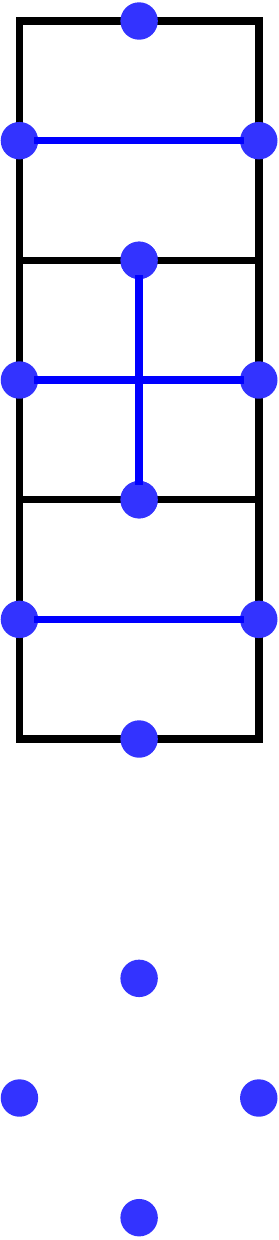}{0.13\textwidth}\hspace{8pt}
\tabfig{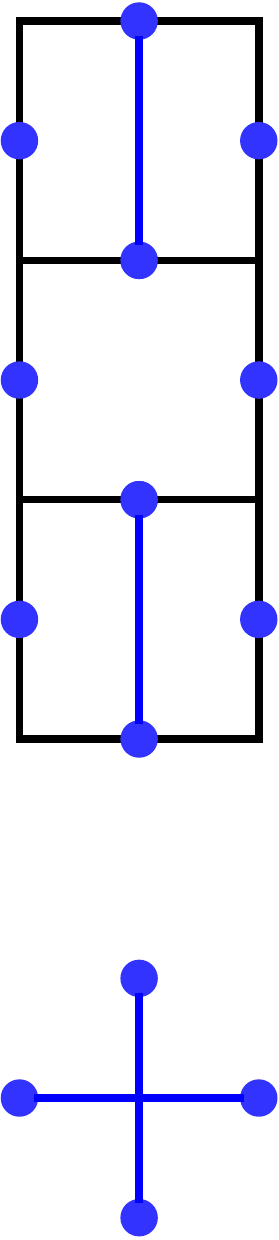}{0.13\textwidth}\hspace{8pt}
\tabfig{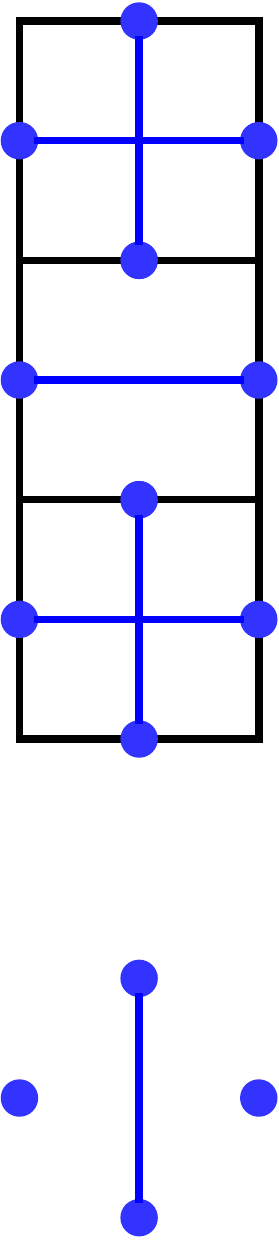}{0.13\textwidth}\hspace{8pt}
\tabfig{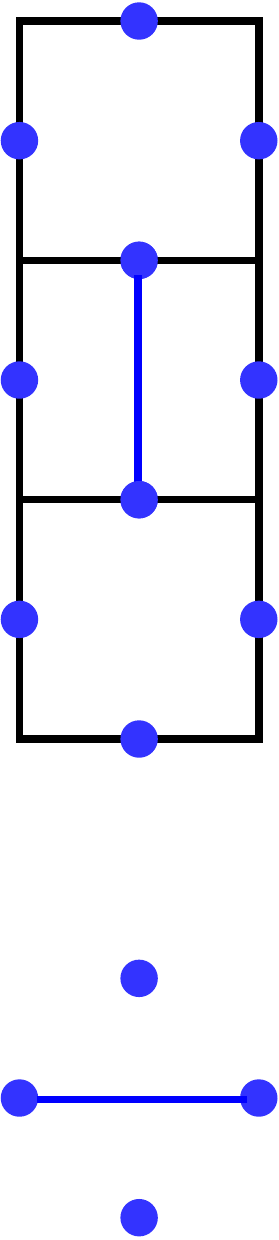}{0.13\textwidth}
& $3$\\ \hline
\end{tabular}}
\end{minipage}
\caption{\textbf{Configurations satisfying the equality gadgets and their signatures.}
The first row shows all configurations of the composite gadget satisfying the equality gadgets, and the second row shows their corresponding signatures.
For each atomic crossing-wire gadget and each of its outer pins, we indicate that the associated Boolean variable is assigned \emph{False} by drawing a line from the pin to the center  of the gadget. The same convention is used for the signatures of the composite gadget.
}
\label{crstr3table}
\end{table}

The same construction extends to strength~$m$, as shown in Figure~\ref{scf1}(b). To see why the resulting composite gadget operates as a crossing-wire   gadget of strength~$m$, consider an arbitrary configuration.

If the configuration violates one of the equality gadgets, then its relative cost is at least~$m$, since each equality gadget has strength~$m$. Therefore, it suffices to consider configurations satisfying both equality gadgets. There are exactly $8$ such configurations, determined by the two signatures of each of the equality gadgets together with the values assigned to the two  free outer pins of the first and last atomic crossing-wire   gadgets.

If the left and right pins of the atomic  crossing-wire   gadgets are assigned opposite truth values, then the two free pins must also be assigned opposite truth values, since the composite gadget has even parity (each constituent gadget has even parity, and an even number of outer pins are glued together). These are precisely the four satisfying configurations of the composite gadget, and their signatures realize the intended logic of a  crossing-wire    gadget. Moreover, each of these signatures uniquely determines its corresponding configuration.

If, on the other hand, the left and right pins of the atomic crossing-wire   gadgets are assigned the same truth value, then every atomic crossing-wire   gadget is unsatisfied and therefore incurs a relative cost of~$1$. Hence, the total relative cost is~$m'\geq m$. 

It follows that the constructed composite gadget is a crossing-wire   gadget of strength~$m$.

Figure~\ref{sclayout} shows the layout, dimensions, and symbol of the gadget.

\begin{figure}[!htbp]
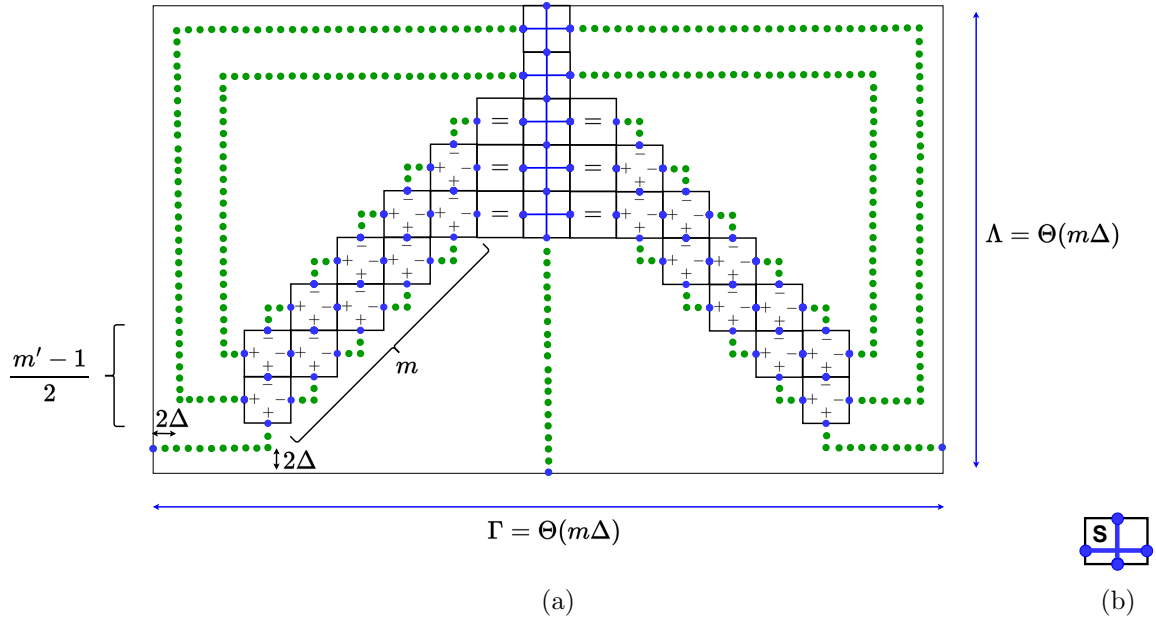
\centering
\sfg{Strong-crossing-gadget-construction.pdf}
\hspace{-3em}
{
\renewcommand{\galscale}{0.4}
\sfg{type-crossing-strong.pdf}
}
\caption{\textbf{Strong crossing-wire   gadget.}
The layout of the gadget is shown in~(a). It has dimensions
$\Lambda \times \Gamma$, where
$\Lambda=\Theta(m\Delta)$ and
$\Gamma=\Theta(m\Delta)$.
The gadget's symbol is shown in~(b). 
}
\label{sclayout}
\end{figure}

\subsection{Hardness of approximation of Minimum-Weight Join}\label{mwjha}

In this section, we prove that the Minimum-Weight Join problem is hard to approximate within an additive error of $\Omega(K^{1/7})$, where $K$ denotes the lattice side length, assuming that $P\neq NP$. As in the NP-hardness reduction, we begin by constructing a closed composite gadget $G$ from a given $3$-CNF formula $\phi$ with $m$ clauses. The construction is identical to that of Section~\ref{redS}, except that the variable and crossing-wire gadgets are replaced by their strong counterparts.  We then prove two key lemmas. The Strong Gadget Lemma (Lemma~\ref{strlem}) shows that every sufficiently cheap configuration must satisfy all strong composite gadgets, while the Gap Lemma (Lemma~\ref{gaplem}) combines the Strong Gadget Lemma with the Localization Lemma to show that, for $\Delta = \Theta(m^5)$, the gap between the minimum weight of a join and the baseline cost is exactly the minimum number of unsatisfied clauses.
Finally, in Theorem~\ref{mainthMWJ}, we use the Gap Lemma to reduce Gap-$3$SAT  to the Minimum-Weight Join problem and conclude from H{\aa}stad's Inapproximability Theorem that the latter problem is hard to approximate. The soundness of this reduction relies on the equality established in the Gap Lemma, which in turn follows from the Localization Lemma.

We construct from $\phi$ a closed composite gadget $G$ using strength-$m$ composite variable gadgets and strength-$m$ crossing-wire gadgets, as shown in Figure~\ref{layoutfS}. 
The layout is identical to that of Section~\ref{redS}, except that each atomic crossing-wire gadget is replaced by its strength-$m$ counterpart, and the high-degree variable gadget in Figure~\ref{degkvarg}(c) is replaced by its strength-$m$ counterpart in Figure~\ref{strvgk}(c). The wires of the latter gadget are spaced by the width 
$
\Gamma=\Theta(m\Delta)
$ 
of the strong crossing-wire gadget. Instead of increasing the depth of the construction by $8\Delta$, each wire now increases it by exactly $2\Lambda$, where
$
\Lambda=\Theta(m\Delta)
$ 
is the depth of the strong crossing-wire gadget. 
The resulting composite gadget $G$ has dimensions
$
O(m^2\Delta)\times O(m^2\Delta), 
$
as shown in the figure.

\begin{figure}[!htbp]\centering
\includegraphics[width=\textwidth]{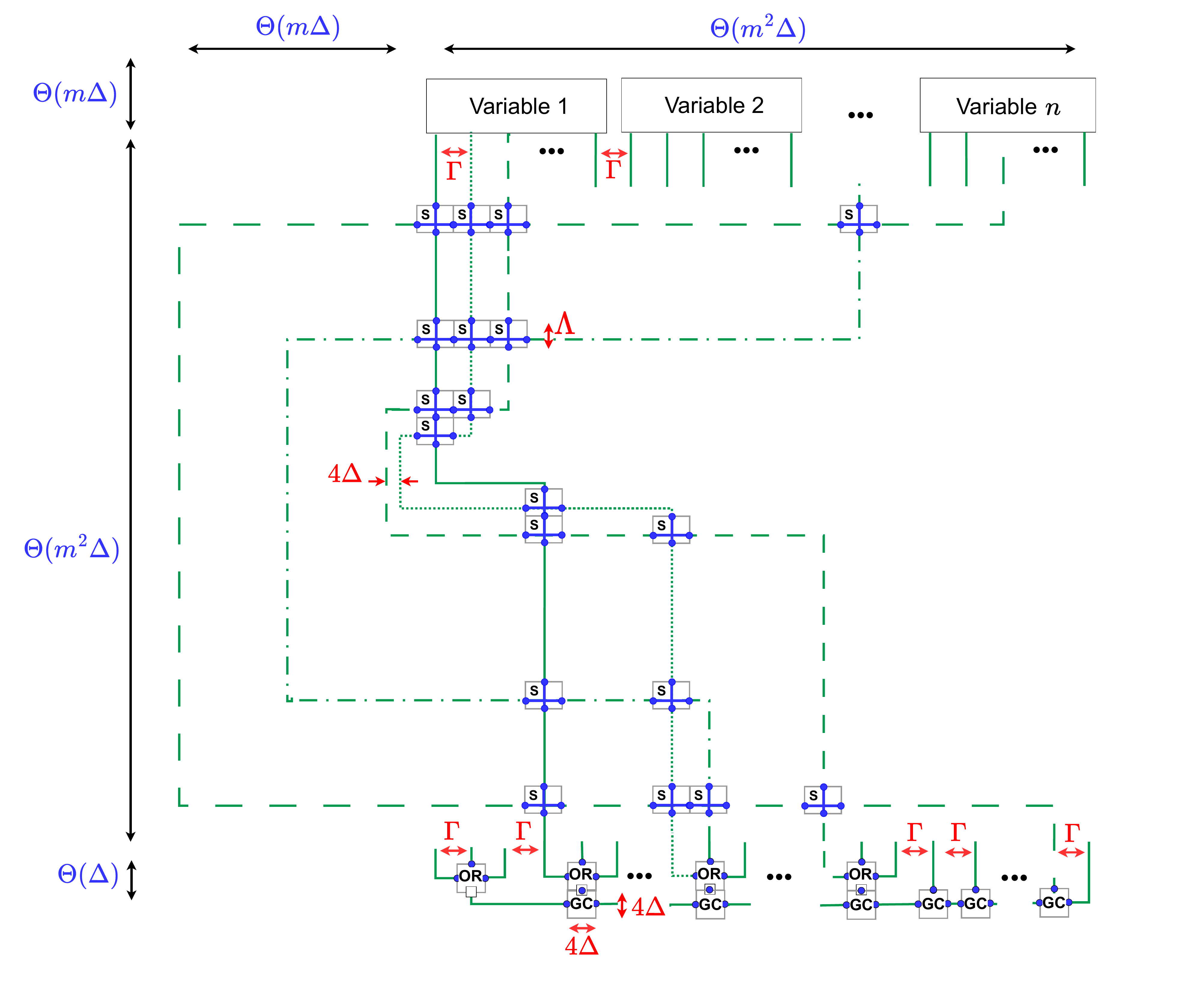}
\caption{\textbf{Closed composite gadget constructed from an instance of Gap-3SAT}.}
\label{layoutfS}
\end{figure}

As in the NP-hardness reduction, the satisfiability of $G$ is equivalent to that of $\phi$.

The following lemma shows that cheap configurations are forced to satisfy all strong composite gadgets. 
\begin{lemma}[Strong Gadget Lemma]\label{strlem}
Let $B$ be the baseline cost, i.e., the sum, over all atomic gadgets $g$ in $G$, of the minimum cost of a signature of $g$.
Consider any configuration $\chi$ of $G$. If the cost of $\chi$ is less than $B+m$, then every atomic gadget in $G$, except possibly some clause gadgets, is assigned a satisfying signature in $\chi$.
\end{lemma}
{\em Proof.}
Let $\mathcal L$ denote the set of clause gadgets in $G$, and let $\mathcal S$ denote the set of strong variable and crossing-wire gadgets. For each gadget $g\in\mathcal L\cup\mathcal S$, let $\chi_g$ denote the
signature of $g$ induced by $\chi$.

Since all signatures of the wire and garbage-collection gadgets are satisfying, the cost of $\chi$ decomposes into the baseline cost and the relative costs of the clause and strong gadgets:
\[
\operatorname{cost}(\chi)
=
B+
\sum_{g\in\mathcal L\cup\mathcal S}
\operatorname{relative\mbox{-}cost}(\chi_g).
\]
Therefore,
\begin{equation}\label{abeq}
\sum_{g\in\mathcal L\cup\mathcal S}
\operatorname{relative\mbox{-}cost}(\chi_g)
<m.
\end{equation}
Assume that some atomic gadget $a$ is assigned a nonsatisfying signature in $\chi$. Suppose that $a$ belongs to a strong composite variable or crossing-wire gadget $s$. Then $\chi_s$ is nonsatisfying. Since $s$ has strength $m$, the relative cost of $\chi_s$ is at least $m$, contradicting \eqref{abeq}. Therefore, $a$ must be a clause gadget.
\finito

The next lemma uses the Strong Gadget Lemma and the Localization Lemma to show that, for $\Delta=\Theta(m^5)$, the gap between the minimum weight of a join  and the baseline cost is exactly the minimum number of unsatisfied clauses.  
Compared to the NP-hardness reduction, the Localization Lemma now requires a larger value of $\Delta$, namely 
$\Theta(m^5)$ 
 instead of $\Theta(m^2)$, because the strong variable and crossing-wire gadgets substantially increase the total number of atomic gadgets in the construction.
\begin{lemma}[Gap Lemma]\label{gaplem} 
Let: 
\begin{itemize}
\item  $B$ be the  baseline cost, i.e., the sum, over all atomic gadgets $g$, of the minimum cost of a signature of $g$;
\item  $W^*$ be  the minimum weight of a join for the defects in $G$; and 
\item $U^*$ be the minimum number of unsatisfied clauses of $\phi$ over all assignments of $\phi$.
\end{itemize}
For a sufficiently large choice of $\Delta=\Theta(m^5)$, we have  
$$
W^*-B=U^*.
$$
Moreover, for this choice of $\Delta$, 
 $B=\Theta(m^9)$ and 
 all outer boxes of atomic
gadgets in $G$ 
are contained in a
$K\times K$ square, where $K=\Theta(m^7)$.
\end{lemma} 
{\em Proof.}
Let $C^*$ denote the minimum cost of a configuration of $G$. First, we use Lemma~\ref{strlem} to show that
\begin{equation}\label{lock}
C^*=B+U^*.
\end{equation}
We then argue, using the Localization Lemma and following the proof of Lemma~\ref{soundnessnp},  that choosing $\Delta$ sufficiently large guarantees that $C^*=W^*$.
Note that  we know that 
$W^*\leq C^*$ since  the locally optimal join associated with a minimum-cost configuration has weight $C^*$. 
It therefore suffices to show that 
\begin{equation}\label{eineq} 
W^*\ge C^*.
\end{equation}
Each assignment of the variables of $\phi$ gives rise to a configuration of $G$ 
in which every atomic gadget, except possibly some of the  clause gadgets, is assigned a satisfying signature.

Since every nonsatisfying signature of a clause gadget has relative cost $1$, the cost of the resulting configuration is 
$B+U$, 
where $U$ is the number of clauses left unsatisfied by the assignment.

Since every $3$-CNF formula admits an assignment satisfying at least one clause, there exists a configuration $\chi$ with 
$\operatorname{cost}(\chi)<B+m$.

 Let $\chi^*$ be a minimum cost 
 configuration. Thus   
$\operatorname{cost}(\chi^*) \le
\operatorname{cost}(\chi) < B+m$. 
  It follows from Lemma~\ref{strlem} that 
every atomic gadget in $G$, except possibly some clause gadgets, is assigned a satisfying
signature in $\chi^*$. 
The signatures of the strong variable gadgets in $\chi^*$ therefore determine an assignment $x$ of the variables of $\phi$.
Hence  
$C^*
=
\operatorname{cost}(\chi^*)
=
B+U$, 
where $U$ is the number of clauses left unsatisfied by $x$.
Since $\chi^*$ has minimum cost, we must have
$U=U^*$,
proving~\eqref{lock}.

To establish \eqref{eineq},  
let $J^*$ be a minimum-weight join for the defects in $G$. As in the proof of Lemma~\ref{soundnessnp}, it suffices to show that $J^*$ is local. Indeed, if $J^*$ is local, then its weight is at least the cost of its  configuration, which is in turn at least $C^*$, since $C^*$ is the minimum cost of a configuration.

By the Localization Lemma, to guarantee that  
$J^*$  is local 
it is enough  
to choose
\[
R+\frac{2\Delta}{5}-3\delta-1\ge B+m,
\]
since the weight of $J^*$ is 
$$W^*\leq C^* = B+U^*< B+m,$$
where  the last inequality follows from $	U^*< m$.

That is, we need
\begin{equation}\label{plank}
\Delta\ge\frac{5}{2}(b+m+3\delta+1),
\end{equation}
where $b=B-R$ is the sum, over all atomic gadgets $g$, of the minimum cost of a signature of the nucleus of $g$. We exclude the wire and garbage-collection gadgets from this sum, since their contribution is zero.

The total number of atomic gadgets in the composite variable gadgets is
$\Theta(m^2)$. Each composite crossing-wire   gadget contains $\Theta(m^2)$ atomic gadgets, and there are $\Theta(m^2)$ composite crossing-wire   gadgets in total. Additionally, there are $\Theta(m)$ atomic clause gadgets. 
For each atomic gadget, the minimum cost of a nucleus signature is $\Theta(1)$, except for  atomic variable gadgets, for which it is $\Theta(\Omega)=\Theta(m)$. 
Hence, the contribution of the composite crossing-wire   gadgets to $b$ is $\Theta(m^2\times m^2\times m)=\Theta(m^5)$. This dominates all remaining contributions. Therefore, $b=\Theta(m^5)$. 
On the other hand, $\delta=\Theta(\Omega)=\Theta(m)$, due to the atomic variable gadget. 
Hence, setting $\Delta=\Theta(m^5)$ is sufficient to guarantee~\eqref{plank}. 

Accordingly, since the dimensions of $G$ are
$O(m^2\Delta)\times O(m^2\Delta)$, all outer
boxes of the atomic gadgets are contained in a $K\times K$ square for
some $K=\Theta(m^7)$.

Finally, the relay cost 
 $R=\Theta(\Delta t)$, where $t=\Theta(m^4)$ is the total number of atomic gadgets. Hence, 
$B=b+R=
\Theta(m^9)$. 
\finito

Using H{\aa}stad's Inapproximability Theorem (Theorem~\ref{hastadinapxthm}) together with the Gap Lemma, we show in Theorem~\ref{mainthMWJ}  that
Minimum-Weight Join is hard to approximate within an additive error
of $cK^{1/7}$ for some constant $c>0$.

\begin{definition}[Gap-3SAT]\label{gap3sat}
For a constant $0< \rho < 1$, the problem
Gap-$3$SAT$[1,\rho]$ is the following promise problem.

\begin{itemize}
\item \emph{Input:} A $3$-CNF formula $\phi$ with $m$ clauses.
\item \emph{YES instance:} $\phi$ is satisfiable.
\item \emph{NO instance:} Every assignment satisfies at most $\rho m$ clauses.
\end{itemize}
That is, given the input formula, the goal is to distinguish between the YES and NO instances. No guarantee is made for instances that satisfy neither condition.
\end{definition}

\begin{theorem}[H{\aa}stad's Inapproximability Theorem~\cite{Has01}]\label{hastadinapxthm}
For every constant $0< \varepsilon<1/8$, 
 Gap-$3$SAT $[1,7/8+\varepsilon]$
is NP-hard.
\end{theorem}

\HAMWJ*

{\em Proof.} 
Fix any constant  $0< \varepsilon<1/8$.
Assume, for the sake of contradiction, that there is a polynomial-time algorithm $\mathcal A$ for the Minimum-Weight Join problem that is guaranteed to return a join of weight at most
$W^*+cK^{1/7}$.
The constant $c$ will be chosen later.

Using $\mathcal A$, we construct the following polynomial-time algorithm $\mathcal B$ for Gap-$3$SAT$[1,7/8+\varepsilon]$. Given a $3$-CNF formula $\phi$ with $m$ clauses: 
\begin{enumerate}
\item Construct the closed composite gadget $G$ from $\phi$, as described above, and compute its baseline cost $B$.
\item Let $K=\Theta(m^7)$ be an integer  such that the outer boxes of all atomic gadgets in $G$ are contained in a $K\times K$ square. If the domain is a torus, set its side length to $K$. 
\item Run $\mathcal A$ on the defects of $G$ to obtain a join $J$.
\item If the weight of $J$ is at most $B+cK^{1/7}$, return YES; otherwise return NO.
\end{enumerate} 
Let  $U^*$ be the 
 minimum number of unsatisfied clauses of $\phi$.
 
 If $\phi$ is satisfiable, then $U^*=0$. Hence, by the Gap Lemma, $W^*=B$.   
Therefore, $\mathcal A$ returns a join of weight at most 
$B+cK^{1/7}$,
and so $\mathcal B$ outputs YES.

Suppose now that every assignment satisfies at most $(7/8+\varepsilon)m$ clauses.
Then $U^*\geq (1/8-\varepsilon)m$.  By the Gap Lemma, $W^*=B+U^*$,  and therefore 
 $W^*\ge B+(1/8-\varepsilon)m$. 
 Consequently, every join returned by $\mathcal A$ has weight at least $W^*$, and hence greater than $B+cK^{1/7}$,  provided that 
\begin{equation}\label{conde}
cK^{1/7}
<
(1/8-\varepsilon)m.
\end{equation}
Under this condition, $\mathcal B$  returns NO. 

Since $K=\Theta(m^7)$, we have $K^{1/7}=\Theta(m)$.  Therefore, choosing $c>0$ sufficiently small guarantees~\eqref{conde}.

Thus $\mathcal B$ solves Gap-$3$SAT$[1,7/8+\varepsilon]$ in polynomial time, contradicting H{\aa}stad's Inapproximability Theorem unless $P=NP$. 
\finito

\begin{remark}
\emph{
Our hardness result is formulated in terms of additive error rather than the
more common multiplicative approximation ratio. Additive error is particularly
natural here when approximation guarantees are expressed as functions of the
lattice side length. Indeed, the same decoding instance can be embedded into
arbitrarily larger lattices without changing the optimization problem or its
optimum value. Such padding can artificially strengthen a lattice-side-length-dependent
multiplicative inapproximability bound, whereas it weakens a lattice-side-length-dependent
additive bound.}
\end{remark}

\section{Consequences for topological-code decoding}\label{ccodes}

This section derives inapproximability results for Minimum-Weight Decoding of the toric and planar surface codes, and for Separate Minimum-Weight Decoding of the $4.8.8$ color code on the torus 
 from the corresponding inapproximability result for the Minimum-Weight Join problem established in the previous section.

\subsection{Toric code}
The Minimum-Weight Decoding problem for the $[[2L^2,2,L]]$ toric code is essentially the Minimum-Weight Join problem on the torus of side length $L$. The only difference is that, while errors in the toric code may contain cycles, we require joins to be acyclic for technical convenience. 
This is not restrictive, since cycles can be repeatedly removed from
an error in polynomial time.  Doing so preserves its syndrome, does not increase its weight, and yields a join.  This also shows that the minimum error weight is equal to the minimum join weight. 
Hence, any approximation algorithm for the Minimum-Weight Decoding problem yields an approximation algorithm for the Minimum-Weight Join problem. Since the toric code has $N=2L^2$ qubits, Theorem~\ref{mainthMWJ} immediately implies the following.

\HAToric*

\subsection{Planar surface code}\label{planarS}

In the $[[2L^2-2L+1,1,L]]$ planar surface code, acyclic errors are captured by \emph{relative joins}, which are defined analogously to joins except that paths are allowed to connect a defect to a boundary vertex. In addition, the sets of primal and dual defects need not have even cardinality.  More precisely, the \emph{domain} is now the lattice of the planar surface code and its dual, as shown in Figure~\ref{tpe}(b). The \emph{boundary vertices} are shown as squares in the figure. Defects and paths now lie in this domain.

A \emph{relative primal join} with respect to a prescribed set of primal defects is a collection of edge-disjoint primal paths, each connecting either two defects or a defect and a primal boundary vertex, whose union forms an acyclic graph and induces odd degree at each defect. A \emph{relative dual join} is defined analogously with respect to a prescribed set of dual defects. A \emph{relative join}, with respect to both the primal and dual defects, consists of a relative primal join together with a relative dual join. The \emph{weight} of a relative join is the total number of edges in its primal and dual joins, counting each primal--dual edge pair only once.

To show that this problem is hard to approximate, we adjust the proof of Theorem~\ref{mainthMWJ} by embedding the closed composite gadget $G$ constructed in Section~\ref{mwjha} into the lattice of the planar surface code while ensuring that all defects of $G$ are sufficiently far from the boundary vertices of the lattice. 
This guarantees that neither a minimum-weight relative join nor an approximately minimum-weight relative join uses paths connecting defects to boundary vertices. Consequently, such relative joins are joins  in the square-lattice domain, and hence the Gap Lemma applies.

\HAPlanar*
{\em Proof.} 
Fix any constant $0<\varepsilon<1/8$. Assume, for the sake of contradiction, that there exists a polynomial-time algorithm $\mathcal A$ that is guaranteed to return  an error of weight at most
$W^*+cN^{1/18}$.
The constant $c$ will be chosen later.

Using $\mathcal A$, we construct the following polynomial-time algorithm $\mathcal B$ for Gap-$3$SAT$[1,7/8+\varepsilon]$. Given a $3$-CNF formula $\phi$ with $m$ clauses: 
\begin{enumerate}
\item Construct the closed composite gadget $G$ from $\phi$, as described in Section~\ref{mwjha}, and compute its baseline cost $B$. 

\item Embed $G$ in the lattice of the $[[2L^2-2L+1,1,L]]$ planar surface code and its dual, where $L$ is chosen sufficiently large so that every primal and dual defect of $G$ is at distance at least $B+2m$ from every boundary vertex.  Thus, the number of qubits 
$N=2L^2-2L+1$.
\item Run $\mathcal A$ on the defects of $G$ to obtain an error $E$.
\item Remove all cycles from $E$ and let $J$ be the resulting relative join for the defects in $G$. 
\item If the weight of $J$ is at most 
$B+cN^{1/18}$, return YES; otherwise return NO.
\end{enumerate} 
First note that removing cycles from $E$ does not increase its weight. Since $E$ has weight at most $W^*+cN^{1/18}$, the resulting relative join $J$ also has weight at most $W^*+cN^{1/18}$. 
Since cycles can be removed from any error without increasing its weight,
$W^*$ is also the minimum weight of a \emph{relative join} for the defects in $G$.

Let  $\widehat{W}$ denote the minimum weight of a \emph{join} for the defects in $G$, and 
let $U^*$ denote the minimum number of unsatisfied clauses of $\phi$.

We claim that
$
W^*=\widehat{W}.
$ 
Indeed, every join is also a relative join, so $W^*\le\widehat{W}$. Conversely, by construction, every defect of $G$ is at distance at least $B+2m$ from the boundary. Hence any relative join containing a path from a defect to a boundary vertex has weight at least $B+2m$. Since, by the Gap Lemma,
$
\widehat{W}=B+U^*< B+m,
$
no minimum-weight relative join can contain such a path. Therefore every minimum-weight relative join is in fact a join, proving $W^*=\widehat{W}$.

The same argument shows that every relative join $J$ whose weight is guaranteed to be at most $W^*+cN^{1/18}$ is in fact a join, provided that  
\begin{equation}\label{condep}
cN^{1/18}< m. 
\end{equation}
Indeed,
if $J$ has  a 
path from a defect to a boundary vertex,  then its weight is at least $$B+2m > \widehat{W} + m = 
W^* + m > W^*+cN^{1/18},$$ 
contradicting the guarantee on $J$. Consequently, assuming~\eqref{condep}, the output of $\mathcal A$ results in a join.

The remainder of the argument is similar  to that of Theorem~\ref{mainthMWJ}.

 If $\phi$ is satisfiable, then, by the Gap Lemma, $\widehat{W}=B$.   
Therefore,  the output of $\mathcal A$ results in a join of weight at most 
$B+cN^{1/18}$,
and so $\mathcal B$ outputs YES.

If every assignment satisfies at most $(7/8+\varepsilon)m$ clauses, then 
$U^*\geq (1/8-\varepsilon)m$.  By the Gap Lemma, $\widehat{W}=B+U^*$,  and therefore 
 $W^* = \widehat{W} \ge B+(1/8-\varepsilon)m$. 
 The output of $\mathcal A$ results in a join of 
 weight at least $W^*$, and hence greater than $B+cN^{1/18}$, provided that  
\begin{equation}\label{condeps}
cN^{1/18}< (1/8-\varepsilon)m. 
\end{equation}
Therefore, $\mathcal B$ outputs NO.

To guarantee the embedding in Step~2, it suffices to choose
$L>K+2(B+2m)$, where $K$ is 
the minimum side length of a square containing the outer boxes of all its atomic gadgets in $G$.  
By the Gap Lemma, $K=\Theta(m^7)$ and $B=\Theta(m^9)$, so $L=\Theta(m^9)$ suffices. Consequently, $N=\Theta(L^2)=\Theta(m^{18})$,
and therefore choosing $c>0$ sufficiently small guarantees~\eqref{condep} and 
\eqref{condeps}.

Thus $\mathcal B$ solves Gap-$3$SAT$[1,7/8+\varepsilon]$ in polynomial time, contradicting H{\aa}stad's Inapproximability Theorem unless $P=NP$. 
\finito 

\paragraph{Discussion and open direction.} 
Instead of adjusting the proof of Theorem~\ref{mainthMWJ}, we could have used the theorem as a black box in the square-lattice domain and reduced the Minimum-Weight Join problem to the Minimum-Weight Decoding problem for the planar surface code by embedding the given defects into the lattice of the code and its dual. As above, the defects must be placed sufficiently far from the boundary vertices. However, using only the bound $O(K^2)$ on the minimum join weight for defects contained in a $K\times K$ square yields an inapproximability bound of $\Theta(N^{1/28})$, which is weaker than the $\Theta(N^{1/18})$ bound established in Theorem~\ref{planarha}. By instead adjusting the proof, we were able to exploit the stronger bound $B=\Theta(m^9)$ established in the Gap Lemma.

We believe, however, that the reduction from the Minimum-Weight Join problem to the Minimum-Weight Decoding problem for the planar surface code can be refined to recover the stronger $\Theta(N^{1/14})$ inapproximability bound of Theorem~\ref{toricha}. This would follow from the following conjecture.

\begin{conjecture}[Relative joins versus joins]\label{conj1}
There exists a sufficiently small constant $c>0$ such that the following holds. Consider any even-cardinality  sets of 
 primal and dual defects in the lattice of the $[[2L^2-2L+1,1,L]]$ planar surface code and its dual lattice, such that each defect is at distance at least $(1/2-c)L$ from every boundary vertex. Let $W^*$ denote the minimum weight of a relative join for these defects. Then, for  sufficiently large $L$, every relative join of weight at most $W^*+cL$ is in fact a join; that is, it contains no path connecting a defect to a boundary vertex.
\end{conjecture}

The  difficulty in establishing this conjecture lies in the coupling between the primal and dual paths.

\subsection{4.8.8 2D color code}\label{color}

In this section, we establish hardness of approximation for Separate Minimum-Weight Decoding of the $4.8.8$ color code on the torus. Our main observation is that, in the absence of syndrome defects associated with the square faces, decoding $X$ errors in the color code is equivalent to decoding an associated toric code over the depolarizing channel. This correspondence transfers the inapproximability result of Theorem~\ref{toricha} directly to the color code.

Consider  the
$[[4L^2,4,2L]]$ $4.8.8$ color code defined on the torus~\cite{BombinMartinDelgado2006},
whose lattice is shown in Figure~\ref{colorf}(a) for $L=4$. The CSS code associated with the lattice satisfies $C_X=C_Z$. 
Consider the Separate Minimum-Weight decoding problem, in which
$X$- and $Z$-errors occur independently. Accordingly, it suffices to consider
decoding $X$-errors from the given $Z$-syndrome.

{
\renewcommand{\galscale}{0.11}
\begin{figure}[!htbp]
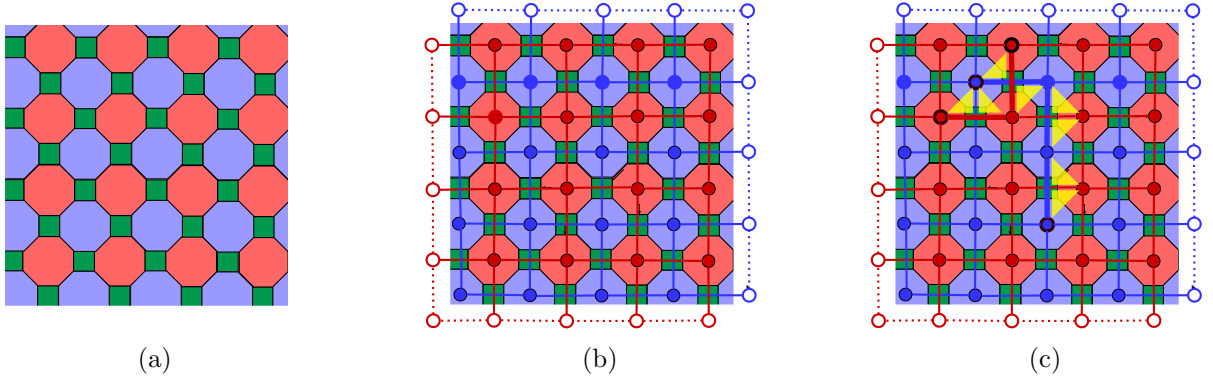
\centering
\sfg{color.pdf}\hfill
\sfg{toric-in-color.pdf}\hfill
\sfg{toric-in-color-corresp.pdf}
\caption{
 (a) \textbf{Lattice of the 
$[[4L^2, 4, 2L]]$  $4.8.8$ color code defined on the torus for $L = 4$}. 
The upper and lower boundaries are identified, as are the left and right boundaries.
The faces are colored red, green, and blue. Unlike the standard coloring convention in the literature, we interchange the red and green colors to match the primal--dual coloring convention adopted for the toric code. The faces correspond to stabilizer checks and the vertices to qubits. 
 (b) \textbf{Shrunk lattices}. The \emph{shrunk blue lattice} has a  blue vertex for each  blue face, with two vertices adjacent whenever they are incident to the same green face. 
 The \emph{shrunk red lattice} is defined analogously on the red faces 
  and is the dual of the shrunk blue lattice.  Vertices and edges identified with those on the opposite side  are shown as unfilled circles and dashed lines, respectively.  
(c) \textbf{Correspondence between color-code and toric-code errors in the absence of green defects.} The blue and red defects are highlighted by black circles. The color-code errors are shown as yellow triangles, and the corresponding toric-code errors are shown as thick blue and red edges.
}
\label{colorf}
\end{figure}
}

 In this decoding problem, the \emph{syndrome defects} are the faces corresponding to unsatisfied checks, and the \emph{errors} are sets of lattice vertices. 
 Given the syndrome defects, we seek a minimum-weight error consistent with the syndrome; that is, a minimum-cardinality set of vertices whose incident checks are unsatisfied exactly at the syndrome defects.

Consider the \emph{shrunk blue lattice} and the \emph{shrunk red lattice} associated with the color code~\cite{BombinMartinDelgado2006}, as defined in Figure~\ref{colorf}(b). These two lattices are dual to each other and define
the $[[2L^2,2,L]]$ \emph{toric code}.

 We establish hardness of approximation for the special case in which all syndrome defects are red or blue; that is, no syndrome defect is associated with a green face. 
 We show that, under this restriction, 
 decoding the color code is equivalent to decoding the associated toric code over the depolarizing channel.  In particular,  
the following lemma establishes a weight-scaling correspondence between the errors in the two decoding problems.

 \begin{lemma}\label{colorcorresp} \textbf{\emph{(Error correspondence between the color code and the toric code  in the absence of green defects)}}   
 Consider a set $S$ of syndrome defects in the  color code consisting only of red and blue faces. 
 Let $T$ be the corresponding set of primal and dual defects in the shrunk blue lattice and its dual, viewed as the syndrome of the associated toric code.

Let $\mathcal E$ denote the set of color-code errors $E$ whose syndrome is $S$ such that, for every green face, $E$ does not contain all four vertices of that face.
Define an equivalence relation on $\mathcal E$ by declaring two errors equivalent if their XOR is the XOR of the vertex sets of a collection of green faces. Errors in the same equivalence class have the same weight.

 Then there is a one-to-one correspondence between the equivalence classes in $\mathcal E$ and toric-code errors whose syndrome is $T$. Under this correspondence, the weight of each error in an equivalence class is twice the weight of the corresponding toric-code error.
 \end{lemma}  
{\em Proof.}  
First, we review the formulation of color-code decoding in terms of triangles \cite{BombinMartinDelgado2006,Bombin2013}.  
Let $\mathcal T$ denote the  color code lattice, and let $\mathcal T^*$ denote its dual lattice, obtained by placing a vertex at the center of each face of $\mathcal T$ and connecting two vertices if the primal faces share an edge. 
Each vertex of $\mathcal T$ is incident to exactly three faces  and therefore corresponds to a triangle in $\mathcal T^*$ whose vertices correspond to the three  faces. 
Consequently, there is a one-to-one correspondence between errors $E$ whose syndrome is $S$ and collections of such triangles whose vertex-wise XOR is  $S$. This is the standard formulation of color-code decoding.

Since $S$ contains no green defects, every green vertex   of $\mathcal T^*$ must be  incident to an even number of triangles. 
If furthermore $E$ belongs to $\mathcal E$,
then   it does not contain all four vertices of any green face. 
Thus, for each green face,  $E$ either contains no triangle incident to that face or exactly two such triangles. Table~\ref{nt25}   lists all possible local configurations in the latter case together with the corresponding toric-code errors.

\begin{table}[!htbp]
\centering
\begin{minipage}{0.47\textwidth}
\centering
{\setlength{\tabcolsep}{1pt}\renewcommand{\arraystretch}{0.47}
\begin{tabular}
{|
>{\centering\arraybackslash}m{0.45\linewidth}|
>{\centering\arraybackslash}m{0.15\linewidth}
||
>{\arraybackslash}m{0.18\linewidth}
|
>{\centering\arraybackslash}m{0.15\linewidth}
|
}
\hline
 \noalign{\vskip 1pt}
 Color code errors  &Weight~& 
 Toric  code errors &Weight~ 
 \\ \hline
 \noalign{\vskip 1pt}
 \tabfig{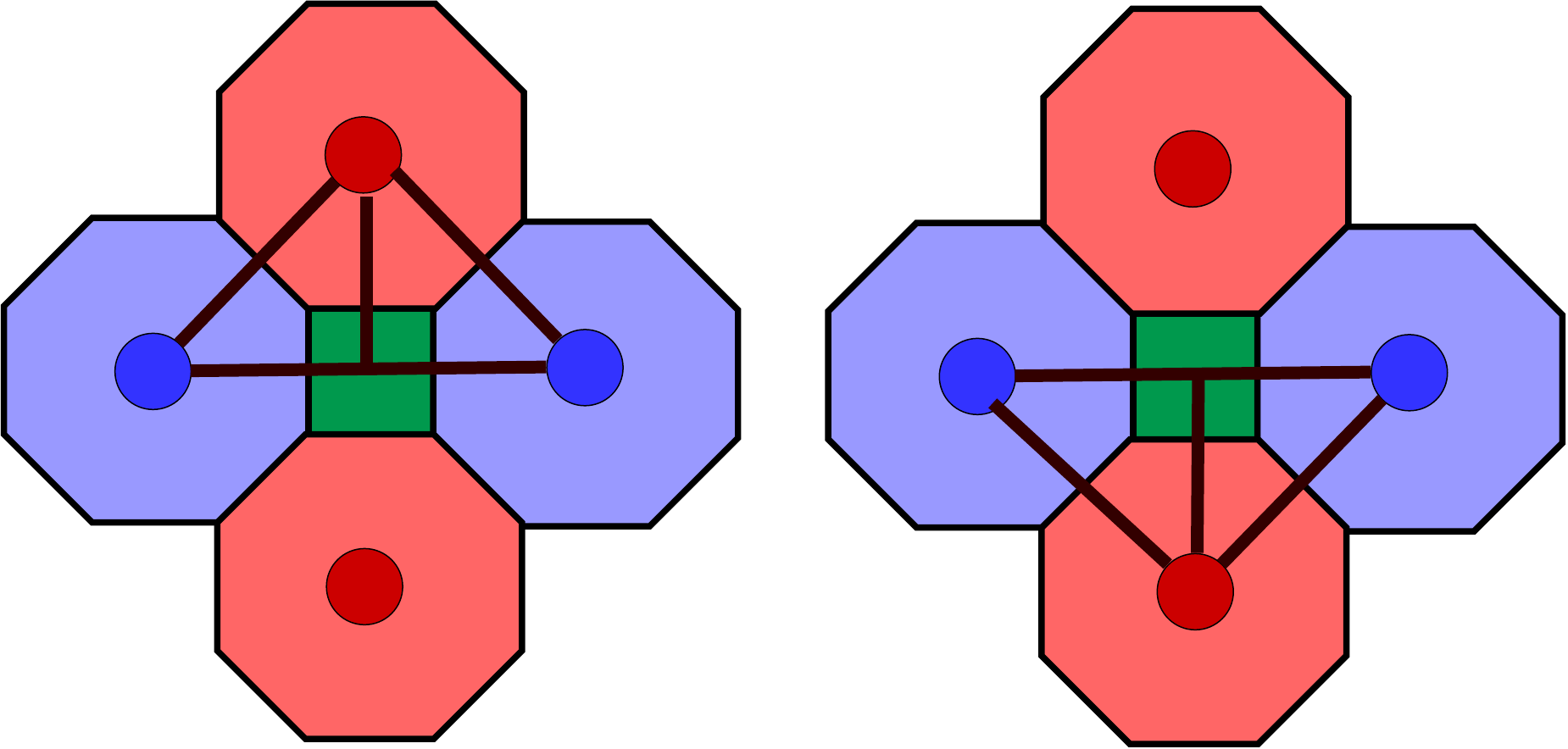}{0.89\linewidth}   &2&
\tabfig{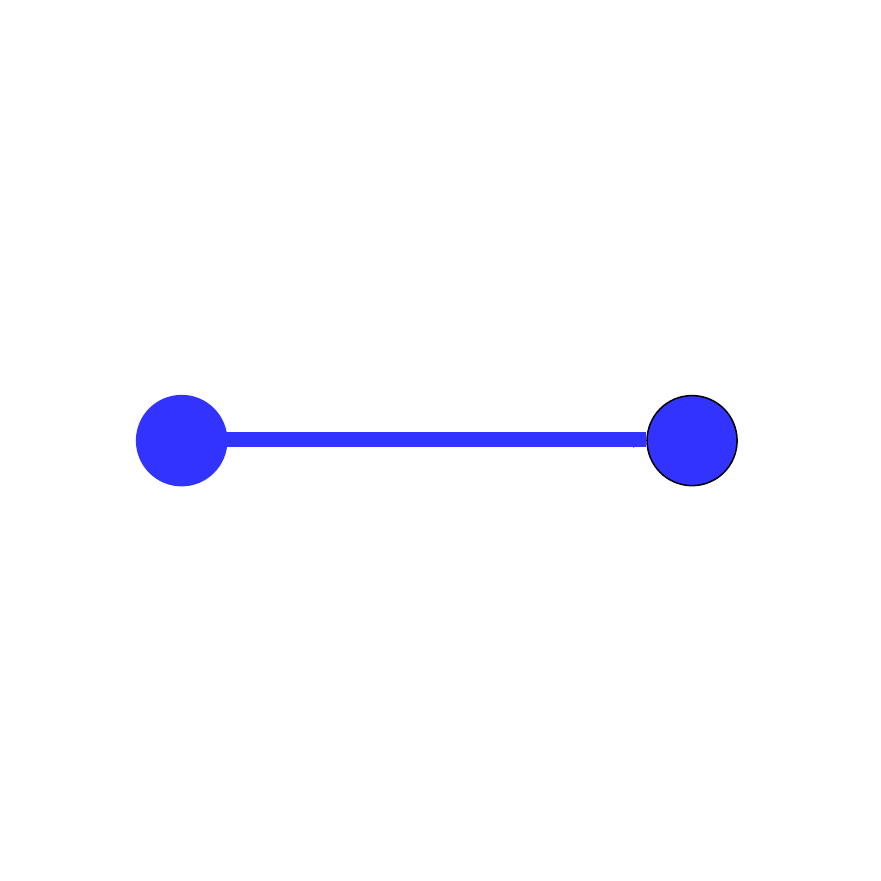}{1\linewidth} &1

 \\ 
 \hline
 
   \noalign{\vskip 1pt}
 \tabfig{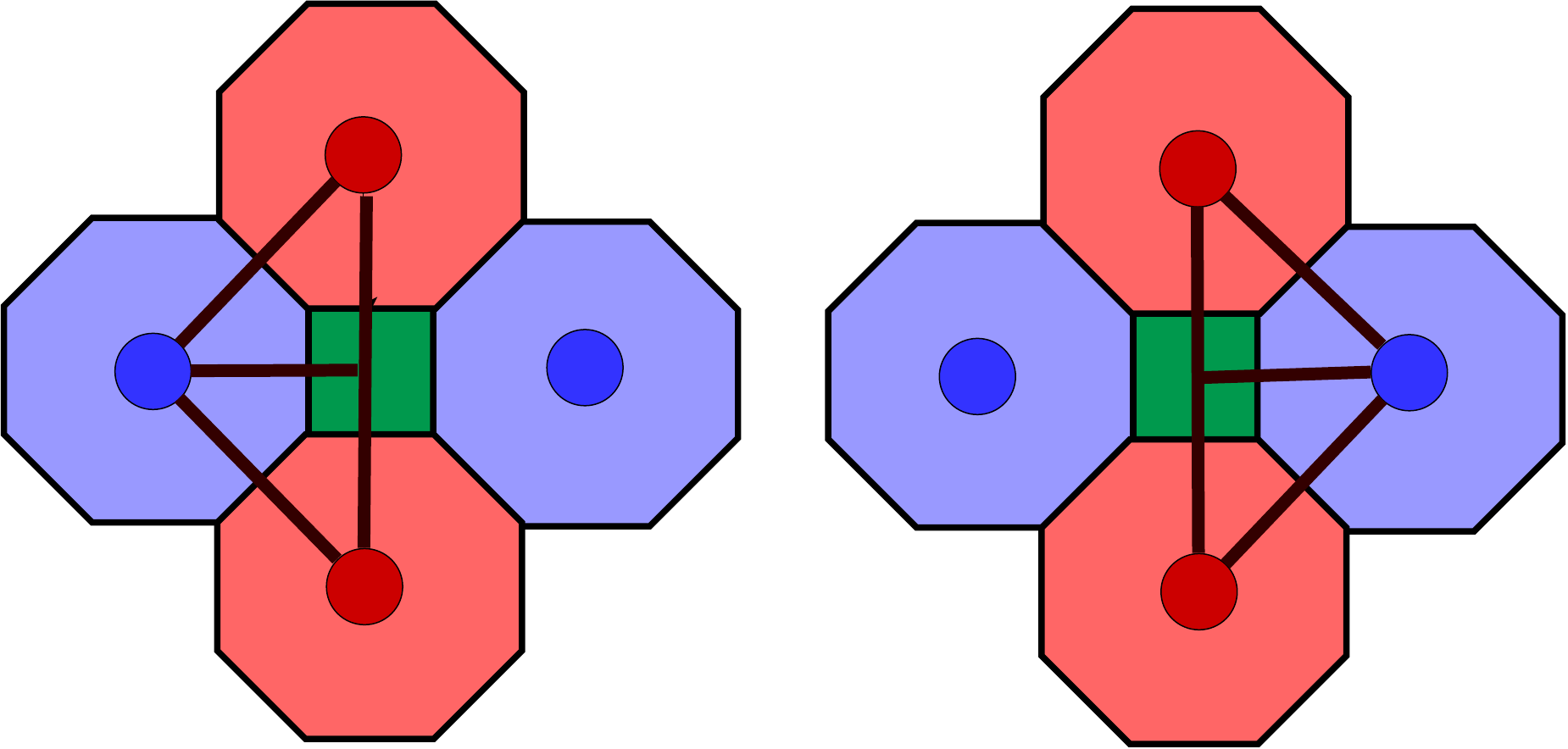}{0.89\linewidth}  &2&
\tabfig{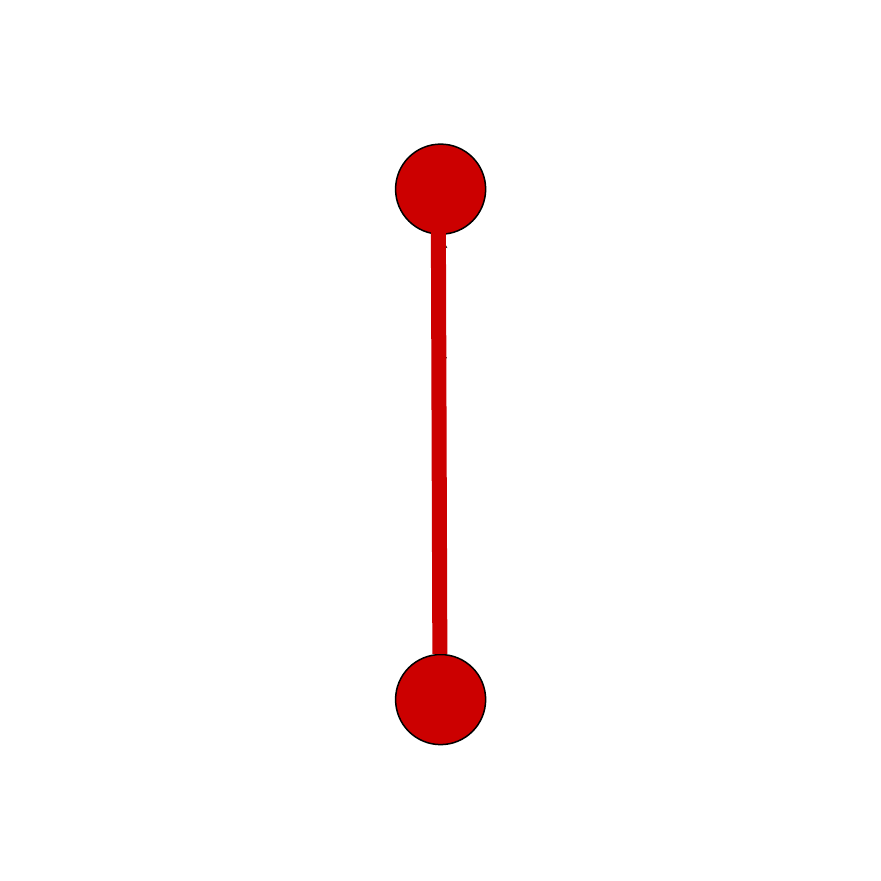}{1\linewidth} &1

 \\ \hline

  \noalign{\vskip 1pt}
 \tabfig{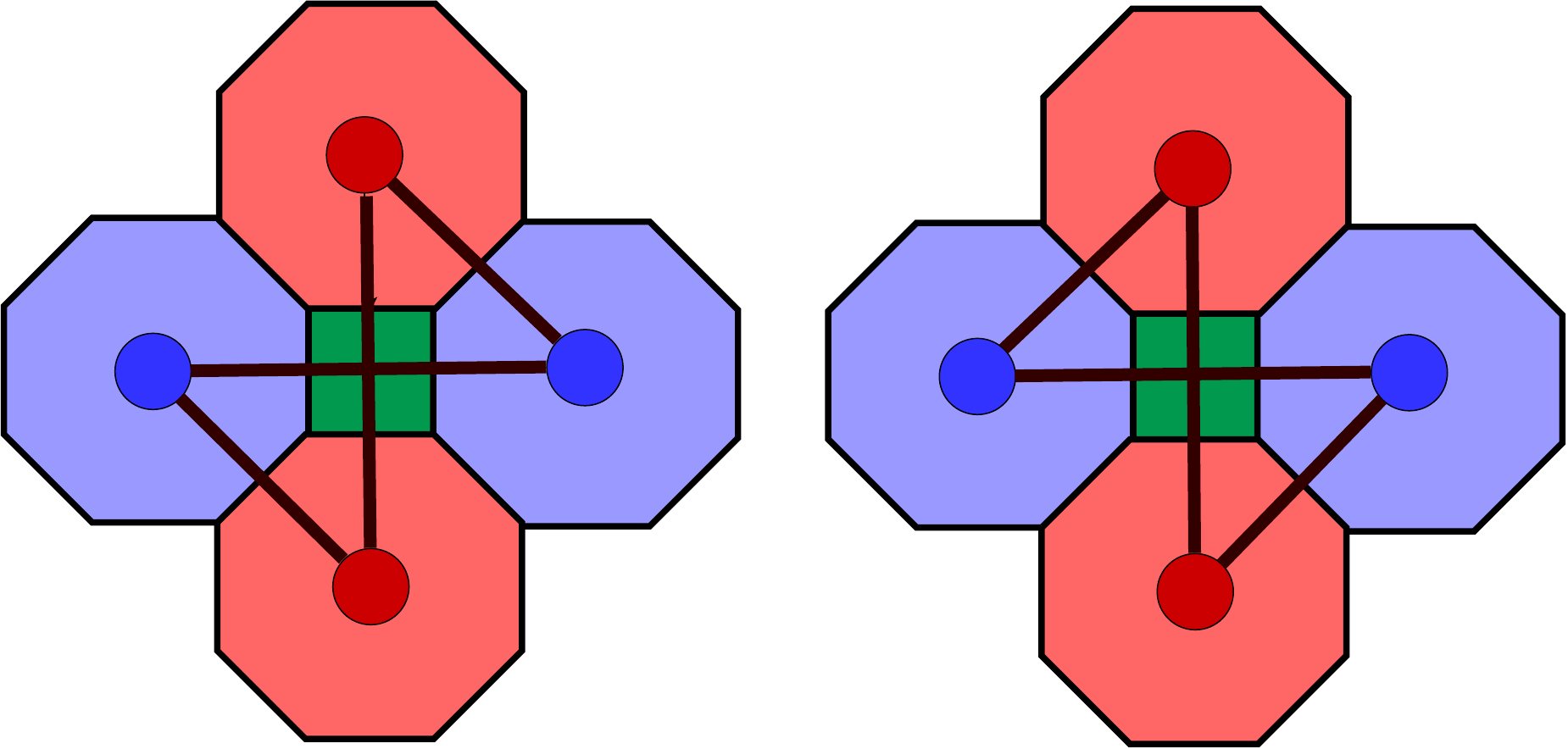}{0.89\linewidth}  &2&
\tabfig{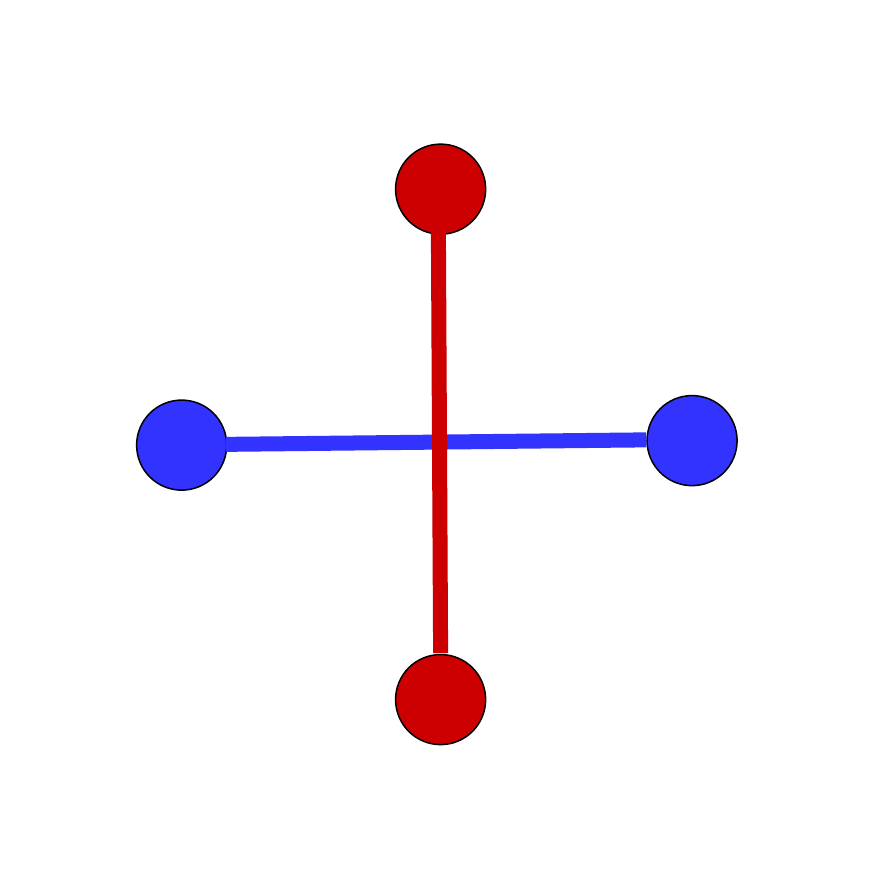}{1\linewidth} & 1

 \\  \hline

\end{tabular}}
\end{minipage}
\hfill
\begin{minipage}{0.47\textwidth}
\centering
{\setlength{\tabcolsep}{1pt}\renewcommand{\arraystretch}{0.47}
\begin{tabular}
{|
>{\centering\arraybackslash}m{0.45\linewidth}|
>{\centering\arraybackslash}m{0.15\linewidth}
||
>{\arraybackslash}m{0.18\linewidth}
|
>{\centering\arraybackslash}m{0.15\linewidth}
|
}
\hline
 \noalign{\vskip 1pt}
 Color code errors  &Weight& Toric  code errors &Weight
 \\ \hline 
 \noalign{\vskip 1pt}
 \tabfig{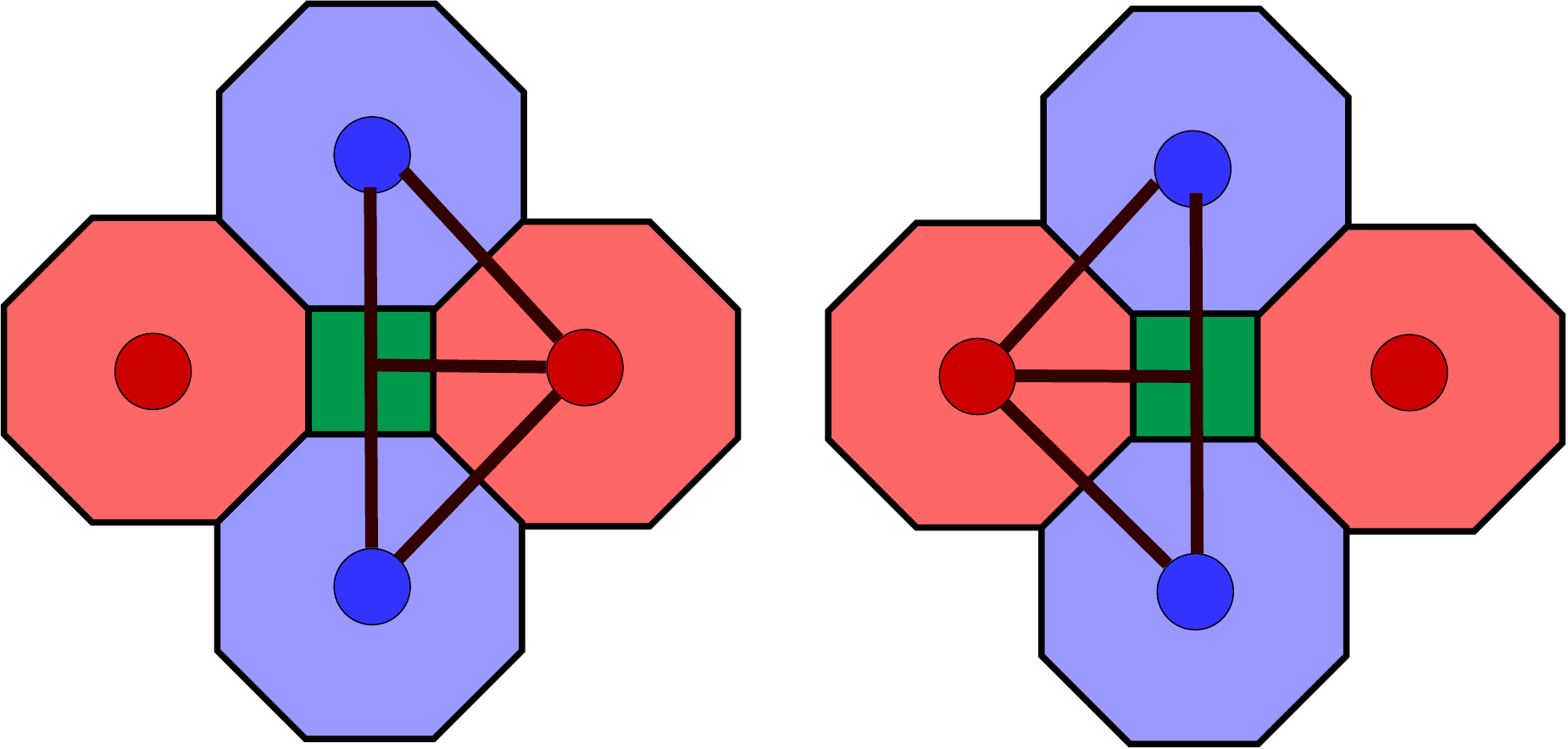}{0.89\linewidth}  &2&
\tabfig{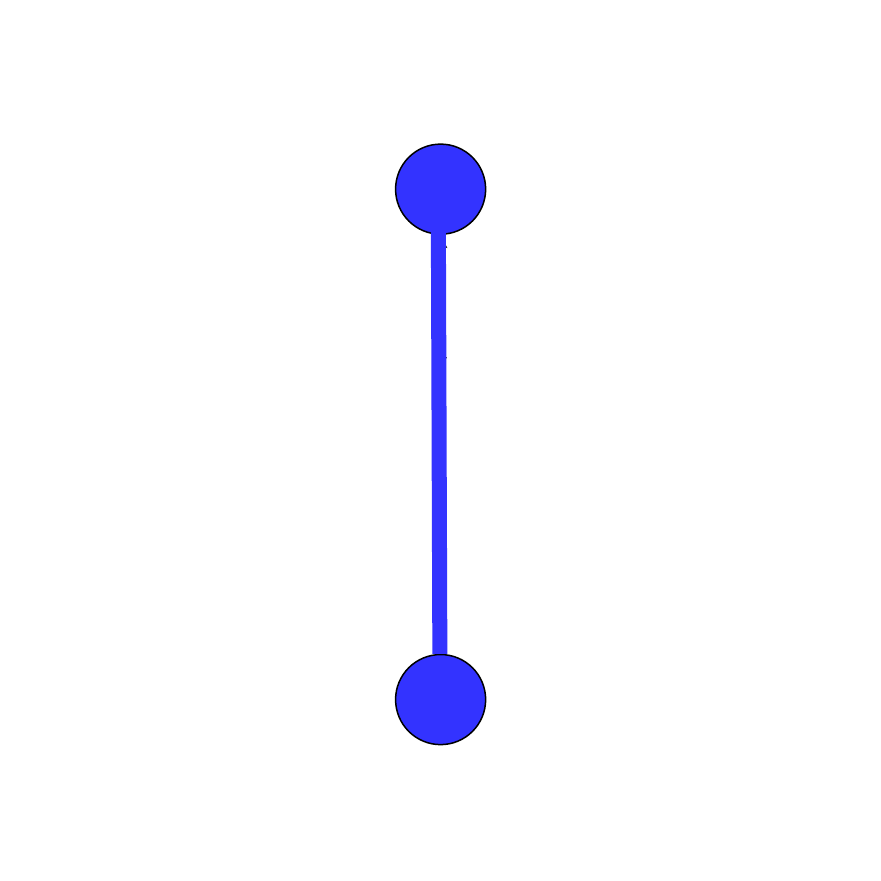}{1\linewidth} & 1

 \\ \hline

  \noalign{\vskip 1pt}
 \tabfig{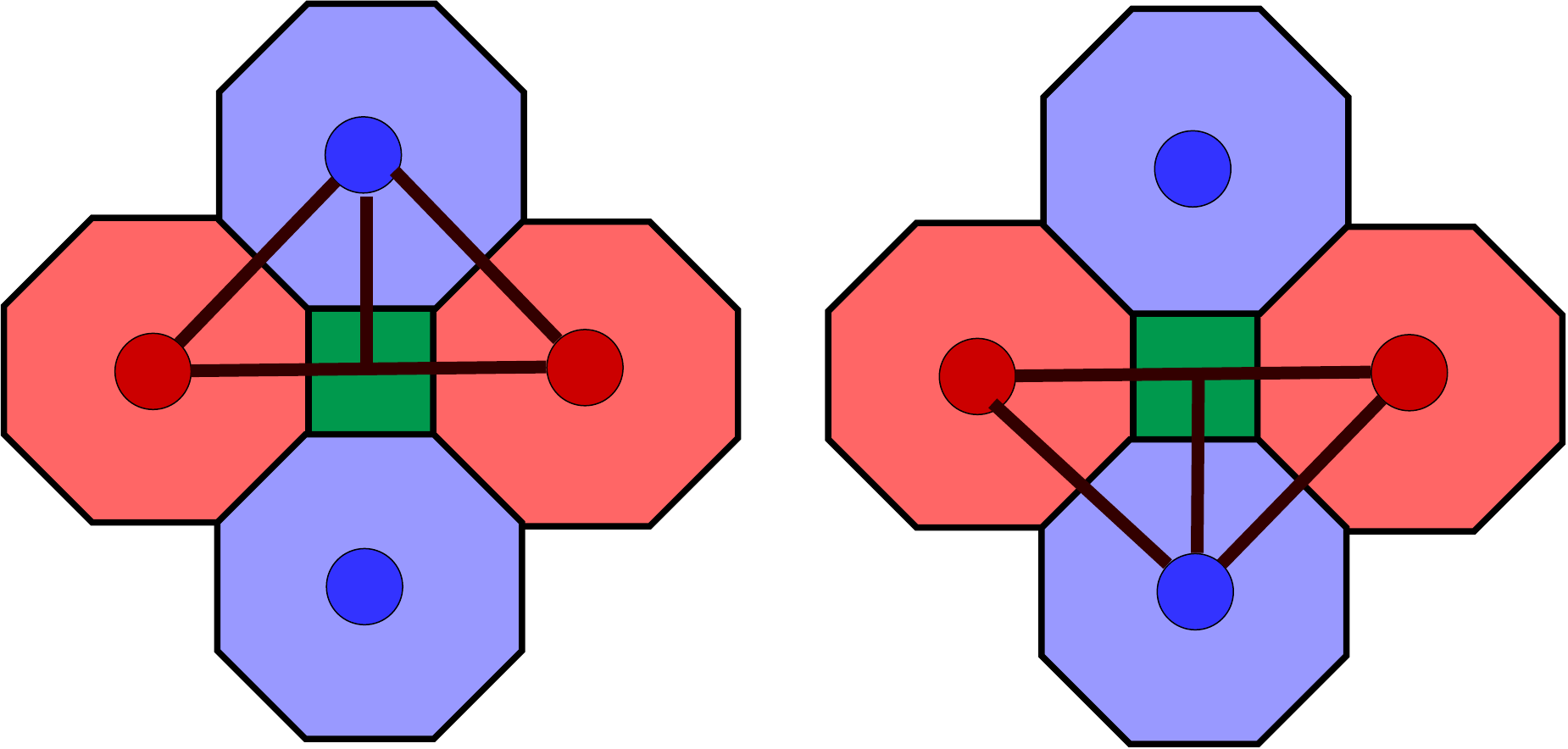}{0.89\linewidth}   &2&
\tabfig{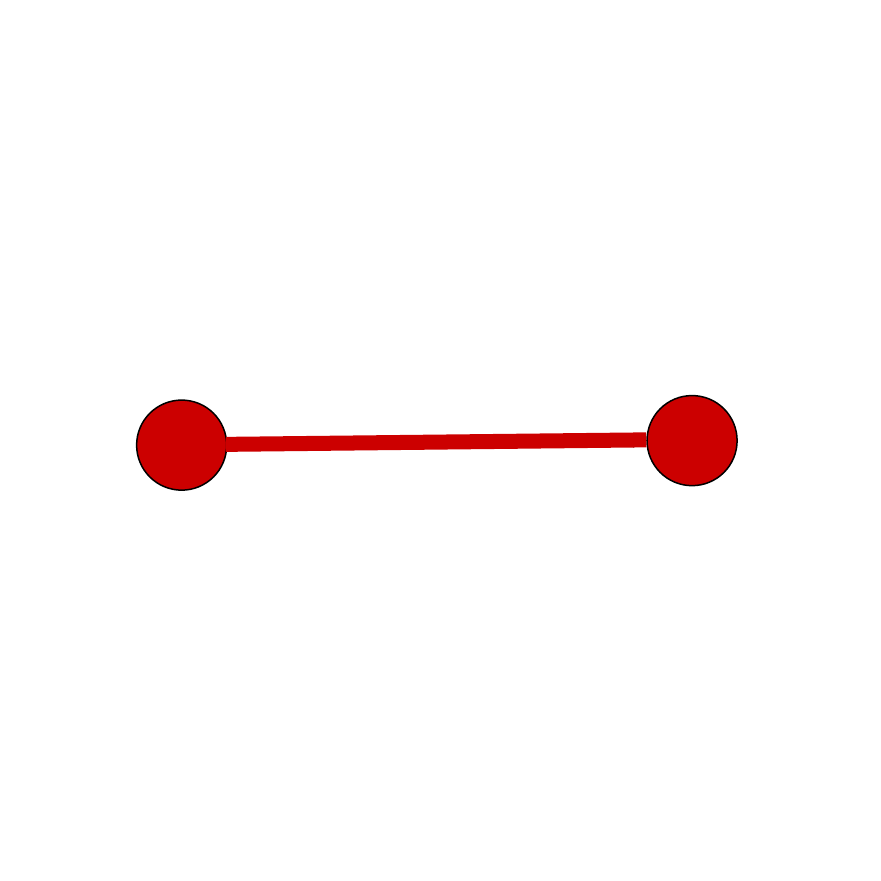}{1\linewidth} & 1

 \\ \hline

  \noalign{\vskip 1pt}
 \tabfig{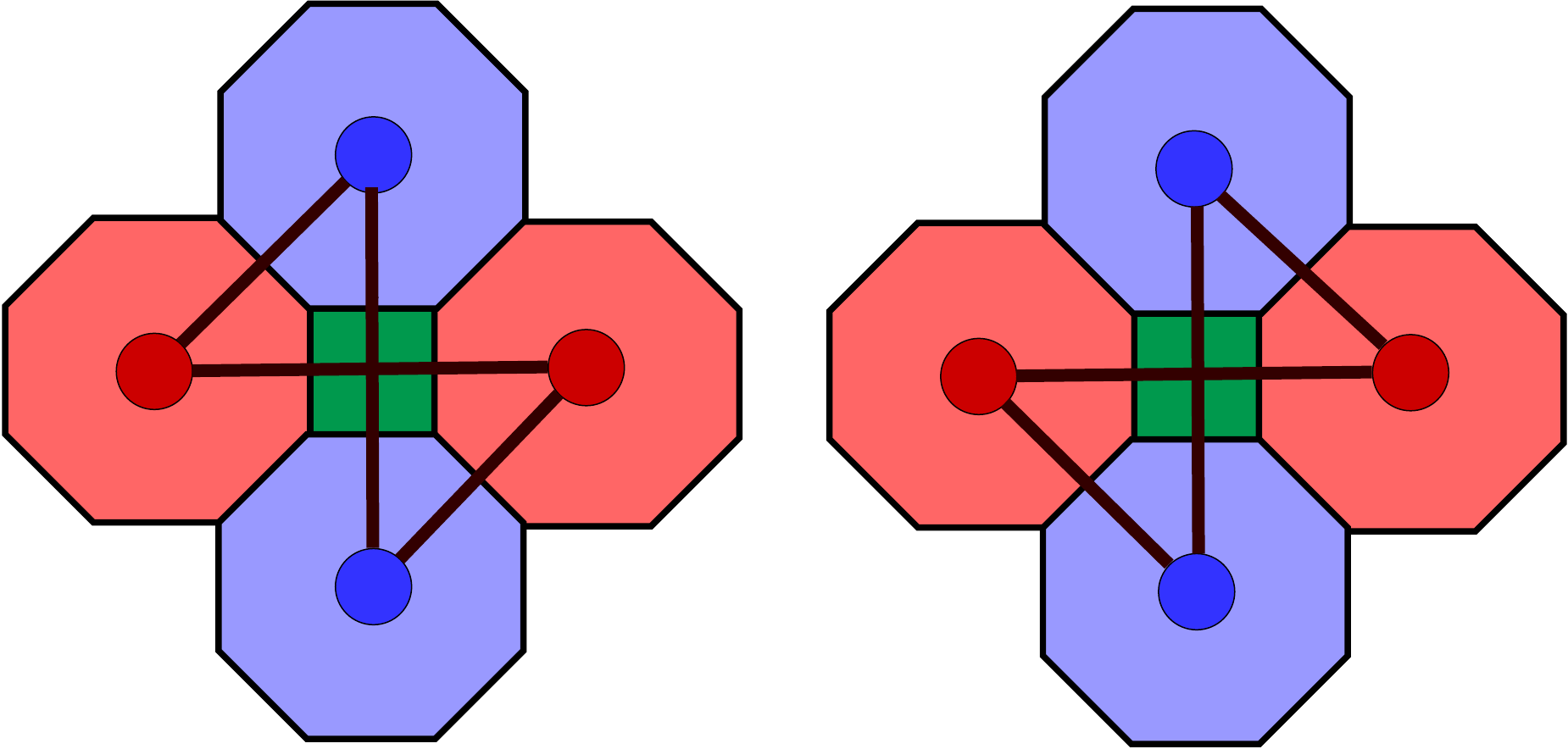}{0.89\linewidth}   &2&
\tabfig{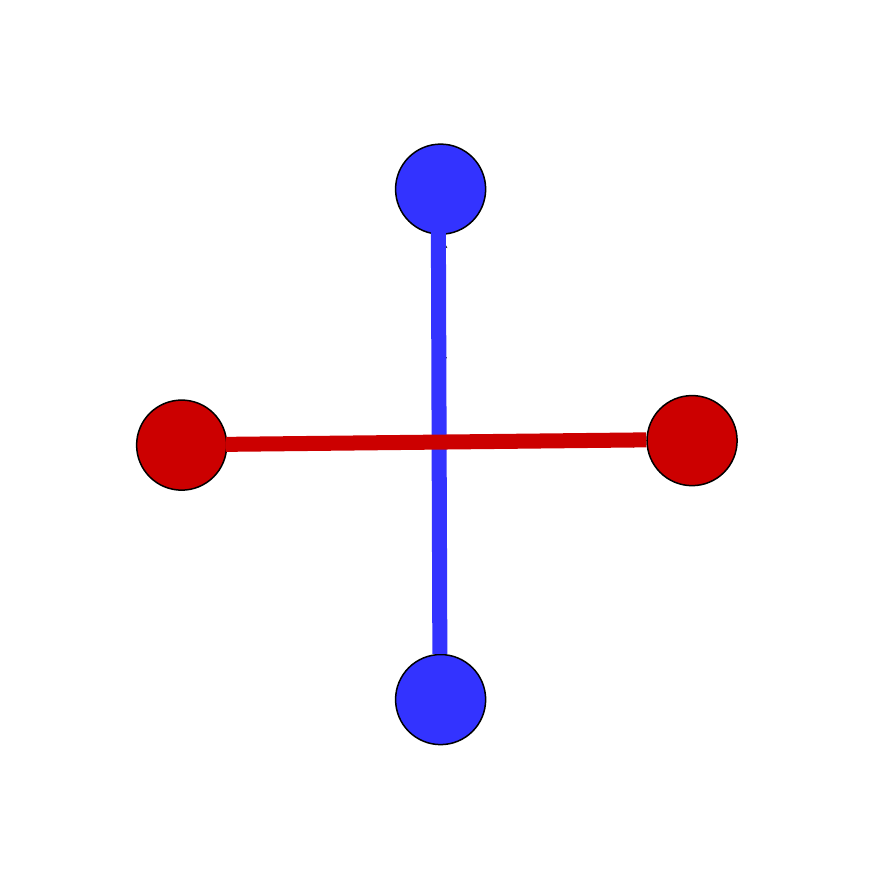}{1\linewidth} & 1

 \\ \hline

\end{tabular}}
\end{minipage}
\caption{\textbf{Local correspondence in the proof of Lemma~\ref{colorcorresp}}}
\label{nt25}
\end{table}
Applying the correspondence independently at each green face yields a one-to-one correspondence between the equivalence classes in $\mathcal E$ and toric-code errors with syndrome $T$. Moreover, as shown in the table, the weight of $E$ is twice the joint weight of the corresponding toric-code error, where a primal edge and its dual are counted only once.  
See Figure~\ref{colorf}(c) for an illustration of this correspondence.
\finito 

Thus, decoding the toric code reduces to decoding the color code by mapping the toric-code defects to syndrome defects associated with the red and blue faces of the color code, with no syndrome defects associated with green faces.  
Accordingly, Theorem~\ref{toricha} immediately implies the following inapproximability result for Separate Minimum-Weight Decoding of the color code.

\HAColor*
{\em Proof.}
Apply Theorem~\ref{toricha} and Lemma~\ref{colorcorresp}. Note that we may assume the approximation algorithm for the color code returns an error that does not contain all four vertices of any green face. Indeed, whenever this is not the case, we repeatedly remove the vertex set of such a face. This preserves the syndrome, strictly decreases the weight of the error, and therefore preserves the approximation guarantee. 
Since the correspondence multiplies error weights by $2$ and the
block lengths of the two codes differ only by a constant factor, the
$\Omega(N^{1/14})$ inapproximability gap is preserved.
\finito

\section{Extensions and open questions}\label{conc}

Our inapproximability results also extend to the rotated surface
code~\cite{Bombin2007,rotated2}, surface codes with
holes~\cite{Dennis2002,Fowler2012SurfaceCodes}, and the triangular $4.8.8$
color code~\cite{BombinMartinDelgado2006}, with an
$\Omega(N^{1/18})$ inapproximability gap in each case. 
We omit the details; establishing the extension to the
triangular color code requires a more technical treatment of the
boundaries.

The weight lower bound in the Localization Lemma is likely not tight. 
Its proof controls nonlocal interactions by considering violations case by case and working with reduced primal components in which all dual paths are omitted and all primal paths connecting defects within the same nucleus are omitted.  
 The weight of these omitted
paths must then be compensated for by the separation parameter
$\Delta$, leading to the relatively large separation required by the lemma.
A stronger version of the Localization Lemma would lead to an improved
inapproximability gap. We conjecture that the following strengthening holds.

If the nucleus of an atomic gadget $g$ has dimensions $A\times B$, we define
the \emph{span of the nucleus} of $g$ to be $\max\{A,B\}$. 

\begin{conjecture}[Stronger localization]
\label{conj2}
There exists a constant $c>0$ such that the following holds.
For any closed composite gadget $G$, if
$\Delta>cS$, where $S$ is the 
maximum span of a nucleus among
 the atomic gadgets
in $G$, then every minimum-weight join for the defects in $G$ is local.
\end{conjecture}

Establishing this conjecture would allow us to set the separation parameter
$\Delta=\Theta(m)$ in the hardness-of-approximation reduction, instead of
$\Theta(m^5)$. This would improve the inapproximability gap for the
Minimum-Weight Join problem from $\Omega(K^{1/7})$ to $\Omega(K^{1/3})$.
Accordingly, the inapproximability gap for the toric code and the $4.8.8$
color code would improve to $\Omega(N^{1/6})$. The same gap would extend
to the planar surface code under Conjecture~\ref{conj1}.

More broadly, how large can a sublinear additive inapproximability gap be?
In particular, what is the largest exponent $\alpha<1$ for which an
$\Omega(N^\alpha)$ inapproximability gap holds? 
Given the
polynomial-time approximation schemes for minimum-weight
decoding~\cite{WaltersApprox26,GWKApprox26}, we know that $\alpha<1$.

\section*{Acknowledgments} 
The authors are grateful to  Ali Sobh for  stimulating  discussions on this work and for helpful  feedback on the first draft of the paper.

\section*{Author Contributions}

L.B. conceived the project, developed the main  results and proofs, and wrote the manuscript. G.K. developed the reduction from the toric code to the color code, contributed to discussions, and prepared most of the figures.  ChatGPT (OpenAI) was used to assist with text
editing, rephrasing for clarity, checking notation consistency, and
searching for relevant literature. 
All AI-assisted content and references were thoroughly reviewed and verified before inclusion in the manuscript.

\addcontentsline{toc}{section}{References}
\bibliographystyle{IEEEtran}
\bibliography{refs}


\appendix

\addcontentsline{toc}{section}{Appendix}

\section{Proof of the  Localization Lemma}
\label{appA}

In this section, we establish Lemma~\ref{loclemma}, which is restated below for convenience.

\LOCLemma*

Throughout this section, let $G$ be a closed composite gadget.

Recall that a join is \emph{local} if (1) it is the union of its restrictions to the gadgets of $G$, and (2) no dual path in the restriction of $J$ to one gadget couples with a primal path in the restriction of $J$ to another gadget.

The proof proceeds by introducing a special class of joins, called \emph{canonical joins}, and then establishing a lower bound on the weight of  a  non-canonical join.  
In a canonical join, every relay and outer pin is incident to exactly one path and is matched properly: an outer pin is connected to one of  its corresponding relay pins, while a relay pin is connected either to its corresponding outer pin, another relay pin, or a nucleus defect. Moreover, defects in the nucleus of a gadget are connected only to defects in the same nucleus. These conditions guarantee the first property of a local join. 
For the second property, we further restrict the primal and dual paths. Dual paths are confined to a box, called the \emph{red box}, lying between the relay and outer boxes. Primal paths connecting nucleus defects are confined to the outer box, while the lengths of all primal paths incident to relay pins are bounded. These restrictions ensure that
primal paths in one gadget cannot couple with dual paths in another.

We now formalize this strategy by defining canonical joins in Definition~\ref{candef} and proving in Corollary~\ref{canislocal} that every canonical join is local.

Consider the box of
$L_\infty$-radius $9\Delta/5-\delta$ centered at the gadget, as shown in Figure~\ref{canfig}. We
refer to this box as the \emph{red box}.

\begin{figure}[!htbp]\centering
\includegraphics[width=0.4\textwidth]{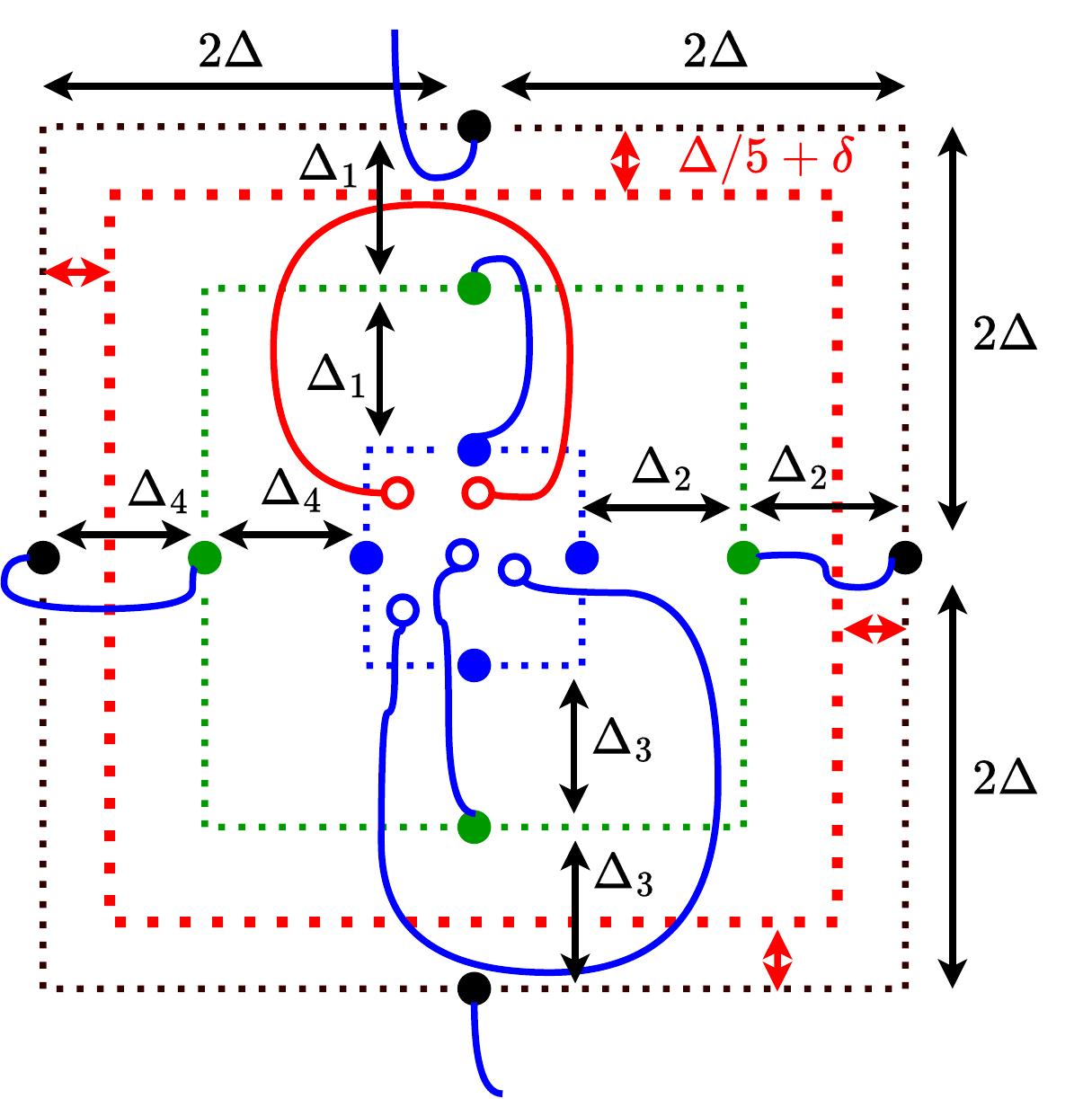}
\caption{\textbf{Local paths in a canonical join.}
The outer, red, relay, and nucleus boxes are shown by dotted black, red, green, and blue lines, respectively. Outer, relay, and nucleus pins are shown as filled black, green, and blue circles, respectively; interior nucleus defects are shown as unfilled blue circles. Primal and dual paths are shown in red and blue, respectively.  The condition $\delta<2\Delta/5$ ensures that $\Delta/5+\delta<\Delta_i$ for every $i$ (since $\Delta_i\ge \Delta-\delta$), and therefore the red box strictly contains the relay box.
}

\label{canfig}
\end{figure}

\begin{definition}[Canonical joins]\label{candef}
A  join $J$ for the defects in  $G$ is called \emph{canonical} if it satisfies the following conditions.

\begin{enumerate}
\item[(1)] \textbf{Matching conditions:}
\begin{enumerate}
\item[(A)] Every relay or outer pin 
 has $J$-degree~1, where the \emph{$J$-degree} of a defect is the number of paths in $J$ incident to that defect.
\item[(B)] Every outer pin is matched to one of the two relay pins of the gadgets sharing that pin.

\item[(C)] 
 Every relay pin of a gadget $g$ is matched  by a path to one of the following:
\begin{itemize}

\item[$\bullet$] its corresponding outer pin by a path of length at most $7\Delta/5$;
\item[$\bullet$] a defect in the nucleus of $g$ by a path of length at most $7\Delta/5$; 
\item[$\bullet$] another relay pin of $g$ by a path of length at most $12\Delta/5$.
\end{itemize}
\end{enumerate}
\item[(2)] \textbf{Primal nucleus locality condition.}
Every path incident to a primal defect of a gadget $g$ that is not connected to a relay pin of $g$ is contained entirely within the outer box of $g$ and connects the defect to a primal defect in the nucleus of $g$.

\item[(3)]  \textbf{Dual locality condition.}
Every path incident to a dual defect of a gadget $g$ is contained entirely within the red box of $g$ and connects the defect to a dual defect in the nucleus of $g$.

\end{enumerate}
\end{definition}

In a canonical join, every primal or dual path connects two defects in the same gadget. Thus, the join is the union of its restrictions to the atomic gadgets. The following lemma shows that primal and dual paths cannot couple unless they are associated with the same gadget, and hence that a canonical join is local.

\begin{lemma}[Separation Lemma]\label{seplem}
Let $J$ be a canonical join, let $p$ be a primal path in $J$ connecting two defects in a gadget $g$, and let $q$ be a dual path in $J$ connecting two defects in a different gadget $g'$. Then $p$ and $q$ do not contain edges dual to each other.
\end{lemma}
{\em Proof.}  
By condition~(3), $q$ is contained entirely within the red box of $g'$. By conditions~(1) and~(2), $p$ is contained in a box of $L_\infty$-radius  $2\Delta+\Delta/5+\delta/2$ centered at $g$. Indeed, the worst case occurs when $p$ connects an outer pin to a relay pin. These two pins are at distance at least $\Delta-\delta$, while the path connecting them has length at most $7\Delta/5$, and hence can extend beyond the outer box by at most
$(7\Delta/5-(\Delta-\delta))/2=\Delta/5+\delta/2$.

The red box of $g'$ has $L_\infty$-radius $9\Delta/5-\delta$. Since the centers of $g$ and $g'$ are at $L_\infty$-distance at least $4\Delta$, the red box of $g'$ is disjoint from the box containing $p$. Hence, $q$ cannot contain an edge dual to an edge of $p$.
\finito

\begin{corollary}\label{canislocal}
Every canonical join  is local. 
\end{corollary} 
 We reduce the Localization Lemma to the following three lemmas. The first establishes a lower bound on the weight of a join that violates the matching conditions but contains no relatively long primal paths. The second establishes a lower bound on the weight of a join that contains relatively long primal paths.  
 The third is a local lemma concerning a single gadget.  It bounds the coupling between the relay-pin paths satisfying the matching conditions and any dual path crossing from the interior of the nucleus box to the exterior of the red box.
\begin{lemma}[Matching conditions: Short violations]
\label{canlemAS}
Let $J$ be a join that violates the matching conditions. Suppose that every path in $J$ incident to an outer or relay pin has length less than $3(\Delta-\delta)$. Then the weight of $J$ is at least
\[
R+2\Delta/5-2\delta.
\]
\end{lemma} 
\begin{lemma}[Joins with long primal paths]
\label{canlemAL}
Let $J$ be a join containing a primal path of length at least $3(\Delta-\delta)$ whose endpoints are not both defects in the same nucleus. Then the weight of $J$ is at least
\[
R+\Delta-3\delta.
\]
\end{lemma}
\begin{lemma}[Local coupling] \label{localred}
Consider any gadget $g$. For each relay pin of $g$, consider a path incident to that pin and satisfying condition~(1C) in the definition of a canonical join. 
 Let $q$ be a dual path connecting a dual-lattice node in the nucleus box to a node outside the red box. Then $q$ contains at least $2\Delta/5-3\delta-1$ edges whose dual edges lie in the red box and are not contained in any of the paths incident to the relay pins.
\end{lemma}
Combining Lemmas~\ref{canlemAS} and~\ref{canlemAL}, we obtain the following.
\begin{corollary} [Matching conditions]
\label{canlemA} 
If a join $J$ violates the matching  conditions, then the weight of $J$ is at least
\[
R+2\Delta/5-3\delta. 
\]
\end{corollary}

The proofs of these three lemmas are given in Sections~\ref{canlemAP}, \ref{canlemAPP}, and~\ref{canlemBP}, respectively. The first two are primarily combinatorial, while the third is geometric.

We establish the Localization Lemma  from these   three lemmas in Section~\ref{prus}. Before showing this, we isolate the following simple observation, which is used repeatedly throughout the proof and explains the origin of the relay cost $R$ in the lower bounds.

The observation applies to sets of primal paths that are not necessarily joins. We call a set $J$ of simple primal paths incident to primal defects in $G$ a \emph{relay cover} if, for every gadget, each relay pin is incident to at least one path in $J$.

\begin{lemma}[Relay covers]\label{localrelaycover}
The weight of every relay cover $J$ is at least the relay cost $R$ of $G$.
\end{lemma} 
{\em Proof.} 
For each relay pin $r$, place an $L_1$ ball $B(r)$ centered at $r$ 
with radius equal to the depth of $r$, i.e., the distance from $r$ 
to its corresponding outer pin. These balls are diamonds that are either 
disjoint or intersect only along their boundaries. Since each relay pin 
$r$ is incident to at least one path in $J$, and no other defect lies in 
the interior of $B(r)$, such a path contains at least 
$\operatorname{depth}(r)$ edges before leaving $B(r)$. Since the interiors 
of the balls are disjoint, these edge sets are disjoint for distinct relay 
pins. Hence, the weight of $J$ is at least the sum of the depths of the 
relay pins, which, by definition, is the relay cost $R$ of $G$.
 \finito

A relay cover  that will arise several times in the proof is the
\emph{reduced primal component} of a join, which we define next.
The \emph{reduced primal component} of a join $J$ is the set of paths obtained from $J$ by removing (1) all dual paths
and (2) all primal paths connecting two primal defects within the same nucleus.

\subsection{Proof  of the Localization  Lemma 
using 
Corollary~\ref{canlemA} and Lemma~\ref{localred}}\label{prus}

Consider a non-local join $J$. By Corollary~\ref{canislocal}, $J$ is not
canonical. Hence, with respect to the canonical join conditions, exactly one of the
following cases must occur:
\begin{itemize}
\item[$\bullet$] \emph{Case 1:} $J$ violates condition~(1).
\item[$\bullet$] \emph{Case 2:} $J$ satisfies condition~(1) but violates condition~(2).
\item[$\bullet$] \emph{Case 3:} $J$ satisfies conditions~(1) and~(2) but violates condition~(3).
\end{itemize}

In Case~1, Corollary~\ref{canlemA} implies that the weight of $J$ is at least
$R+2\Delta/5-3\delta$.

In Case 2, there is a path $p$ connecting a primal defect in the
nucleus of a gadget $g$ to a primal defect in the nucleus of a
gadget $g'$, possibly with $g=g'$, such that $p$ is not contained
entirely within the outer box of $g$.
Since $p$ leaves the outer box of $g$, 
it must traverse the distance between a nucleus box and the exterior of the corresponding outer box at least twice:
if $g\neq g'$, once near each endpoint, and if $g=g'$, once when
leaving the outer box and once when returning to the nucleus.
Hence, the length of $p$ is at least $4(\Delta-\delta)$. 
Let $J'$ be the reduced primal component of $J$. After removing $p$
from $J'$, if present, the resulting set $J''$ is still a relay cover.
Thus, Lemma~\ref{localrelaycover} implies that its weight is at least
the total relay cost $R$. Therefore, since $p$ has length at least
$4(\Delta-\delta)$, the weight of $J$ is at least
$R+4(\Delta-\delta)$.

In Case~3, $J$ contains a dual path $q$, incident to a dual defect of a gadget $g$, that is not contained entirely within the red box of $g$.  Let $J'$ be the reduced primal component of $J$. The weight of $J'$ is at least $R$. 
 Hence, the weight of $J$ is at least $R+k$, where $k$ is the number of edges of $q$ whose dual edges lie in the red box of $g$ but are not contained in any path of $J'$. By the argument in the proof of the  Separation Lemma~\ref{seplem}, primal paths associated with gadgets other than $g$ do not reach the red box of $g$. Thus, it suffices to consider the paths incident to the relay pins of $g$. Lemma~\ref{localred} therefore implies that $k\ge 2\Delta/5-3\delta-1$, and hence the weight of $J$ is at least $R+2\Delta/5-3\delta-1$.

\subsection{Proof of Lemma \ref{canlemAS}}\label{canlemAP}

Assume that $J$ violates one of the matching conditions~(1A), (1B), or~(1C), but every path in $J$ incident to an outer or relay pin has length less than $3(\Delta-\delta)$.

Let $J'$ be the  reduced primal component of $J$. Since the matching conditions do not concern the paths   $J\setminus J'$, $J'$ also violates the same matching condition.   
While $J'$ is not a join, it satisfies the following properties:
\begin{itemize}
\item[(J1)] every relay and outer pin has odd $J'$-degree;
\item[(J2)] for each gadget, the parity of the sum of the $J'$-degrees of its primal nucleus defects  equals the parity of the gadget.
\end{itemize}
Apart from the parity constraint in~(J2), the $J'$-degrees of the primal defects in the nucleus are unrestricted and may even be zero.

  We establish the lower bound 
  $R+2\Delta/5-2\delta$  
on the weight of $J$ by proving that it also holds for $J'$.

The strict upper bound $3(\Delta-\delta)$ on path lengths ensures that
relay and outer pins can only be connected to defects in the same gadget or in an adjacent gadget.
We proceed by analyzing these possibilities and
iteratively updating $J'$  until the violated matching conditions are restored. 
Each update is 
the XOR with either a cycle or a path whose endpoints lie in a nucleus.  
We show that the overall  transformation decreases the weight of $J'$
by at least $2\Delta/5-2\delta$. 
The claim then follows from the fact that the updated $J'$ is a relay cover.

\begin{lemma} 
The weight of $J'$ is at least $R+2\Delta/5-2\delta$. 
\end{lemma}
{\em Proof:} 
Consider the following \emph{relaxed} versions of conditions~(1B) and~(1C), in which path-length bounds are ignored and outer and relay pins are allowed to be incident to multiple paths. In the relaxed version of~(1B), every outer pin is only allowed to be connected to the two relay pins of the gadgets sharing that pin. In the relaxed version of~(1C), every relay pin of a gadget is only allowed to be connected to its corresponding outer pin, a defect in the nucleus of the same gadget, or another relay pin of the same gadget.

\medskip
\noindent
\emph{Step 1. Eliminate relaxed matching violations.} 
We first eliminate violations of the relaxed conditions~(1B) and~(1C).
Since every path in $J'$ has length less than $3(\Delta-\delta)$, 
the only possible violations are:
\begin{itemize}
\item[(i)] 
a defect in the nucleus of a gadget is connected to an outer pin of the same gadget; or
\item[(ii)]
a relay pin $r$ is connected to a relay pin $r'$ in an adjacent gadget,
with $r$ and $r'$ associated with the same outer pin.
\end{itemize}  
 Up to symmetry, case~(i) is represented by the five entries in the first  column of Table~\ref{cases}, which are distinguished according to the connectivity of the relay pin associated with the outer pin. Case~(ii) is represented by the last entry in the first column and the first entry in the second column, which are distinguished according to the connectivity of the common outer pin.

{\setlength{\tabcolsep}{9pt}\renewcommand{\arraystretch}{1}
\begin{longtable}{|l|l|l||l|l|l|}
\hline
Old & Updated & Gain & Old & Updated & Gain\\ \hline
\endfirsthead
\hline
Old & Updated & Gain & Old & Updated & Gain\\ \hline
\endhead
\noalign{\vskip 2pt}

(i) 
\tabfigscale{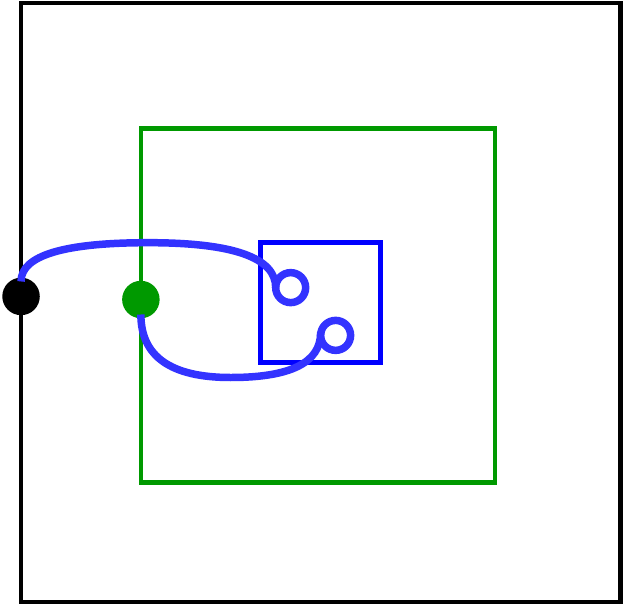}{0.18} 
 & \tabfigscale{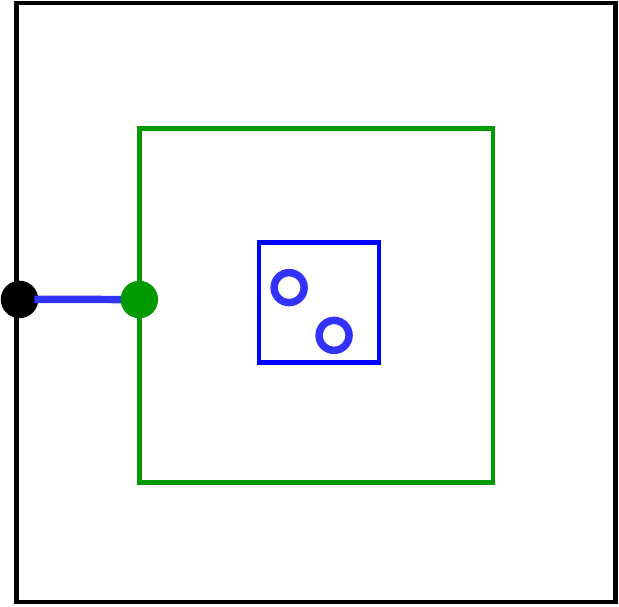}{0.18}  & $2(\Delta-\delta)	$

 &

 (ii) 
  \tabfigscaletrim[48mm]{0.18}{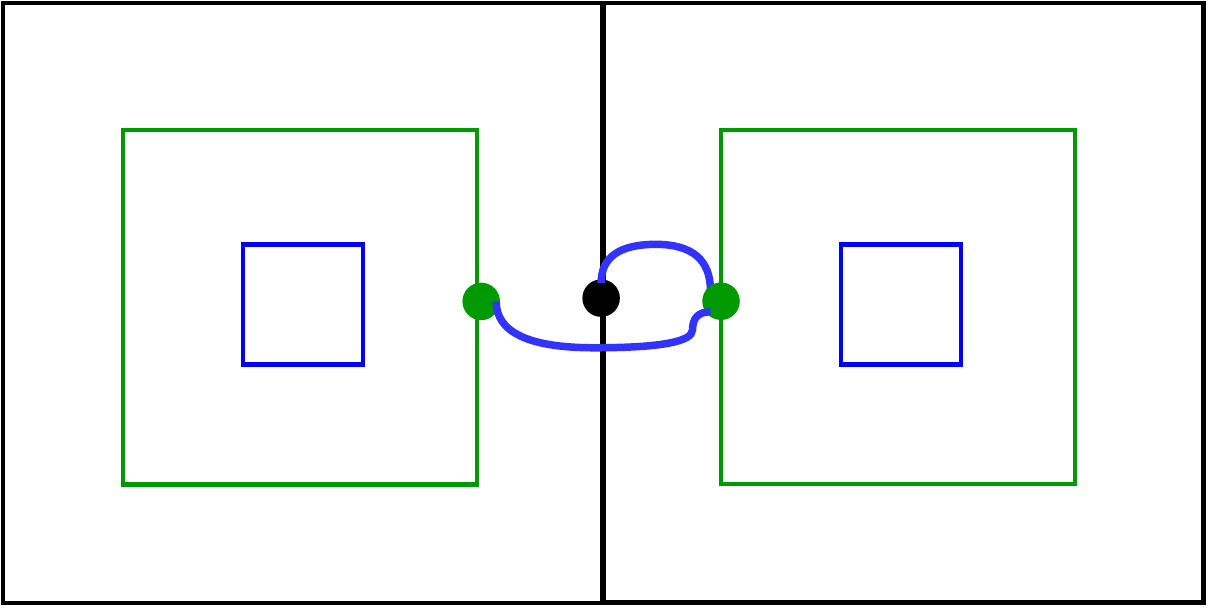} 
 & \tabfigscaletrim[48mm]{0.18}{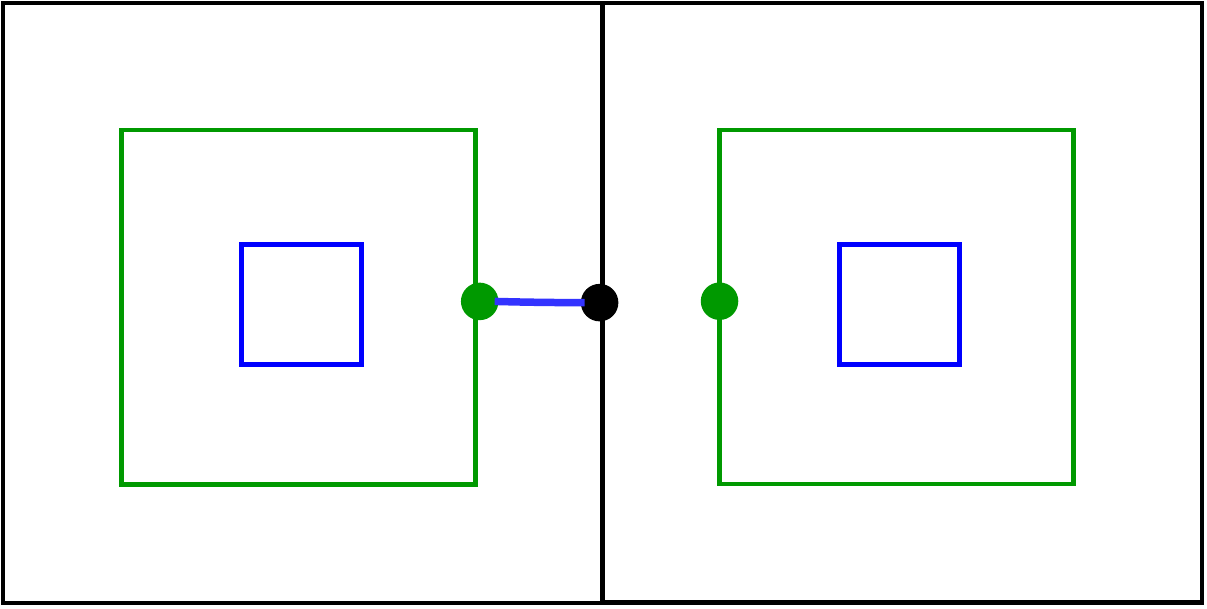}  & $2(\Delta-\delta)$

 \\ \hline
 
 \noalign{\vskip 2pt}

(i)  
 \tabfigscale{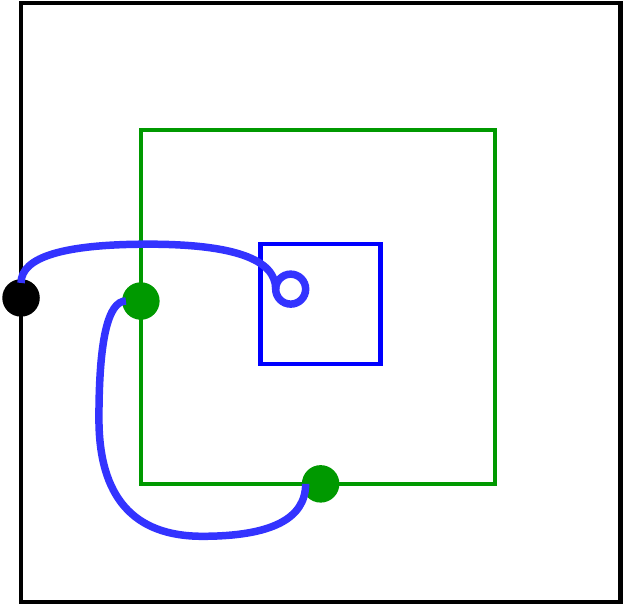}{0.18} 
 & \tabfigscale{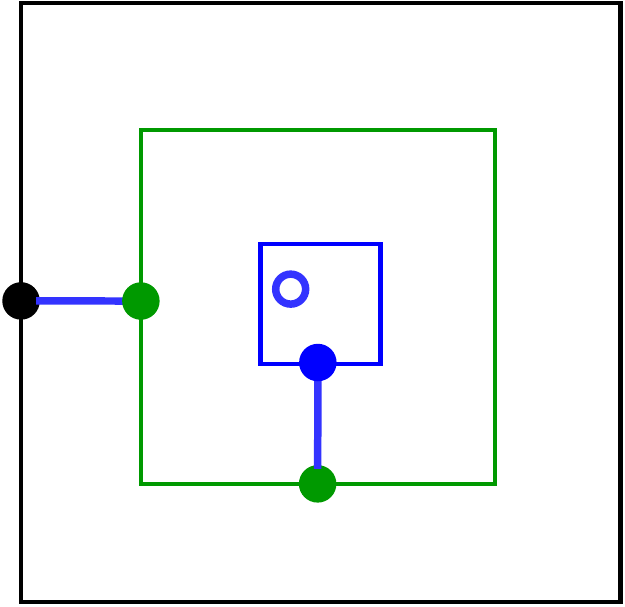}{0.18}  & $2(\Delta-\delta)	$
 
 &
(iii)
  \tabfigscale{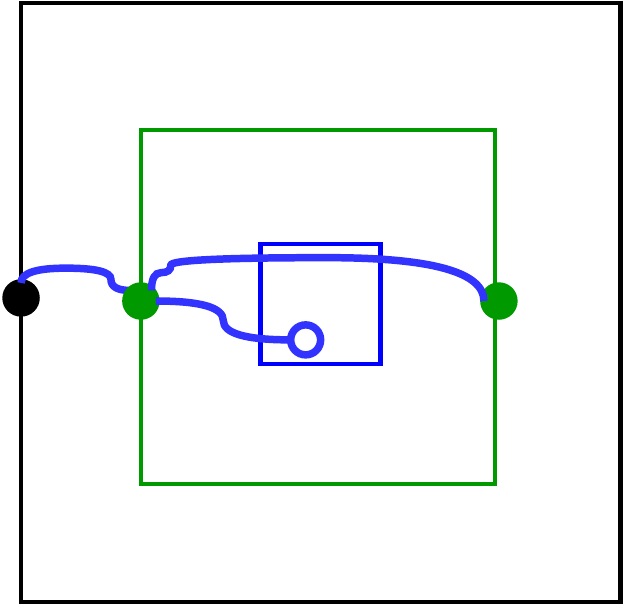}{0.18} 
 & \tabfigscale{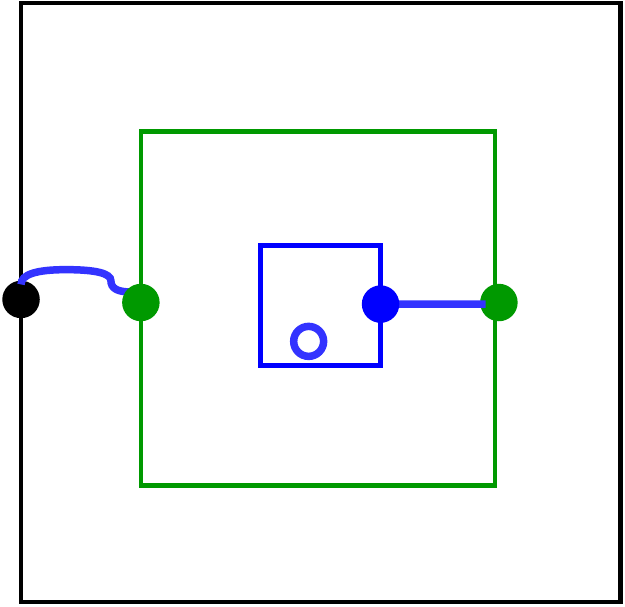}{0.18}  & $2(\Delta-\delta) $

 \\ \hline
 \noalign{\vskip 2pt}
 
(i) 
\tabfigscale{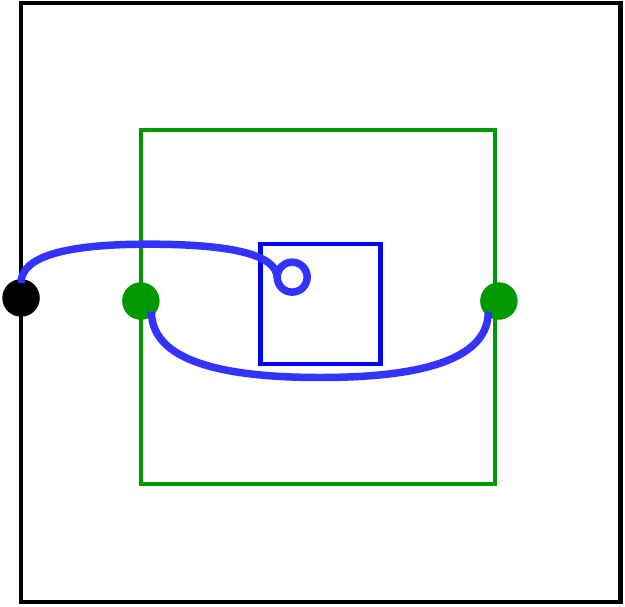}{0.18} 
 & \tabfigscale{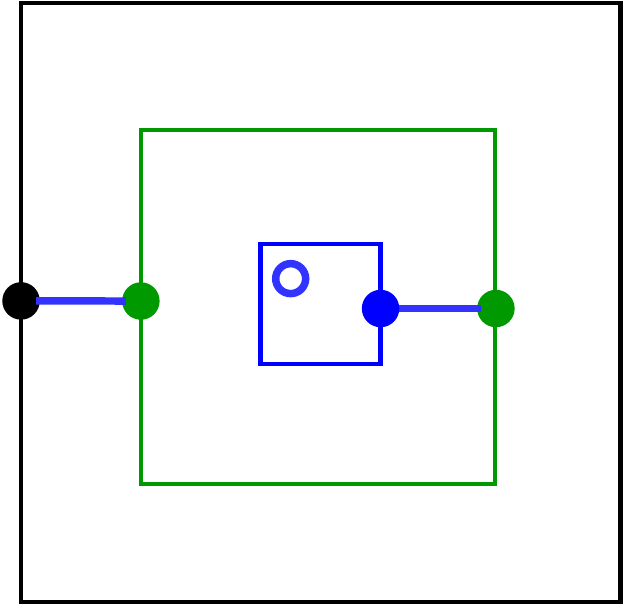}{0.18}  & $2(\Delta-\delta)$
 
 &

 (iii) 
 \tabfigscale{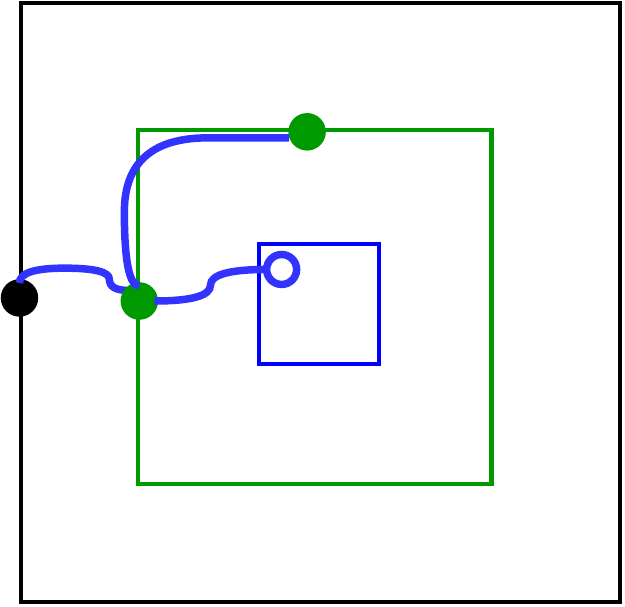}{0.18} 
 & \tabfigscale{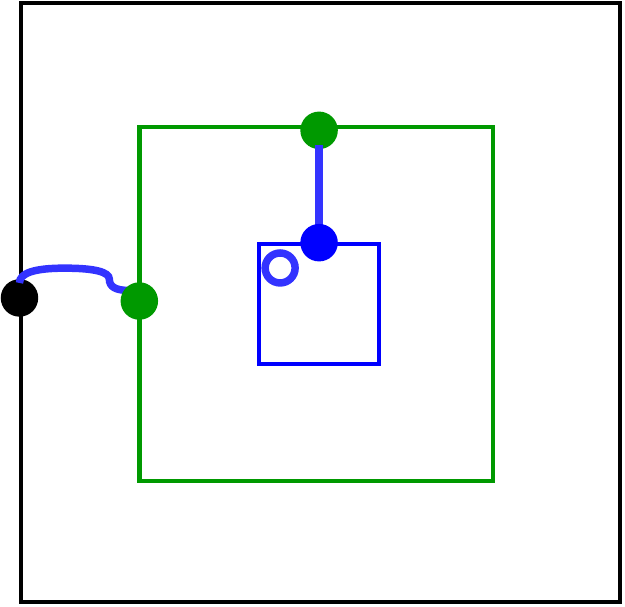}{0.18}  & $2(\Delta-\delta)$

 \\ \hline
 \noalign{\vskip 2pt}
(i)
  \tabfigscale{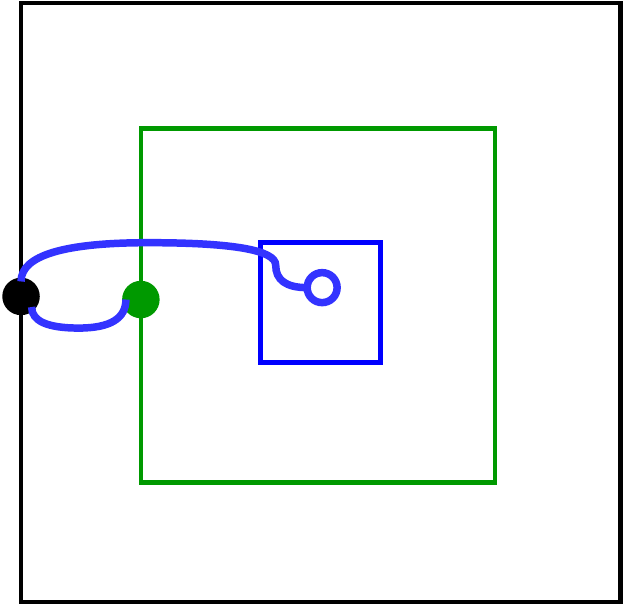}{0.18} 
 & \tabfigscale{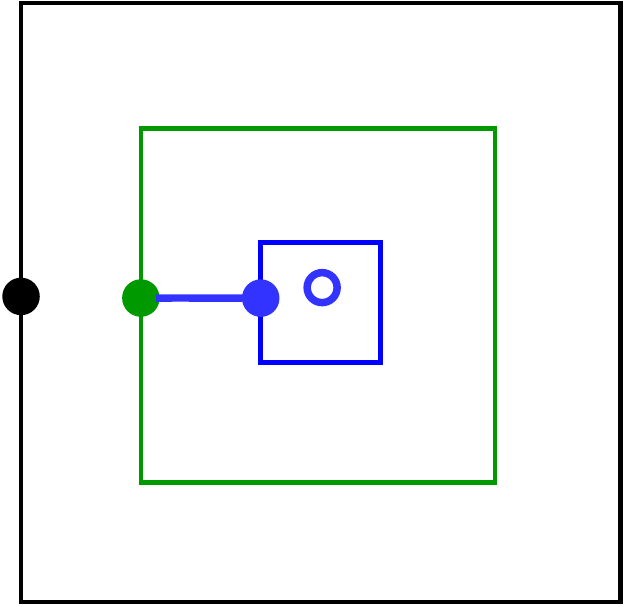}{0.18}  & $2(\Delta-\delta) $
 &
 
(iv)
  \tabfigscale{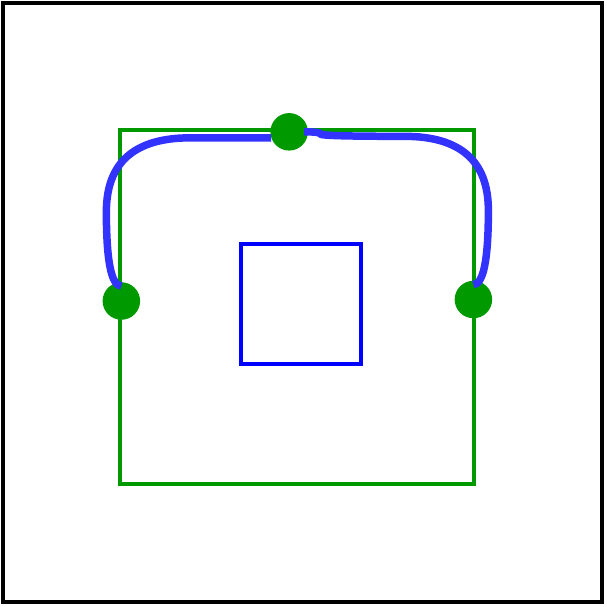}{0.18} 
 & \tabfigscale{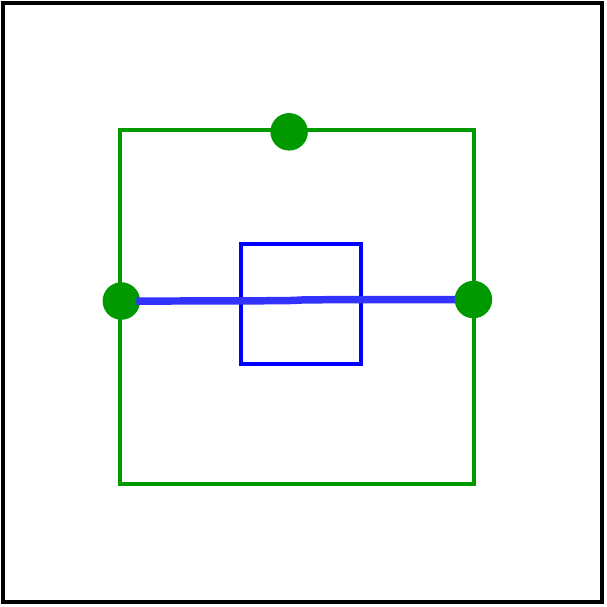}{0.18}  & $2(\Delta-\delta) $
 
 \\ \hline

 \noalign{\vskip 2pt}
 
 \hspace{-1em}
 $\begin{array}{c}\mbox{(i)}\\\mbox{(ii)}\end{array}$  \hspace{-0.5em}
  \tabfigscaletrim[48mm]{0.18}{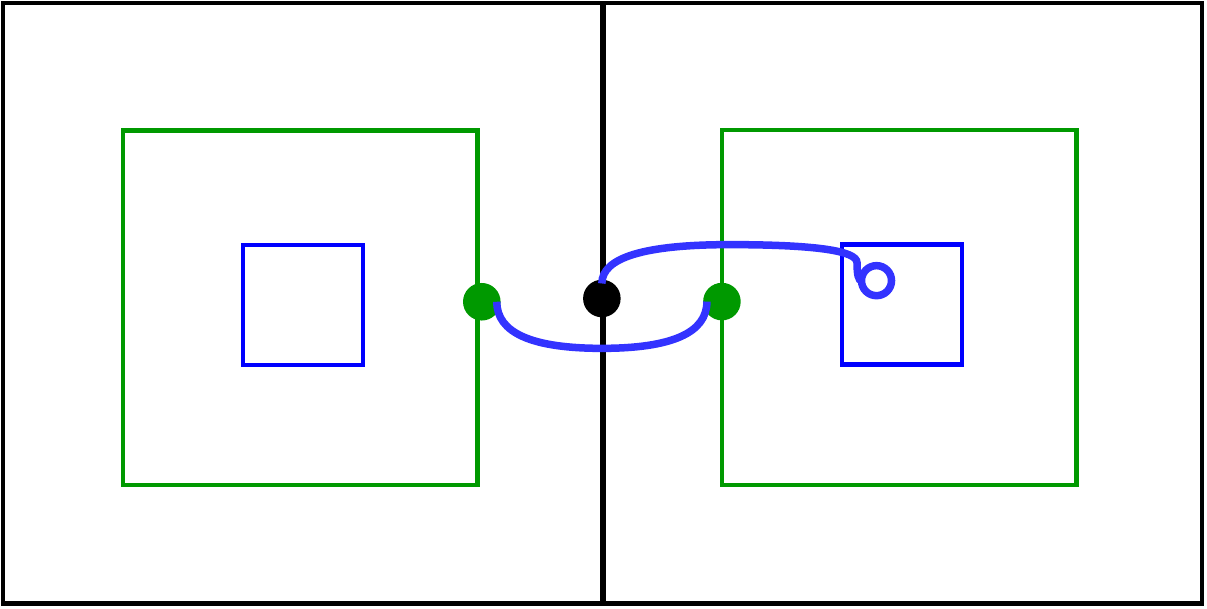} 
 &  \tabfigscaletrim[48mm]{0.18}{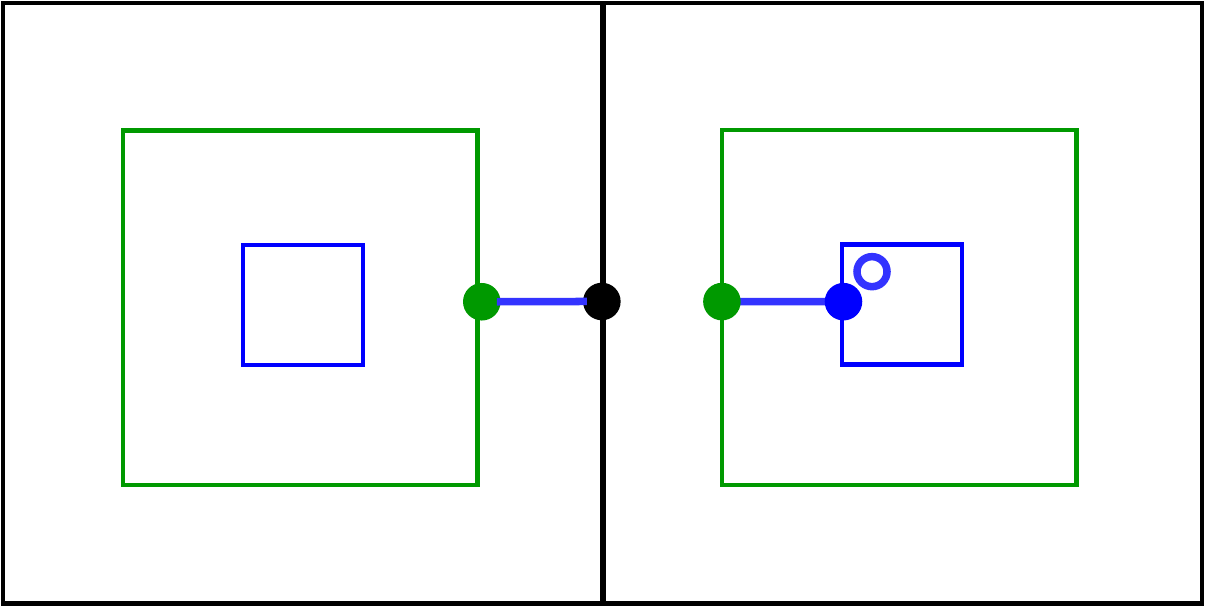}  & $2(\Delta-\delta)$
 &
 (iv) 
 \tabfigscale{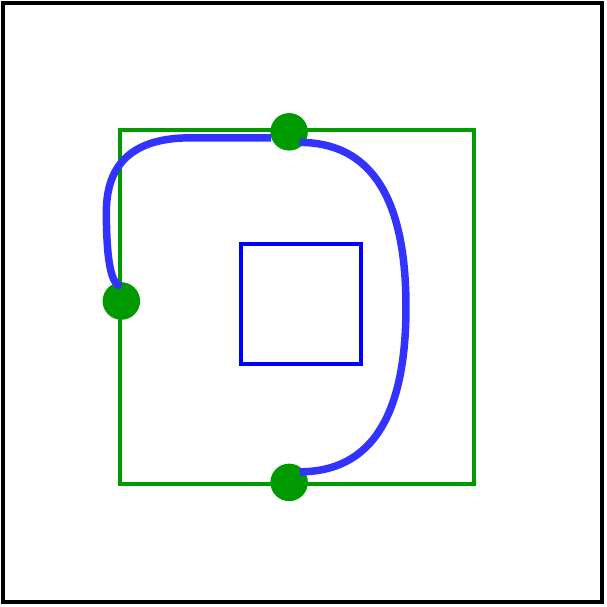}{0.18} 
 & \tabfigscale{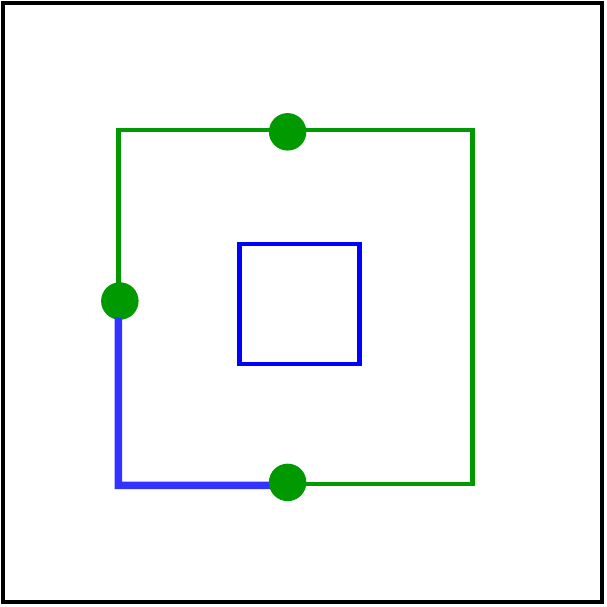}{0.18}  & $2(\Delta-\delta)$ 
 
 \\ \hline

 \caption{\textbf{Updating $J'$}. 
  Outer, relay, and nucleus pins are shown as filled black, green, and blue circles, respectively. The unfilled blue circles represent nucleus defects, which may also be nucleus pins. 
 The cases are shown up to symmetry. 
 }
 \label{cases}
\end{longtable}}

Updating $J'$ as shown, with the new paths chosen to be shortest paths, decreases its weight by at least $2(\Delta-\delta)$. 
Note that the nucleus pins shown must exist because the nucleus is nonempty in these cases.  
Repeat this process until no such violations remain, and  
let $J''$ 
 denote the resulting set of paths.

\medskip
\noindent	
\emph{Step 2. Eliminate violations of condition (1A).} 
If, after resolving violations of relaxed  conditions~(1B) or~(1C), condition~(1A) is violated, then one of the following must occur:
\begin{enumerate}
\item[(iii)] a relay pin is connected to an outer pin, another relay pin, and a defect in the nucleus of the same gadget;
\item[(iv)] a relay pin is connected to two relay pins in  the same gadget;
\item[(v)] the set of defects connected to a relay pin contains a nonempty  even-cardinality subset $S$ of defects in the nucleus of the same gadget; or 
\item[(vi)] 
two non-nucleus defects are connected by at least three paths, with the number of such paths being odd. 
\end{enumerate} 
Note that while $J'$ is initially acyclic, since it is the reduced primal component of a join and joins are acyclic by definition, the update process may create multiple paths connecting the same pair of defects, which is why we need to handle case (vi).

In cases~(iii) and~(iv), update $J''$ as shown in the second and third entries and the last two  entries, respectively, of the second column of Table~\ref{cases}, with the new paths chosen to be shortest paths. 
 In case~(v), update $J''$ by removing the paths to the defects in $S$. 
In case~(vi), update $J''$ by retaining only one of the paths. 
The weight of $J''$ decreases by at least  $2(\Delta-\delta)$ in all cases. 
Repeat this update process until no such violations remain.

\medskip
\noindent
\emph{Step 3. Restore the path-length bounds.} 
At this stage, the only possible violations  of the matching conditions are  violations of the path-length bounds in condition~(1C). 
Such
violations are eliminated by replacing each offending path in $J''$
with a shortest path. If the offending path joins a relay pin to an
outer pin or a nucleus defect, the shortest path has length at most
$\Delta$, so the weight of $J''$ decreases by at least   
  $7\Delta/5-\Delta = 2\Delta/5$. 
  If the offending path joins two relay pins, the shortest path has
length at most $2(\Delta+\delta)$, so the weight of $J''$ decreases by
at least  
  $12\Delta/5-2(\Delta+\delta) = 2\Delta/5-2\delta$.  
Thus, in either case, the weight of $J''$ decreases by at least
$2\Delta/5-2\delta$.

Note that each update preserves properties~(J1) and~(J2), since it is the XOR with either a cycle or a path whose endpoints lie in a nucleus.

In the worst case, the weight of $J''$ is smaller than that of $J'$ by at least $2\Delta/5-2\delta$, since $J'$ has at least one violation. 
The lemma now follows since, by~(J1), $J''$ is a relay cover and hence,
by Lemma~\ref{localrelaycover}, its weight is at least $R$.

\subsection{Proof of Lemma \ref{canlemAL}}\label{canlemAPP}

Let $J$ be a join containing a primal path of length at least
$3(\Delta-\delta)$ whose endpoints are not both defects in the same nucleus.

 As in the proof of Lemma~\ref{canlemAS}, we obtain a lower bound on the
weight of $J$ by establishing the same bound on its reduced primal component
$J'$. Since the endpoints of the long path do not belong to the same nucleus,
$J'$ contains the long path.

We begin with an overview of the proof. 
Unlike the proof of Lemma~\ref{canlemAS}, we do not update $J'$
constructively. Instead, we introduce an arbitrary locally optimal join $I$ of $G$ and work with the reduced primal component  
 $I'$ of $I$. In 
Lemma~\ref{auxgraph}, we combine $J'$ and $I'$ into an auxiliary weighted multigraph\footnote{Throughout, multigraphs may have parallel edges but no self-loops.}
$H$. Its vertices correspond to the non-nucleus defects and the contracted
nuclei, while its edges correspond to the paths in $J'$ and $I'$. The
even-degree structure of $H$ allows us to decompose it into cycles, which play
the same role as the update cycles in the proof of
Lemma~\ref{canlemAS}.
 
In Lemma~\ref{cycledecomp}, we further decompose each cycle into basic 
structures. Finally, Lemma~\ref{lemcyclesreduced} shows that each basic 
structure contributes at least as much weight from $J'$ as from $I'$, while the
structure containing the edge corresponding to the long path contributes at
least $\Delta-3\delta$ more. Since $I'$ is a relay cover, its weight is at
least $R$, and the  lower bound in 
Lemma \ref{canlemAL} on the weight of $J$ follows.

Our goal is to establish  the following.

\begin{lemma} \label{decomp}  The weight of $J'$ is at least $R+\Delta-3\delta$. 
\end{lemma} 
First,  recall from 
 the proof of  Lemma~\ref{canlemAS}, that $J'$ satisfies the following properties:
\begin{itemize}
\item[(J1)] every relay and outer pin has odd $J'$-degree;
\item[(J2)] for each gadget, the parity of the sum of the $J'$-degrees of its primal nucleus defects  equals the parity of the gadget.
\end{itemize}

To establish Lemma~\ref{decomp}, consider any locally optimal join $I$ of $G$, 
as defined in Definition~\ref{locoptj}. 
 Such a join can be obtained from any configuration of $G$ by finding, for each gadget $g$, a minimum-weight $S$-join, where $S$ is the signature assigned to $g$ by the configuration, and taking the union of these joins.

 Let $I'$ be the reduced primal component of $I$.  Then: 
 
\begin{itemize} 
\item[(I1)] every relay and outer  pin has $I'$-degree~1; namely, every outer pin is matched to an associated relay pin in one of the two gadgets containing it, while every relay pin is matched to its corresponding outer or nucleus pin in the same gadget, or to another relay pin in the same gadget; 
\item[(I2)] 
every path in $I'$ is a shortest path;

\item[(I3)] for each gadget, the parity of the sum of the $I'$-degrees of its primal nucleus defects equals the parity of the gadget. 
\end{itemize}
Note that (I2) holds because in a locally optimal join all paths incident to relay pins are shortest paths, and these are the only paths remaining in $I'$. 

We combine $J'$ and $I'$
 into an auxiliary weighted multigraph $H$ whose even-degree structure will allow us to decompose it into cycles. 
Contract the nucleus of each gadget to a single node, and also include
a contracted nucleus node for gadgets with empty nuclei (the even
garbage-collection gadget and the equality gadget). The vertices of
$H$ are the relay and outer pins of $G$ and the contracted nuclei of
all gadgets.

For each path in $J'$ connecting two defects, add to $H$ an edge between the corresponding nodes with weight equal to the length of the path. We refer to these edges as \emph{$J'$-edges}. Thus, by construction, the weight of $J'$ is the total weight of the
$J'$-edges.

Proceed similarly for the paths in $I'$, except for paths connecting two
relay pins in the same gadget. Since $I$ is locally optimal, if it
contains a path connecting two relay pins in the same gadget, then the
nucleus of that gadget must either be empty or consist of a single node.
This can occur only for the wire, garbage-collection, and equality
gadgets. Thus, the length of such a path is exactly $2\Delta$.
For each such path in a gadget $g$, replace it by two edges connecting
the relay pins to the contracted nucleus of $g$, each of weight $\Delta$.
We refer to the resulting edges as \emph{$I'$-edges}. Similarly, by
construction, the weight of $I'$ is the total weight of the $I'$-edges.

\begin{lemma}[Properties of the auxiliary graph]\label{auxgraph}
 The weighted multigraph $H$ constructed above satisfies the following properties: 
\begin{itemize} 
\item[(H1)] every node of $H$ has even degree;	  
\item[(H2)] every relay and outer pin has exactly one incident  $I'$-edge; 
\item[(H3)] every $I'$-edge is associated with a relay pin, connects that pin either to
its corresponding outer pin or to the contracted nucleus of the same gadget,
and has weight equal to the depth of the relay pin; 
\item[(H4)] the weight of every  $J'$-edge is at least the $L_1$-distance  between the endpoints of its corresponding path in $J'$;
\item[(H5)] at least one  $J'$-edge has weight at least $3(\Delta-\delta)$. 
\end{itemize}
\end{lemma}

{\em Proof.} For relay and outer pins, property~(H1) follows from~(J1) and~(I1), while for the contracted nuclei, it follows from~(J2) and~(I3).  
Properties~(H2) and~(H3) follow directly from the construction of the  $I'$-edges. Property~(H4) follows because the weight of each  $J'$-edge is the length of its corresponding path in $J'$, which is at least the $L_1$-distance between its endpoints. Finally, property~(H5) follows from the hypothesis of Lemma~\ref{canlemAL}.
\finito

Since every node of $H$ has even degree, $H$ can be decomposed into edge-disjoint simple cycles. The following lemma provides the key weight comparison on 	these cycles.

\begin{lemma}\label{lemcycles}
For every cycle $c$ in the decomposition of $H$, $w_{J'}(c) \geq w_{I'}(c)$, where $w_{I'}(c)$ and $w_{J'}(c)$ denote the total weights of the 
$I'$-edges and  $J'$-edges in $c$, respectively.

Moreover, 
 if $c$ contains an edge satisfying~(H5), then $w_{J'}(c) \geq w_{I'}(c) + \Delta - 3\delta$.
\end{lemma}

\medskip
\noindent
\emph{Proof of Lemma~\ref{decomp} using Lemma~\ref{lemcycles}.}
Sum the inequalities of Lemma~\ref{lemcycles} over all cycles. The total weight of the $J'$-edges in $H$ is exactly the weight of $J'$, while the total weight of the $I'$-edges in $H$ is that of $I'$, which is 
 at least the relay cost $R$ of $G$. The latter follows from the fact that $I'$ is a relay cover and therefore has weight at least $R$. In fact, equality holds by properties~(H2) and~(H3). Indeed, there is
exactly one $I'$-edge associated with each relay pin, and its weight is
equal to the depth of that relay pin. Hence, the total weight of the
$I'$-edges in $H$ is the relay cost $R$. 
 \finito

To establish Lemma~\ref{lemcycles}, we decompose the cycles into basic   structures and establish the required inequalities on these structures. We define these structures next.

{
\renewcommand{\galscale}{0.17}
\begin{figure}[!htbp]
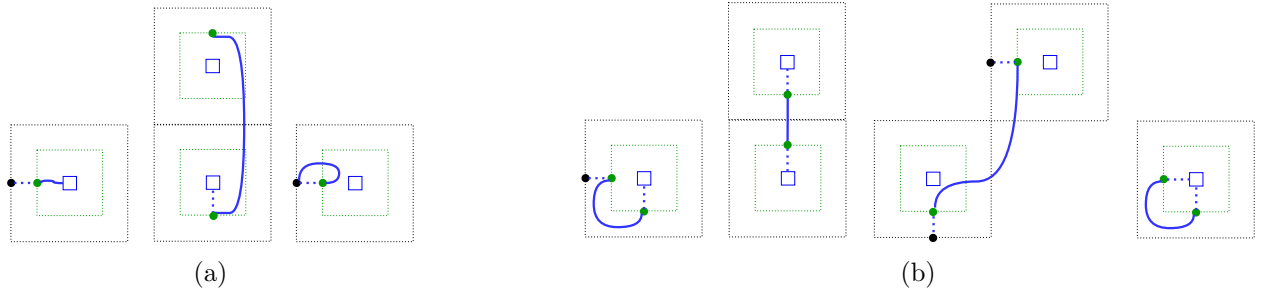
\centering
\sfg{connection.pdf}
\hfill
\sfg{double-connection.pdf}
\caption{
(a) Examples of three \textbf{connections}, where the last is a cycle; (b) examples of four \textbf{double connections}, where the last is a cycle. Thick dashed lines represent 
$I'$-edges and solid lines represent $J'$-edges. Green nodes are relay pins, nodes shown as boxes are contracted nucleus nodes, and black nodes are outer pins.}
\label{connectionsf}
\end{figure}
}

We call a simple  path or cycle in $H$ a \emph{connection} if it consists of an $I'$-edge and a $J'$-edge that share a common relay pin. 

We call a simple path or cycle in $H$ a \emph{double connection} if it
consists of three edges $(u,r_1)$, $(r_1,r_2)$, and $(r_2,v)$, where
$r_1$ and $r_2$ are distinct relay pins, $(u,r_1)$ and $(r_2,v)$ are
$I'$-edges, and $(r_1,r_2)$ is a $J'$-edge. The cycle case corresponds
to $u=v$. Thus, by~(H2) and (H3), neither $u$ nor $v$ is a relay pin.

See Figure~\ref{connectionsf} for examples of connections and double connections.

\begin{lemma}[Decomposing cycles]\label{cycledecomp}

Every cycle in the decomposition of $H$ is an edge-disjoint union of connections, double connections, and $J'$-edges. 
\end{lemma}
{\em Proof.}
Consider any cycle $c$ in the decomposition of $H$.
It follows from~(H2) and~(H3) that every $I'$-edge in $c$ connects
a relay pin to a node that is not a relay pin, and that no two
$I'$-edges in $c$ are incident to the same relay pin. 
This is the only property we need to establish the decomposition.

Let $X$ be the set of relay pins in $c$ incident to $I'$-edges,
$A$ the set of $I'$-edges in $c$, and $B$ the set of $J'$-edges in $c$. Thus, each relay pin in $X$ is incident to a unique edge in $A$, and
there is a one-to-one correspondence between the relay pins in $X$ and
the  edges in $A$.
Since $c$ is a \emph{cycle}, each relay pin  in $X$ is incident to \emph{exactly two edges} in 
$c$.
One of these edges is in $A$ and the other in $B$.

If two distinct relay pins $r_1,r_2$ in $X$ are adjacent in $c$, then
they must be connected by an edge in $B$.
For every such pair $\{r_1,r_2\}$, consider the double connection consisting
of the edge in $A$ incident to $r_1$, the edge $(r_1,r_2)$ in $B$, and the edge
in $A$ incident to $r_2$.

For every relay pin $r$ in $X$ that is not connected by an edge in $B$ to another relay pin in $X$, consider the connection consisting of the two edges in $A$ and $B$ incident to $r$.

This covers all edges in $A$, and all constructed connections and
double connections are pairwise edge-disjoint, since there is a one-to-one
correspondence between the relay pins in $X$ and the  edges in
$A$, and any two relay pins in $X$ incident to the same edge in $B$ were
grouped into a double connection.

Finally, consider as a separate $J'$-edge each edge in $B$ not covered
by the constructed connections and double connections. These pieces
therefore form an edge-disjoint decomposition of $c$.
\finito

Lemma~\ref{lemcycles} then follows from Lemma~\ref{cycledecomp} and
Lemma~\ref{lemcyclesreduced} below, which establishes the required weight
comparisons for connections, double connections, and $J'$-edges.

\begin{lemma}\label{lemcyclesreduced}
If $c$ is a connection, double  connection, or a $J'$-edge in $H$, then  $w_{J'}(c) \geq w_{I'}(c)$.

Moreover, if $c$ contains an  edge satisfying~(H5), then  $w_{J'}(c) \geq w_{I'}(c)+\Delta-3\delta$.
\end{lemma}
{\em Proof.} 
We start with the connection case. A key observation is that for any adjacent $I'$-edge and $J'$-edge  incident to a common relay pin $r$, the weight of the $I'$-edge is at most that of the $J'$-edge. This follows from~(H3) and~(H4), 
together with the fact that the depth of $r$ is the minimum distance from $r$ to any other defect.

Let $c$ be a connection, with $I'$-edge $e$ and $J'$-edge $f$, which are incident to a common relay pin.  Then
$w_{I'}(c)=w(e)$ and $w_{J'}(c)=w(f)$, 
where $w$ denotes the weight in $H$.  By the above observation, $w(e)\leq w(f)$, and hence $w_{J'}(c)\geq w_{I'}(c)$. Moreover, if $c$ contains an  edge satisfying~(H5), then
$w_{J'}(c)\geq 3(\Delta-\delta)\geq w_{I'}(c)+2\Delta-3\delta$,
since $w_{I'}(c)=w(e)\leq\Delta$.

For the double  connection, the key observation is that if $r_1$ and $r_2$ are distinct relay pins, $e_1$ and $e_2$ are the  $I'$-edges incident to $r_1$ and $r_2$, respectively, and $f$ is a  $J'$-edge 
 connecting $r_1$ and $r_2$, then $w(e_1)+w(e_2)\leq w(f)$.

To see this, note that by~(H3),
$w(e_1)$ and $w(e_2)$ are the depths of $r_1$ and $r_2$, respectively.  
 Thus, if we place an $L_1$ ball of radius $w(e_1)$ centered at $r_1$, and similarly an $L_1$ ball of radius $w(e_2)$ centered at $r_2$, then the two diamonds are either disjoint or intersect only along their boundaries. Consequently,
$w(e_1)+w(e_2)\leq d(r_1,r_2)\leq w(f)$,
where $d(r_1,r_2)$ denotes the $L_1$ distance between $r_1$ and $r_2$, and the second inequality follows from~(H4).

Consider a double connection $c$. Let $e_1$ and $e_2$ be its two  $I'$-edges, let $r_1$ and $r_2$ be their respective relay endpoints, and let $f$ be the 
$J'$-edge  connecting $r_1$ and $r_2$. 
Then $w_{I'}(c)=w(e_1)+w(e_2)$ and  $w_{J'}(c)=w(f)$. 
Therefore, by the key observation, 
$w_{J'}(c)\geq w_{I'}(c)$.

Moreover, if $c$ contains  an  edge satisfying~(H5), then
$w_{J'}(c)\geq 3(\Delta-\delta)
\geq w_{I'}(c)+\Delta-3\delta$, 
since $w_{I'}(c)=w(e_1)+w(e_2)\leq 2\Delta$.

Finally, the $J'$-edge  case is trivial since $w_{I'}(c)=0$. In particular,   $w_{J'}(c) \geq w_{I'}(c)$. Moreover, if $c$ is the   edge satisfying~(H5), then $w_{J'}(c)\geq 3(\Delta-\delta)= w_{I'}(c)+3(\Delta-\delta)$. 
\finito

\subsection{Proof of Lemma~\ref{localred}}\label{canlemBP}
In this section, we establish Lemma~\ref{localred}, which is restated below for convenience.

\medskip
\medskip
\noindent 
\textbf{Lemma~\ref{localred} 
(Local coupling)}~
{\em 
Consider any gadget $g$. For each relay pin of $g$, consider a path incident to that pin and satisfying condition~(1C) in the definition of a canonical join, i.e., the path matches 
the  relay pin  to exactly one of the following:
\begin{itemize}
\item[$\bullet$] its corresponding outer pin by a path of length at most $7\Delta/5$;
\item[$\bullet$] a defect in the nucleus of $g$ by a path of length at most $7\Delta/5$; or
\item[$\bullet$] another relay pin of $g$ by a path of length at most $12\Delta/5$.
\end{itemize}
 Let $q$ be a dual path connecting a dual-lattice node in the nucleus box to a node outside the red box. Then $q$ contains at least $2\Delta/5-3\delta-1$ edges whose dual edges lie in the red box and are not contained in any of the paths incident to the relay pins.
}
\medskip
\medskip

The key idea is that high coupling between a primal and a dual path requires both paths to move diagonally over a substantial distance. Such movement is limited when the primal path connects two horizontally or vertically aligned primal lattice points and has bounded length.  This intuition is formalized in Lemma~\ref{singlecoupling}.

To handle the coupling between $q$ and multiple primal paths, we restrict attention to the complement of the relay box within the red box, shown in orange in Figure~\ref{ringcases}(a). 
 We call this region the \emph{ring} of the gadget.  Working in the ring, rather than in the entire red box, simplifies the analysis of the interaction between the dual path and the multiple paths incident to the relay pins, as shown in Lemma~\ref{multicoupling}. The reason is that, within the ring, the minimum distance between such paths is large,  
yielding a lower bound on the number of edges of the dual path that are not dual to any primal edge.

Throughout the remainder of this section, let $q'$ denote a subpath of $q$ with exactly one node in the relay box and exactly  one node outside the red box. Thus, $q'$ crosses from the relay box into the ring exactly once and from the ring to the exterior of the red box exactly once. We work with $q'$ rather than with the entire path $q$.

{
\renewcommand{\galscale}{0.4}
\begin{figure}[!htbp]
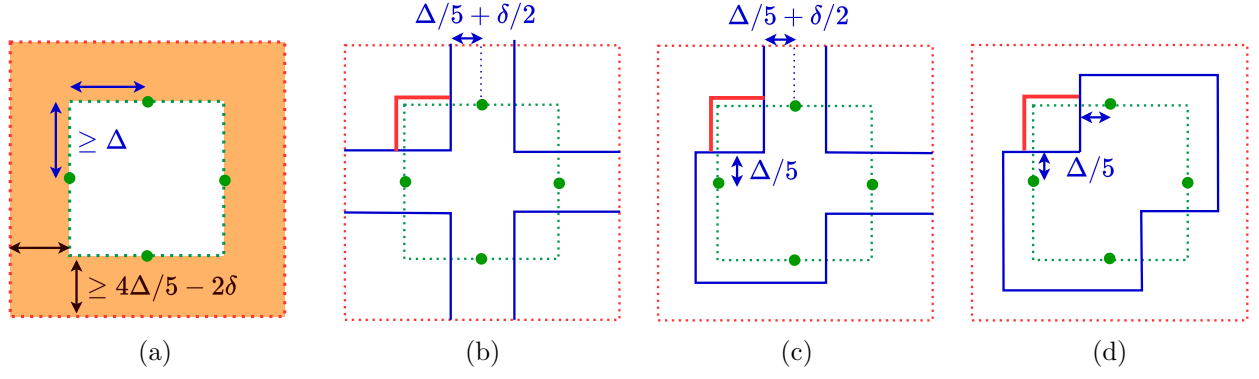
\centering
\sfg{ring.pdf}\hfill
\sfg{localCase1.pdf}\hfill
\sfg{localCase2.pdf}\hfill
\sfg{localCase2.5.pdf}
\caption{(a) \textbf{The ring.} The ring, shown in orange, is the complement of the relay box within the red box. The width of the ring on the side of the  relay pin $r$ is $\operatorname{depth}(r)-(\Delta/5+\delta)\geq 4\Delta/5-2\delta$, since $\operatorname{depth}(r)\geq\Delta-\delta$.   
The blue regions in parts (b)--(d) contain the portions of the pin paths lying in the red box in \textbf{Cases~1--3 in the proof of  Lemma~\ref{multicoupling}}, respectively. A path 
connecting a relay pin to its corresponding outer pin or to a nucleus defect 
 can deviate by at most $\Delta/5+\delta/2$ from a shortest path between the two pins, since its length is at most $7\Delta/5$ while their distance is at least $\Delta-\delta$. Similarly, a path connecting two  relay pins can deviate by at most $\Delta/5$ from a shortest path, since its length is at most $12\Delta/5$ while the relay pins are at distance at least $2\Delta$. 
 }
\label{ringcases}
\end{figure}
}

\begin{lemma}[Coupling with multiple primal paths]
\label{multicoupling}
If $q'$
contains two edges whose dual edges lie in the ring and are contained in paths incident to two distinct relay pins, then the number of 
edges in $q'$ whose dual edges lie in the ring and are not contained in any of the paths incident to the relay pins is at least 
 $8\Delta/5 -
 \delta-1$.
\end{lemma}
{\em Proof.}
We consider three cases, depending on
the number of adjacent  relay pins that are connected.

\emph{Case 1.} Suppose that no two adjacent relay pins are connected. 
Then the portions of the corresponding pin paths contained in the red box lie in the blue cross shown in Figure~\ref{ringcases}(b).
 Since $q'$ is connected, it must traverse a distance of at least $d-1$ without coupling to any of the pin paths, where $d$ is the minimum distance between two connected components in the intersection of the ring and the cross. A shortest such path is shown in red in Figure~\ref{ringcases}(b), and has length 
$d \geq 2(\Delta-
(\Delta/5+\delta/2)) = 
 8\Delta/5 - \delta$.

\emph{Case 2.} Suppose that exactly one pair of adjacent relay pins is connected. Then, within the red box, the pin paths are contained in the blue region shown, up to symmetry, in Figure~\ref{ringcases}(c). Arguing as in Case~1, we obtain the lower bound $d-1$, where
 $d \geq (2\Delta- 
(\Delta/5+\delta/2)-\Delta/5) = 
 8\Delta/5 - \delta/2$.

\emph{Case 3.} Suppose that two  pairs of adjacent relay pins are connected.
Then, within the red box, the corresponding pin paths are contained in the blue region shown, up to symmetry, in Figure~\ref{ringcases}(d). Again, by arguing as in Case~1, we obtain the lower bound $d-1$, where
 $d \geq (2\Delta- 2\Delta/5) = 
 8\Delta/5$. 
 \finito

Thus, we only have to consider the coupling between the dual path $q'$ and a single primal pin path. Lemma~\ref{singlecoupling} handles the key case in which the primal path connects a relay pin to its corresponding outer pin and $q'$ enters the ring through the side containing the relay pin and exits through the side containing the outer pin. Lemma~\ref{singlecouplingbounday} handles all remaining cases.

\begin{lemma}[Main single-path coupling]\label{singlecoupling}
Consider a primal path $p$ connecting a relay pin $r$ to its corresponding outer pin $o$, and suppose that $p$ has length at most $7\Delta/5$. If $q'$ enters the ring through the side containing $r$ and exits through the side containing $o$, then the number 
$Q'$ of edges of $q'$ whose dual edges lie in the ring and are not contained in $p$ is at least $2\Delta/5-3\delta-1$.
\end{lemma}
{\em Proof.}
Assume without loss of generality that $r$ and $o$ are horizontally aligned.

The key observation is that $p$ cannot move far in the vertical direction: it must travel horizontally from $r$ to $o$, while its total length is bounded. On the other hand, $q'$ must travel horizontally from the side of the ring containing $r$ to the exterior of the red box through the side containing $o$, and therefore contains many right-moving edges. Such edges can be dual to only  the upward- or downward-moving edges of $p$. We lower-bound $Q'$ by subtracting the total number of vertical edges of $p$ from the number of right-moving edges of $q'$ whose duals lie in the ring.

Orient $p$ from $r$ to $o$, and orient $q'$ from its endpoint in the relay box to its endpoint outside the red box. Recall that $q'$ has exactly one endpoint in the relay box and exactly one endpoint outside the red box; see Figure~\ref{maincasef}.

{
\renewcommand{\galscale}{0.24}
\begin{figure}[!htbp]
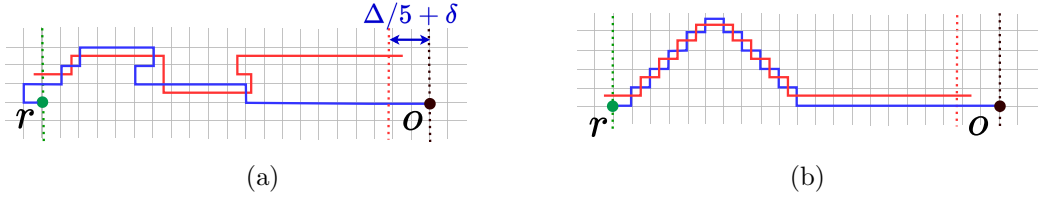
\centering
\sfg{localCase7.pdf}\hfilli\hfilli\hfilli
\sfg{localCase8.pdf}
\caption{\textbf{Single-path coupling.}
Part (a) 
illustrates  the setting in Lemma~\ref{singlecoupling}. Part (b) shows
a primal path of the same length as in (a) that achieves maximum coupling.}
\label{maincasef}
\end{figure}
}

Let $N_{\rightarrow}(p)$, $N_{\leftarrow}(p)$, $N_{\uparrow}(p)$, and $N_{\downarrow}(p)$ denote the numbers of right-, left-, upward-, and downward-moving edges of $p$, respectively. Define $N_{\rightarrow}(q')$, $N_{\leftarrow}(q')$, $N_{\uparrow}(q')$, and $N_{\downarrow}(q')$ analogously. Let $\Delta_{r,o}$ denote the distance between $r$ and $o$. Then
\begin{eqnarray}
N_{\rightarrow}(p)-N_{\leftarrow}(p) &=& \Delta_{r,o},\label{eqr1}\\
N_{\rightarrow}(q')-N_{\leftarrow}(q')&\geq&
\Delta_{r,o}-(\Delta/5+\delta),\label{eqr2}\\
\Delta_{r,o}&\geq&\Delta-\delta.\label{eqr3}
\end{eqnarray}

The length of $p$ is
$N_{\uparrow}(p)+N_{\downarrow}(p)+N_{\leftarrow}(p)+N_{\rightarrow}(p)=\|p\|_1$.
Using~\eqref{eqr1}, we obtain
\begin{equation}\label{eqr4}
N_{\uparrow}(p)+N_{\downarrow}(p)
=\|p\|_1-\Delta_{r,o}-2N_{\leftarrow}(p).
\end{equation}
By the hypothesis of the lemma,
\begin{equation}\label{eqr5}
\|p\|_1\leq 7\Delta/5.
\end{equation}

At most one right-moving edge of $q'$ has its dual edge outside the ring, namely, possibly, the last edge of $q'$. Hence, at least $N_{\rightarrow}(q')-1$ right-moving edges of $q'$ have dual edges in the ring. Therefore,

\begin{eqnarray*}
Q'
&\geq&
N_{\rightarrow}(q')-1-
\bigl(N_{\uparrow}(p)+N_{\downarrow}(p)\bigr)\\
&\geq&
\Delta_{r,o}-(\Delta/5+\delta)+N_{\leftarrow}(q')-1
-\|p\|_1+\Delta_{r,o}+2N_{\leftarrow}(p)
\qquad\mbox{(by~\eqref{eqr2} and~\eqref{eqr4})}\\
&=&
2\Delta_{r,o}-(\Delta/5+\|p\|_1)-\delta-1
+N_{\leftarrow}(q')+2N_{\leftarrow}(p)\\
&\geq&
2(\Delta-\delta)-(\Delta/5+7\Delta/5)-\delta-1
\qquad\mbox{(by~\eqref{eqr3}, \eqref{eqr5}, and
$N_{\leftarrow}(p),N_{\leftarrow}(q')\geq0$)}\\
&=&
2\Delta/5-3\delta-1.
\end{eqnarray*}
\finito

\begin{lemma}
[Remaining single-path cases]\label{singlecouplingbounday}
Consider a primal path $p$ incident to a relay pin $r$ and satisfying condition~(1C) in the definition of a canonical join. If $p$ and $q'$ do not satisfy the hypotheses of Lemma~\ref{singlecoupling}, then the number of edges of $q'$ whose dual edges lie in the ring and are not contained in $p$ is at least $3\Delta/5-5\delta/2-1$.
\end{lemma}
{\em Proof.}
First, suppose that no edge of $q'$ whose dual edge lies in the ring is
contained in $p$. Since $q'$ crosses the ring from the relay box to the
exterior of the red box, it contains at least
$4\Delta/5-2\delta-1$ such edges, as the width of the ring is at least
$4\Delta/5-2\delta$. Hence, the desired bound follows.
We may therefore assume that $p$ and $q'$ couple within the ring.

We argue that the failure of the hypotheses of Lemma~\ref{singlecoupling} implies that $q'$ must travel a substantial distance without coupling to $p$, either upon entering the ring or upon leaving it.

Assume first that $p$ connects $r$ to its corresponding outer pin. Then the portion of $p$ contained in the red box is contained in the blue rectangle shown, up to symmetry, in Figure~\ref{ring2f}(a).

{
\renewcommand{\galscale}{0.4}
\begin{figure}[!htbp]
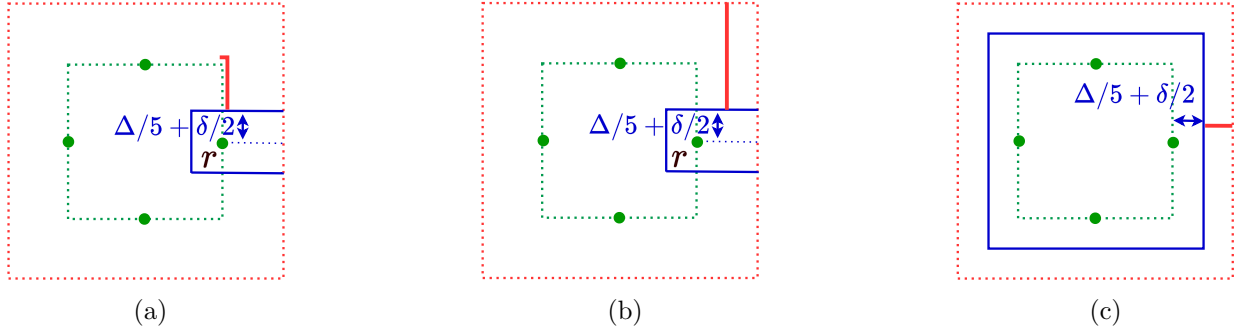
\centering
\sfg{localCase6.pdf}\hfill 
\sfg{localCase4.pdf}\hfill
\sfg{localCase5.pdf}
\caption{\textbf{Cases in the proof of Lemma~\ref{singlecouplingbounday}.}} 
\label{ring2f}
\end{figure}
}

If $q'$ does not enter the ring through the side containing $r$, then it must traverse a distance of at least $d-1$ without coupling to $p$, where $d$ is the minimum distance from the intersection of the blue rectangle  with the ring to a node on a side of the ring not containing $r$. A shortest such path is shown in red in Figure~\ref{ring2f}(a), and has length
$d\ge \Delta-(\Delta/5+\delta/2)=4\Delta/5-\delta/2$.

If $q'$ does not exit the ring through the side containing $o$, then it must traverse the ring without coupling to $p$, as shown in Figure~\ref{ring2f}(b). This gives the bound $d-1$, where
$d\ge \Delta+(4\Delta/5-2\delta)-(\Delta/5+\delta/2)=8\Delta/5-5\delta/2$.

Finally, suppose that $p$ does not connect $r$ to its corresponding outer pin. Then $p$ is contained in the blue rectangle shown, up to symmetry, in Figure~\ref{ring2f}(c). Arguing as above, we obtain the bound $d-1$, where
$d\ge (4\Delta/5-2\delta)-(\Delta/5+\delta/2)=3\Delta/5-5\delta/2$.
\finito

\subsubsection{Completing the proof of   Lemma~\ref{localred}}

In this section, we prove  Lemma~\ref{localred} using
Lemmas~\ref{multicoupling}, \ref{singlecoupling}, and~\ref{singlecouplingbounday}.

Let $q$ be a dual path connecting a dual-lattice node in the nucleus box to a node outside the red box, and let  $q'$ be a subpath of $q$ with exactly one node in the relay box and exactly  one node outside the red box.  Let $Q$ denote the number of edges of $q$ whose dual edges lie in the red box and are \emph{not contained} in any path incident to a relay pin.
Similarly, let $Q'$ denote the number of edges of $q'$ whose dual edges lie in the ring and are \emph{not contained} in any path incident to a relay pin. Thus, $Q\geq Q'$. 
For each relay pin $r$ of the gadget, let $C_r$ denote the number of edges of $q'$ whose dual edges lie in the ring and are \emph{contained} in the path incident to $r$.

We distinguish the following three cases: 
\begin{itemize}
\item[$\bullet$] \emph{Case~1.} There exist two distinct relay pins $r_1\neq r_2$ such that $C_{r_1}\neq 0$ and $C_{r_2}\neq 0$.

\item[$\bullet$] \emph{Case~2.} 
 There is a unique relay pin $r$ for which $C_r\neq 0$,
and $r$ is connected to its corresponding outer pin $o$ by a path
of length at most $7\Delta/5$. 
Moreover, $q'$ enters the ring through the side containing $r$ and exits through the side containing $o$.

\item[$\bullet$] \emph{Case~3.} All remaining cases.
\end{itemize}
In Case~1, Lemma~\ref{multicoupling} implies that $Q'\geq 8\Delta/5 -
 \delta-1$.  
In Case~2, Lemma~\ref{singlecoupling} implies that
 $Q'\geq 2\Delta/5-3\delta-1$. 
In Case~3, Lemma~\ref{singlecouplingbounday} implies that
$Q' \geq 3\Delta/5-5\delta/2-1$.  
Since $Q\ge Q'$, the lemma follows.

\end{document}